\documentclass[11pt]{article}

\usepackage{microtype}

\usepackage{verbatim}
\usepackage{amsmath}
\usepackage{amssymb}
\usepackage{amsthm}
\usepackage{amsthm}
\usepackage{amssymb}
\usepackage{amsmath}
\usepackage{wasysym}
\usepackage{amsmath}
\usepackage{subfig}
\usepackage{graphicx}
\usepackage{caption}
\usepackage{epstopdf}
\usepackage{enumerate}
\usepackage{color}
\usepackage{floatflt}
\usepackage{tabulary}
\usepackage{booktabs}

\usepackage[margin=1in]{geometry}

    \newtheorem{theorem}{Theorem}
    \newtheorem{assumption}{Assumption}
    \newtheorem{lemma}{Lemma}
    \newtheorem{definition}{Definition}
    \newtheorem{proposition}{Proposition}
    \newtheorem{remark}{Remark}
    \newtheorem{corollary}{Corollary}
    \newtheorem{property}{Property}

\newcommand{\Sites}{S}
\newcommand{\Vor}{\text{Vor}}

\newcommand{\D}[2]{D\left(#1 \parallel #2\right)}
\newcommand{\CHS}{{conv}\left\{\Sites\right\}}
\newcommand{\CH}[1]{\text{conv} #1 }
\title{On the Embeddability of Delaunay Triangulations in Anisotropic, Normed, and Bregman Spaces}

\author{Guillermo D. Canas\\Massachusetts Institute of Technology\\guilledc@mit.edu
\and Steven J. Gortler\\Harvard University\\sjg@seas.harvard.edu}

\date{}

\begin{document}
\newcounter{foo}
\maketitle

\begin{abstract}

	
Given a two-dimensional space endowed with a divergence function that is convex in the first argument, 
	continuously differentiable in the second, 
	and satisfies suitable regularity conditions at Voronoi vertices, 
we show that orphan-freedom (the absence of disconnected Voronoi regions) is sufficient 
	to ensure that Voronoi edges and vertices are also connected, and 
	that the dual is a simple planar graph. 
We then prove that the straight-edge dual of an orphan-free Voronoi diagram 
		(with sites as the first argument of the divergence) is always an embedded triangulation. 

Among the divergences covered by our proofs are Bregman divergences, 
	anisotropic divergences, 
	as well as all distances derived from strictly convex $\mathcal{C}^1$  norms 
	(including the $L_p$ norms with $1< p < \infty$). 
While Bregman diagrams of the {first kind} are simply affine diagrams, 
	and their duals ({weighted} Delaunay triangulations) are always embedded, 
	we show that duals of orphan-free Bregman diagrams of the \emph{second kind} are always embedded. 	
	

\end{abstract}

\newpage

\section{Introduction}

Voronoi diagrams and their dual Delaunay triangulations are fundamental constructions with
numerous associated guarantees, and extensive application in
practice (for a thorough review consult~\cite{Aurenhammer13} and references therein).
At their heart is the use of a distance between points, which in the original
version is taken to be Euclidean. 
This suggests that, by considering 
distances other than Euclidean, 
it may be possible to obtain variants which can be well-suited to
a wider range of applications.  

Attempts in this direction have been met with some success. 
Power diagrams~\cite{power} generalize Euclidean distance by associating 
a {bias-term}
to each 
site. The duals of these diagrams
are  
guaranteed to be embedded triangulations, in any number of
dimensions. 
Although this is a strict generalization of Euclidean distance, it is a somewhat 
limited one. The effect of the bias term is to locally enlarge or shrink the
region associated to each site, loosely-speaking ``equally in every
direction". It allows some freedom in choosing local scale, with no
preference for specific directions. 

%
%
%
%

Two related, and relatively recent generalizations of Voronoi diagrams and Delaunay triangulations have been proposed, 
independently, by Labelle and Shewchuk~\cite{LS}, and Du and Wang~\cite{DW}. 
Although their associated anisotropic Voronoi diagrams are, in general, no longer orphan-free
(i.e.~they may have disconnected Voronoi regions), 
Labelle and Shewchuk show that a set of sites exists 
with an orphan-free diagram, whose dual is embedded, in two dimensions. 
They accomplish this by proposing an iterative site-insertion algorithm
that, for any given metric, constructs one such set of sites. 
Note that this is a property of the output of the algorithm, and not a general condition for obtaining embedded triangulations.

The recent work of~\cite{Bregman} discusses Voronoi diagrams and their duals with respect to Bregman divergences. 
They show that Bregman Voronoi diagrams of the \emph{first kind} are simply power diagrams, 
	whose duals are known to always be embedded~\cite{powerdiag}. 
Bregman diagrams of the \emph{second kind} are power diagrams in the dual (gradient) space, 
	but, prior to this work, no results for them were available in the primal space.

In this paper we discuss properties of Voronoi diagrams and Delaunay triangulations for 
	a general class of divergences, 
	including Bregman, quadratic, and all distances derived from strictly convex $\mathcal{C}^1$ norms. 
We show that, 
given a divergence $D$ that is convex in the first argument and continuously differentiable in the second, 
	and under a \emph{bounded anisotropy} assumption on the divergence, 
 if a set of sites produces an orphan-free 
Voronoi diagram with respect to $D$,
then its dual is always an embedded triangulation 
	(or an embedded polygonal mesh with convex faces in general), 
in two dimensions (theorem~\ref{th:main}). 
This effectively states that, regardless of the sites' positions, if the primal 
is well-behaved, then the dual is also well-behaved. 
Further, in a way that parallels the ordinary Delaunay case, the dual has no
degenerate elements (proposition~\ref{prop:ECB}), its elements (vertices, edges, faces) 
are unique (Cor.~\ref{cor:VorI}), and the dual is guaranteed to cover the convex hull of the sites (theorem~\ref{th:main}).



\section{Voronoi diagrams 
		with respect to divergences}\label{sec:setup}

The class of divergences that we consider in this work are non-negative functions 
	$D:\mathbb{R}^2\times\mathbb{R}^2\rightarrow\mathbb{R}$ 
	which are 	strictly convex in the first argument and continuously differentiable in the second, 
	and such that $\D{x}{x}=0$ for all $x\in\mathbb{R}^2$. 
Following~\cite{Bregman}, we let
\begin{equation}\label{eq:defball}
B_1(p; \rho) = \{v\in\mathbb{R}^2 : \D{v}{p} \le \rho\}, \quad\quad B_2(p;\rho) = \{v\in\mathbb{R}^2 : \D{p}{v} \le \rho\} 
\end{equation}
be, respectively, balls of the \emph{first} and \emph{second kind}, centered at $p\in\mathbb{R}^2$ of radius $\rho$. 
Note that balls of the first kind are necessarily convex since $\D{\cdot}{p}$ is convex. 
We also assume that $D$ satisfies what we term a \emph{bounded anisotropy} condition, 
	defined in assumption~\ref{ass:BAA} below. 

Given a set $\Sites=\{s_1,\dots,s_n\}\subset\mathbb{R}^2$ of 
$n$ distinct sites on the plane, and a divergence $D:\mathbb{R}^2\times\mathbb{R}^2\rightarrow\mathbb{R}$, 
the Voronoi regions 
of the first and second kinds~\cite{Bregman} are:  
\begin{eqnarray}
\label{eq:defvor1}   {\Vor^1_i} = \{p\in\mathbb{R}^2 : \D{p}{s_i} \le \D{p}{s}, \forall s\in \Sites \}, \\
\label{eq:defvor2}   {\Vor^2_i} = \{p\in\mathbb{R}^2 : \D{s_i}{p} \le \D{s}{p}, \forall s\in \Sites \}, 
\end{eqnarray}
respectively, and are indexed by the site its points are closest to.  
Of course, the two kinds of Voronoi diagrams are different because $D$ is in general not symmetric. 
In the sequel, and whenever not otherwise specified, we will assume that balls are of the \emph{first kind} (convex), 
	and Voronoi diagrams, and their dual Delaunay triangulations are of the \emph{second kind}. 
For instance, we will use the convexity of balls (of the first kind) to prove 
	that every face in a Delaunay triangulation (of the second kind) 
	satisfies an \emph{Empty Circum-Ball} property (proposition~\ref{prop:ECB}) 
	that parallels the empty circumcircle property of Euclidean Delaunay triangulations.

Consider the following definition of Voronoi element:
\begin{definition}\label{def:VorI}
	For each subset $I\subseteq\{1,\dots,n\}$, the set $\Vor_I=\cap_{i\in I}\Vor^2_i \setminus \cup_{j\not\in I}\Vor^2_j$ is a Voronoi element of order $|I|$. 
Elements of orders $1$, $2$, and $|I|\ge 3$ are denoted regions, edges, and vertices, respectively. 
\end{definition}
\begin{remark}
The set of all Voronoi elements $\Vor_I$ forms a partition of the plane. 
\end{remark}

\begin{table}[htbp]\caption{Notation}\label{table:notation}
\begin{center}
\begin{tabular}{r c p{11cm} }
\toprule
$\D{\cdot}{\cdot}$ && A non-negative divergence strictly convex in its first argument and continuously differentiable in the second. \\
$D_F(\cdot \parallel \cdot)$ && Bregman divergence (section~\ref{sec:DF}). \\
$D_f(\cdot \parallel \cdot)$ && Csisz\'ar divergence (section~\ref{sec:Df}).\\
$D_Q(\cdot \parallel \cdot)$ && Quadratic divergence (seciton~\ref{sec:DQ}).\\
$\gamma$ && Global lower bound on the ratio of eigenvalues 
	of metric $Q$ (quadratic divergence, lemma~\ref{lem:DQgamma}) 
	or of the Hessian of $F$ (Bregman divergence, lemma~\ref{lem:DFgamma}). \\
$\Sites = \{s_1,\dots,s_n\}$ && Set of $n$ sites. \\
$L_{ij}$ && The supporting line of sites $s_i$, $s_j$. \\
$\CHS$ && Convex hull of $\Sites$. \\
$W=\{w_i\in \Sites : i=1,\dots,m\}$ && Subset of sites on the boundary of $\CHS$, in clock-wise order. \\ 
$B(\cdot\, ;\cdot)$ && Convex ball of the first kind (equation~\eqref{eq:defball}). \\
$\theta_p(v)$ && The ball (of the first kind) $B(v;\D{p}{v})$  centered at $v$ with $p$ in its boundary. \\
$\Vor_i$ && Voronoi region of the second kind corresponding to site $s_i$ (equation~\eqref{eq:defvor2}). \\
$\Vor_I$ && Voronoi element of order $|I|=1$ (Voronoi region), $|I|=2$ (Voronoi edge), or $|I|\ge 3$ (Voronoi vertex).\\
$G=(\Sites,E,F)$ && The straight-edge dual triangulation with vertices at the sites. \\
$B$ && The edges in the topological boundary of $G$ (incident to one face). \\
$\mathcal{B}=(w_i,w_{i\oplus 1})_{i=1}^{|W|}$ && The edges in the boundary of $\CHS$. \\
$\pi$ && Projection from $C(\sigma)$ onto $\partial\CHS$ (section~\ref{sec:boundary}). \\
$\nu_\sigma$ && Projection function onto a circle of radius $\sigma$ (section~\ref{sec:boundary}). \\
$H^{+}_{ij}, H^{-}_{ij}$ && The half-spaces on either side of $L_{ij}$, chosen so $H^{+}_{ij}\cap\Sites=\phi$ (fig.~\ref{fig:pinu}). \\
$C(\sigma)$ && The origin-centered circle of radius $\sigma$ (with respect to the natural metric). \\
\bottomrule
\end{tabular}
\end{center}
\end{table}

The following ``bounded anisotropy" condition is assumed to hold. 
It is written in its most general (but very technical) form below, 
	but it becomes much simpler in particular cases, as shown in Section~\ref{sec:summary}. 
Typically, it can be rewritten as a simple regularity condition on a symmetric positive definite matrix, 
	such that its ratio of minimum to maximum eigenvalues
	(a measure of anisotropy) is globally bounded away from zero.

\begin{figure}[thb]
\centering
\includegraphics[width=2.1in]{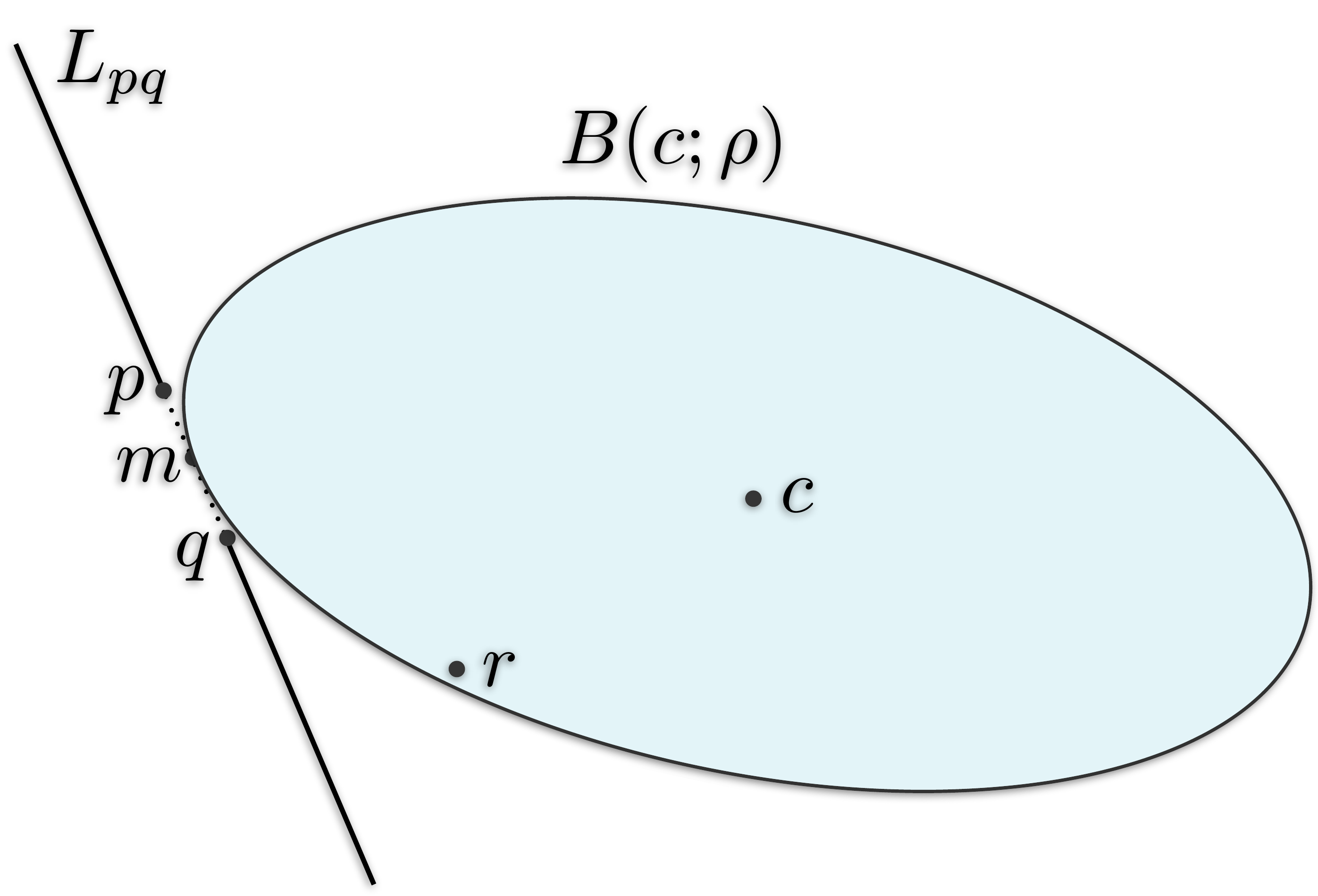}
\caption{The \emph{bounded anisotropy asumption} ensures that balls of the first kind are globally well-behaved. }
\label{fig:gamma}
\end{figure}

\begin{assumption}[Bounded anisotropy]\label{ass:BAA}
For every two points $p,q\in\mathbb{R}^2$ with supporting line $L_{pq}$, and every point $r\not\in L_{pq}$, 
	there is a sufficiently large value $\mu > 0$ such that 
	for every point $c\in\mathbb{R}^2$ lying on the same side of $L_{pq}$ as $r$, 
		such that $\|c\| > \mu$, and 
		whose closest point $m$ in $L_{pq}$ lies in the segment $\overline{pq}$, 
	it is $\D{r}{c} < \D{m}{c}$. 
\end{assumption}
\begin{remark}
Note that the condition $\|c\|<\mu$ depends on the (arbitrary) choice of origin. 
Assumption~\ref{ass:BAA} is, however, independent of this choice. 
\end{remark}
Loosely speaking, this condition ensures that balls of the first kind are 
	not just convex, but also ``sufficiently round". 
For instance, it is satisfied by all the $L_p$ distances with $1<p<\infty$, 
	but not for $p=1,\infty$, since (aside from not being strictly convex) the corresponding balls have ``kinks".

\begin{assumption}[Extremal gradients]\label{ass:EGA}
	For each Voronoi vertex $\Vor\{i_1,\dots,i_m\}$ with $m\ge 3$, 
		the gradients $g_j(p)\equiv\nabla_p \D{s_{i_j}}{p}, j=1,\dots,m$, at $p\in\Vor_I$ are distinct and extremal, 
		i.e.\ they are vertices of the convex hull:  $\CH\{g_1(p),\dots,g_m(p)\}$. 
\end{assumption}

\begin{remark}
	In the ``typical" case that $m=3$, the above simply means that $g_1,g_2,g_3$ are not colinear. 
	Given two distinct gradients $g_1\ne g_2$, requiring $g_3$ not to be colinear only constraints it 
		to be outside a line.  
	If $D$ is the $L_p$ distance  (or any other non-spatially-varying divergence), 
	the extremal gradient assumption can be shown to be always automatically satisfied at Voronoi vertices. 
	Finally, the extremal gradient assumption will be shown to imply that Voronoi vertices are composed of isolated points, 
		and therefore, when satisfied, the assumption only needs to be enforced at a discrete set of points. 
\end{remark}

\subsection{Orphan-free Voronoi diagrams and dual triangulations}\label{sec:simpleplanar}

As described in the classic survey by Aurenhammer~\cite{Aurenhammer:1991}, 
	planar Voronoi diagrams and Delaunay triangulations are duals in a graph theoretical sense. 
Associated to the ordinary Voronoi diagram is a simple, planar (primal) graph with vertices at 
	points equidistant to three or more sites (Voronoi vertices), 
	and edges composed of line segments 
	equidistant to two sites (Voronoi edges). 
Because edges are always line segments, the graph is simple (has no multi-edges or self-loops), 
	and this construction provides an embedding of the graph, which must therefore be planar.

For Voronoi diagrams defined by divergences, the situation is markedly different. 
The incidence relations between Voronoi elements cannot be so easily established. 
For instance, Voronoi edges may be disconnected and incident to any number of Voronoi vertices. 
For this reason, we begin our proof by constructing an embedding of a primal graph 
	from the incidence relations of the Voronoi diagram (definition~\ref{def:incidence}), 
	in a way that generalizes ordinary Voronoi diagrams, 
	and show that this graph is simple and planar (section~\ref{sec:planar}).
This primal graph is then dualized into a simple, planar graph. 
The dual graph is denoted the Delaunay \emph{triangulation} because, as will be shown, 
	it is composed of convex faces which can be triangulated without breaking any of 
	its important properties, such as embeddability or the \emph{empty circum-ball} property (property~\ref{cor:VorI}).

The rest of the paper makes heavy use of the following trivial lemmas, which we include here for convenience. 
The first follows directly from the properties of $D$, while the second
	is a direct consequence of the strict convexity of $\D{\cdot}{p}$ and the continuity of $D$
	(note that $D$ is globally continuous since it is continuous in the second argument 
	and convex in the first, and therefore it is also continuous in the first argument~\cite{rockafellar1997convex}).

\begin{lemma}\label{lem:site}
Every site $s_i\in\Sites$ is an interior point of its corresponding Voronoi region $\Vor_i$. 
\end{lemma}

\begin{lemma}\label{lem:midpoint}
Given two sites $s_i,s_j\in \Sites$ with supporting line $L_{ij}$, 
	all points $p\in L_{ij}$ that are equidistant to $s_i$ and $s_j$ 
	belong to the segment $\overline{s_i s_j}$. 
Furthermore, there is always at least one such point. 
\end{lemma}

\section{Summary of results}\label{sec:summary}

\begin{figure}[ht]
\begin{center}
\subfloat{\includegraphics[height=3.5cm]{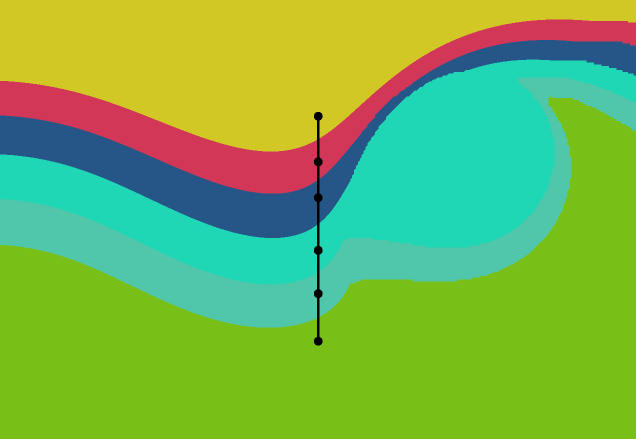}}\quad
\subfloat{\includegraphics[height=3.5cm]{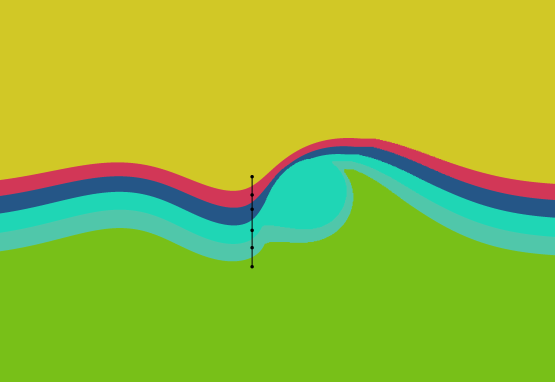}}
\caption{If all sites are colinear, the dual is always a chain connecting consecutive sites along their supporting line.
This structure is independent of the divergence, and doesn't require assumption~\ref{ass:BAA}. }
\label{fig:colinear}
\end{center}
\end{figure}

Consider first the special case that all sites in $\Sites$ are colinear. 
The structure of the Voronoi diagram and the Delaunay triangulation is very simple in this case. 
If we order the sites $s_1,\dots,s_n$ sequentially along their supporting line, 
lemma~\ref{lem:midpoint} shows that there must be Delaunay edges between successive sites, 
	while the strict convexity of the balls implies that these are the only edges
	(all points in $\overline{s_i s_{i\oplus 1}}$ are strictly closer to $s_i,s_{i\oplus 1}$ than to any other site), 
	and that there are no Delaunay faces
	(since three colinear points cannot be in the boundary of a strictly convex ball). 
The following proposition does not require assumption~\ref{ass:BAA} nor~\ref{ass:EGA}. 

\begin{proposition}
{For all divergences} $D$, 
	the Delaunay triangulation of a set of colinear sites is a chain connecting successive sites $s_i,s_{i\oplus 1}$, $i=1,\dots,n-1$ 
	along their supporting line. 
\end{proposition}

With the colinear site case covered, we assume in the remainder 
	that \emph{not all} sites are colinear, 
	and that $D$ satisfies assumptions~\ref{ass:BAA} and~\ref{ass:EGA}. 

We begin, in section~\ref{sec:planar}, by constructing a primal graph from the incidence relations between Voronoi elements, 
	and dualize it to obtain a simple, planar graph. 

\begin{theorem}\label{th:simpleplanar}
The dual of the primal Voronoi graph of an orphan-free Voronoi diagram 
	is a simple, connected, planar graph. 
\end{theorem}

\begin{remark}
Note that the differentiability of $D$ with respect to the second argument is only used in (a small neighborhood around) Voronoi vertices 
	(a set of isolated points). 
Everywhere else, it suffices that $D$ is continuous in its second argument. 
\end{remark}

While this dual graph is an embedded planar graph with curved edges, we then show
	that it is also an embedded planar graph with vertices at the sites and straight edges. 

\begin{theorem}\label{th:main}
	The straight-edge dual of a primal Voronoi graph
	(obtained from an orphan-free Voronoi diagram of a set of sites $\Sites$)
	is embedded with vertices at the sites, 
		has (non-degenerate) strictly convex faces, and covers the convex hull of $\Sites$. 
\end{theorem}

As described in Section~\ref{sec:simpleplanar}, lemmas~\ref{lem:regionSC} and~\ref{lem:SCedges} 
	can be used in conjunction with theorem~\ref{th:main} to 
	conclude that orphan-freedom 
	is a sufficient condition 
	for the well-behavedeness of not just the 
	dual, but also of the primal Voronoi diagram. 
Note that this excludes isolated Voronoi edges (those not incident to any Voronoi vertex), 
	which are shown to be contained in Voronoi regions, 
	and are considered part of their containing regions (section~\ref{sec:propedges}).

\begin{corollary}\label{cor:VorI}
	All the elements of an orphan-free Voronoi diagram are connected, 
	with the exception of isolated Voronoi edges. 
\end{corollary}
\begin{remark}
	Isolated edges are connected components of a Voronoi edge 
	which are incident to a single Voronoi region. 
	Since they do not affect the construction of the primal Voronoi graph, 
	they can be safely discarded, as shown in section~\ref{sec:propedges}.
\end{remark}

Perhaps the most fundamental property of the diagrams that we use in the proofs is that 
	every dual face has an ``empty" circumscribing \emph{convex} ball. 
This empty circum-ball (ECB) property
	is analogous to the empty circumcircle property of ordinary Voronoi diagrams:
\begin{proposition}[Empty Circum-Ball property]\label{prop:ECB}
	For every dual face with vertices $s_{i_1},\dots,s_{i_k}$ there is a convex ball 
		that circumscribes $s_{i_1},\dots,s_{i_k}$ and contains no site in its interior. 
\end{proposition}

Indeed, since to every dual face $f$ with vertices $s_{i_1},\dots,s_{i_k}$ ($k\ge 3$) corresponds 
	a Voronoi element $\Vor_{\{i_1,\dots,i_k\}}$, 
	any point $c\in\Vor_{\{i_1,\dots,i_k\}}$ serves as center of an empty circumscribing ball of $f$. 
To see that this ball must be ``empty", note that 
no site $s'$ 
	can be strictly inside the circumscribing ball (certainly not $s_{i_1},\dots,s_{i_k}$, since they are in the boundary), 
	or $c$ would be closer to $s'$ than to $s_{i_1},\dots,s_{i_k}$, 
	and therefore it would not be $c\in\Vor_{\{i_1,\dots,i_k\}}$. 

Notice that, although we consider Voronoi diagrams of the second kind, 
	it is the convexity of balls of the \emph{first kind} that establishes the ECB condition. 
The ECB property is, in general, not satisfied by Delaunay triangulations of the first kind. 

\begin{figure}[htbp]
\begin{center}
\includegraphics[width=2.5in]{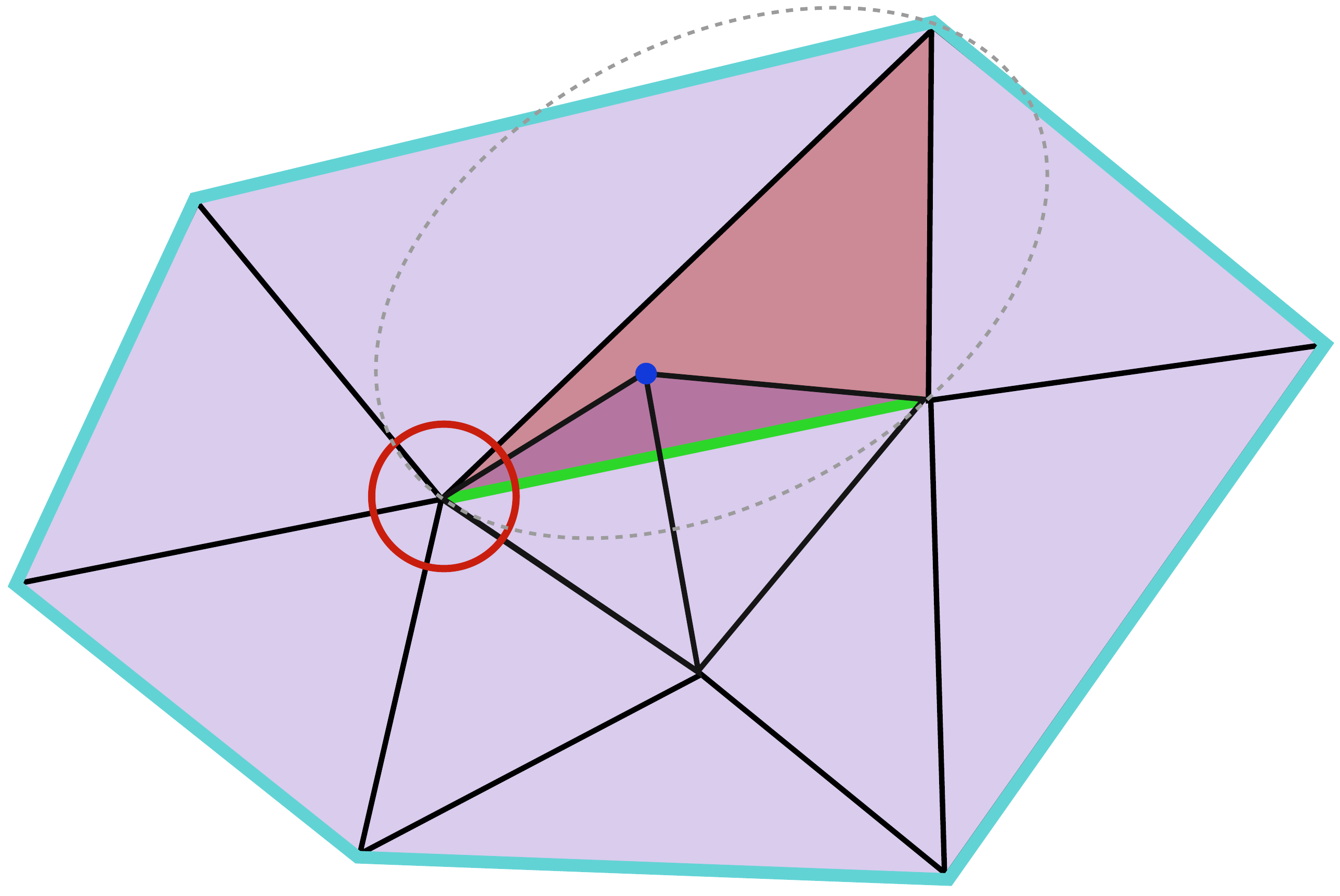}
\caption{
	We prove that the Delaunay triangulation is embedded (theorem~\ref{th:main}) 
	by showing that its boundary is simple and convex (corollary~\ref{cor:boundary}), 
		and its interior is a ``flat sheet": it has no edge fold-overs (green edge) (lemma~\ref{lem:ef}).  
	We use a discrete version of the Poincar\'e-Hopf index theorem (lemma~\ref{lem:ph}) to prove that  
		an edge fold-over would create a ``wrinkle" (circled) somewhere in the triangulation (lemma~\ref{lem:index-1}), 
		which in turn would force some vertex (blue) to ``invade" a face (red) (lemma~\ref{lem:non-negative}), 
		breaking the face's empty circum-ball (grey, dotted) condition (proposition~\ref{prop:ECB}).}
\label{fig:outline}
\end{center}
\end{figure}

After establishing that a Voronoi diagram can be associated with an embedded planar primal graph 
	which can be dualized into a planar dual graph (section~\ref{sec:planar}), 
	the rest of the paper is concerned with the proof of our main claim (theorem~\ref{th:main}), 
	whose structure is outlined in figure~\ref{fig:outline}. 
The proof of embeddability of the straight-edge dual is divided in two parts. 
In the first part (section~\ref{sec:boundary}), we use the bounded anisotropy assumption (assumption~\ref{ass:BAA})
	to show that the 
	\empty{topological} boundary of the straight-edge dual Delaunay triangulation 
	(the set of edges shared by only one face) coincides with the boundary of the convex hull of the sites, 
	and therefore is a simple, closed polygonal chain, a fact necessary for the second part of the proof to proceed. 
	Section~\ref{sec:boundary} is the more technical part of the proof; at its heart it is an application of Brouwer's fixed point theorem. 
In section~\ref{sec:interior}, we use the theory of discrete one-forms~\cite{1form} 
	to show that the Delaunay triangulation has no fold-overs (is a ``flat sheet") 
	and is therefore a single-cover of the convex hull of $\Sites$. 
Note that these two results, along with the ECB property, mirror similar properties of ordinary Delaunay triangulations.

The above results can be particularized to a number of existing divergences and metrics. 
We briefly discuss next a few of them, as well as simple conditions for assumption~\ref{ass:BAA} to hold
	for some of them (with proofs in Appendix A). 

\subsection{Bregman divergences}\label{sec:DF}
Given a strictly convex, everywhere differentiable function $F:\mathbb{R}^2\rightarrow\mathbb{R}$, 
	the Bregman divergence
\begin{equation}\label{eq:defDF}
 D_F(p \parallel q) \equiv F(p) - F(q) - \langle p-q, \nabla F(q)\rangle 
\end{equation}
	is the (non-negative) difference between $F(p)$ and the first-order Taylor approximation of $F(p)$ around $q$
	(the first order Lagrange remainder). 
Bregman divergences are widely used in statistics and 
	include the Kullback-Leibler divergence. 
By the (strict) convexity of $F$, and the definition of $D_F$ it 
	it is clear that, whenever $F$ is twice continuously differentiable, 
	$D_F$ is (strictly) convex in the first argument and continuously differentiable in the second. 

From the definition of $D_F$, it is clear that Bregman Voronoi diagrams of the first kind are composed of regions 
	\[ \Vor^1_{F,i} 
			    = \{ p\in\mathbb{R}^2 : \langle p, \nabla F(s) - \nabla F(s_i)\rangle \le 
			    						F(s_i)-\langle s_i, \nabla F(s_i)\rangle + 
									F(s) - \langle s, \nabla F(s)\rangle , \forall s\in \Sites \}, \]
which are intersections of half-spaces of the form $\{p\in\mathbb{R}^2 : \langle p, a\rangle \le b\}$.
Furthermore, 
	Bregman Voronoi diagrams of the first kind are simply power diagrams~\cite{Bregman}, 
	and thus their dual Delaunay triangulations of the first kind are always embedded~\cite{powerdiag,DMG}. 

On the other hand, Bregman diagrams of the \emph{second} kind can be shown to be affine diagrams 
	only in the dual (gradient) space~\cite{Bregman}. 
In the original space, the cells $\Vor^2_{F,i}$ are not simple intersections of half-spaces and, in general, they have curved boundaries. 
Prior to this work, no guarantees concerning Bregman Delaunay triangulations of the second kind were available. 
%

\begin{lemma}[Bounded anisotropy for Bregman divergences]\label{lem:DFgamma}
If $F\in\mathcal{C}^2$ and there is $\gamma > 0$ such that the Hessian of $F$ has ratio of eigenvalues bounded by $\lambda_{\text{min}}/\lambda_{\text{max}}\ge \gamma$,
	then assumption~\ref{ass:BAA} holds. 
\end{lemma}

\subsection{Quadratic divergences}\label{sec:DQ}

As is well known, the approximation efficiency of a 
	piecewise-linear function supported on a triangulation can be greatly improved by 
	adapting the shape and orientation of its elements to the target function~\cite{triangle,DAzevedo,DBLP:conf/imr/CanasG06}. 
An effective way to construct such anisotropic triangulations is to dualize a Voronoi diagram 
	derived from an anisotropic divergence~\cite{LS,DW}. 

By considering a $\mathcal{C}^1$ metric 
	(in coordinates: a function $Q:\mathbb{R}^2\rightarrow\mathbb{R}^{2\times 2}$ that is symmetric, positive definite), 
	we define the quadratic divergence as:
\begin{equation}\label{eq:defDQ} 
	D_Q(p \parallel q) \equiv \left[ (p-q)^t Q(q) (p-q) \right]^{1/2}, 
\end{equation}
which is clearly strictly convex in the first argument and continuously differentiable in the second. 
Voronoi diagrams and Delaunay triangulations with respect to $D_Q$, of the first and seconds kinds, have been considered in the literature. 
The diagram and the dual triangulation of the \emph{first kind} were proposed by Labelle and Shewchuk~\cite{LS}, 
	while those of the second kind were discussed by Du and Wang~\cite{DW}. 
While the work of Du and Wang does not provide theoretical guarantees, 
	that of Labelle and Shewchuk provides an algorithm that is guaranteed to output \emph{a} set of sites for which the 
	Voronoi diagram of the \emph{first kind} is orphan-free, and whose corresponding Delaunay triangulation is embedded. 


%
%

\begin{lemma}[Bounded anisotropy for quadratic divergences]\label{lem:DQgamma}
If there is $\gamma > 0$ such that $Q$ has ratio of eigenvalues bounded by $\lambda_{\text{min}}/\lambda_{\text{max}}\ge \gamma$,
	then assumption~\ref{ass:BAA} holds. 
\end{lemma}

Note that the above condition on the bounded anisotropy of $Q$ may commonly hold in practice, 
for instance if the metric is sampled on a compact domain and continuously extended to the plane by reusing sampled values only. 

In the case of quadratic divergences, there already exists sufficient conditions to generate orphan-free Voronoi diagrams. 
In particular, it has been shown that if $\sigma$ is a bound on a certain measure of variation of $Q$, 
	then any (asymmetric) $\epsilon$-net with respect to $D_Q$ that satisfies $\epsilon\sigma \le 0.098$ 
	(corresponding to a roughly $10\%$ variation of eigenvalues between Voronoi-adjacent sites)
	is guaranteed to be orphan-free~\cite{avd}.

\subsection{Normed spaces}\label{sec:Lp}

Our results also cover all normed spaces with a continuously differentiable, strictly convex norm, 
	including the $L_p$ spaces, 
	but excluding the cases $p=1$ and $p=\infty$. 
\begin{lemma}[Bounded anisotropy for normed spaces]
	\label{lem:Lpgamma}
	Distances derived from strictly convex $\mathcal{C}^1$ norms satisfy assumption~\ref{ass:BAA}.
\end{lemma}

\subsection{Csisz\'ar f-divergences}\label{sec:Df}

Given a convex real function $f$ with $f(1)=0$ and two measures $\rho,\mu$ over a probability space $\Omega$, 
	Csisz\'ar's f-divergence~\cite{CsiszarTutorial} is
\begin{equation}\label{eq:defDf}
	D_f(\rho\parallel \mu) \equiv \int_{\Omega}d\mu\, f\left(\frac{d\rho}{d\mu}\right) 
\end{equation}
where $\rho$ is absolutely continuous with respect to $\mu$, 
	and therefore has a Radon-Nikodym derivative $d\rho/d\mu$. 

If $f$ is strictly convex, then the f-divergence is strictly convex in the first argument and continuously differentiable in the second 
(in this case it is also jointly convex). 
For instance, the strictly convex function $f:x\mapsto \left(\sqrt x - 1\right)^2$ generates the Hellinger distance. 
F-divergences are functions of measures, and thus often in practice restricted to the probability simplex. 

\begin{remark}The limitation of our work to two dimensions implies that results for f-divergences 
	are limited to probability measures supported on just three atoms. 
Their applicability is thus somewhat limited, and  are only included for completeness. 
\end{remark}

\section{Primal Voronoi diagram and dual Delaunay triangulation}\label{sec:planar}


In this section we use the definition of Voronoi diagram (definition~\ref{def:VorI})
	to construct an embedded simple planar 
	graph whose incidence relations match 
	those of the Voronoi diagram. 
We then dualize this graph to obtain an embedded simple planar 
	graph 
	with vertices at the sites and curved edges. 
Section~\ref{sec:dual} will then show that the dual graph is also embedded when replacing curved edges by straight segments. 
Recall that we have assumed that not all sites are colinear (the colinear case is described in section~\ref{sec:summary}).

\subsection{Assumptions}\label{sec:assumptions}

We begin by making the following two technical assumptions. 

\vspace*{0.1in}\noindent{\bf Path-connectedness.}
Assume that all connected components of Voronoi elements are also path-connected. 
In fact, given the assumption below, as well as assumptions~\ref{ass:BAA} and~\ref{ass:EGA}, 
	we only need to further assume that connected components of Voronoi \emph{edges} are path-connected. 
Indeed, Voronoi regions are open and Voronoi vertices will be shown to be composed of isolated points, 
	and therefore their connected components are automatically path-connected~\cite[p.\ 158]{munkres2000topology}.

\vspace*{0.1in}\noindent{\bf Boundaries of Voronoi regions.}
Further assume that the boundary of bounded, simply-connected Voronoi regions are simple, closed (Jordan) curves. 
For unbounded regions $U$, we assume that they can be first mapped 
	through a continuous transformation $T:U\rightarrow U'$
	onto a bounded set $U'$, for instance through an appropriate M{\"o}bius transformation. 
Bounded simply-connected sets whose boundary is a Jordan curve 
	are those that are uniformly connected \emph{im kleinen}~\cite{Moore1918}\footnote{
	A space $M$ is uniformly connected \emph{im kleinen} if for every $\varepsilon>0$ there is $\delta_\varepsilon>0$ 
	such that for every pair of points $p,q\in M$ with $\|p-q\|_2<\delta_\varepsilon$
		there is a connected subset $V\subseteq M$ with $p,q\in V$ and $V\subseteq B_2(p;\varepsilon)$. 
	}.


\subsection{Properties of Voronoi elements}\label{sec:properties}

Before constructing an appropriate primal graph from the connectivity relations of the Voronoi diagram, 
	we first establish some relevant properties of the diagram's elements. 

We say that Voronoi element $\Vor_I$ is incident to Voronoi element $\Vor_J$ 
	(denoted $\Vor_I\rightsquigarrow\Vor_J$)
	if their closures overlap  
	and $\overline{\Vor_I}\cap\overline{\Vor_J}\subseteq\overline{\Vor_J}$. 

From this incidence relation we build a primal Voronoi graph, whose dual is the Delaunay triangulation with respect to $D$. 
Since ``planar graphs, and graphs embeddable on the sphere are one and the same''~\cite[p.\ 247]{bondy2008graph}, 
	we consider incidence relations 
	on the Riemann sphere (by stereographically projecting the plane onto $\mathbb{S}^2$), 
	where the added vertex at infinity is defined to be incident to unbounded elements on the plane. 
Geometric constructions will, however, typically be carried out on the plane for convenience. 

%

\subsubsection{Incident elements}\label{sec:incidence}

\begin{figure}[htbp]
   \centering
	\includegraphics[width=2.5in]{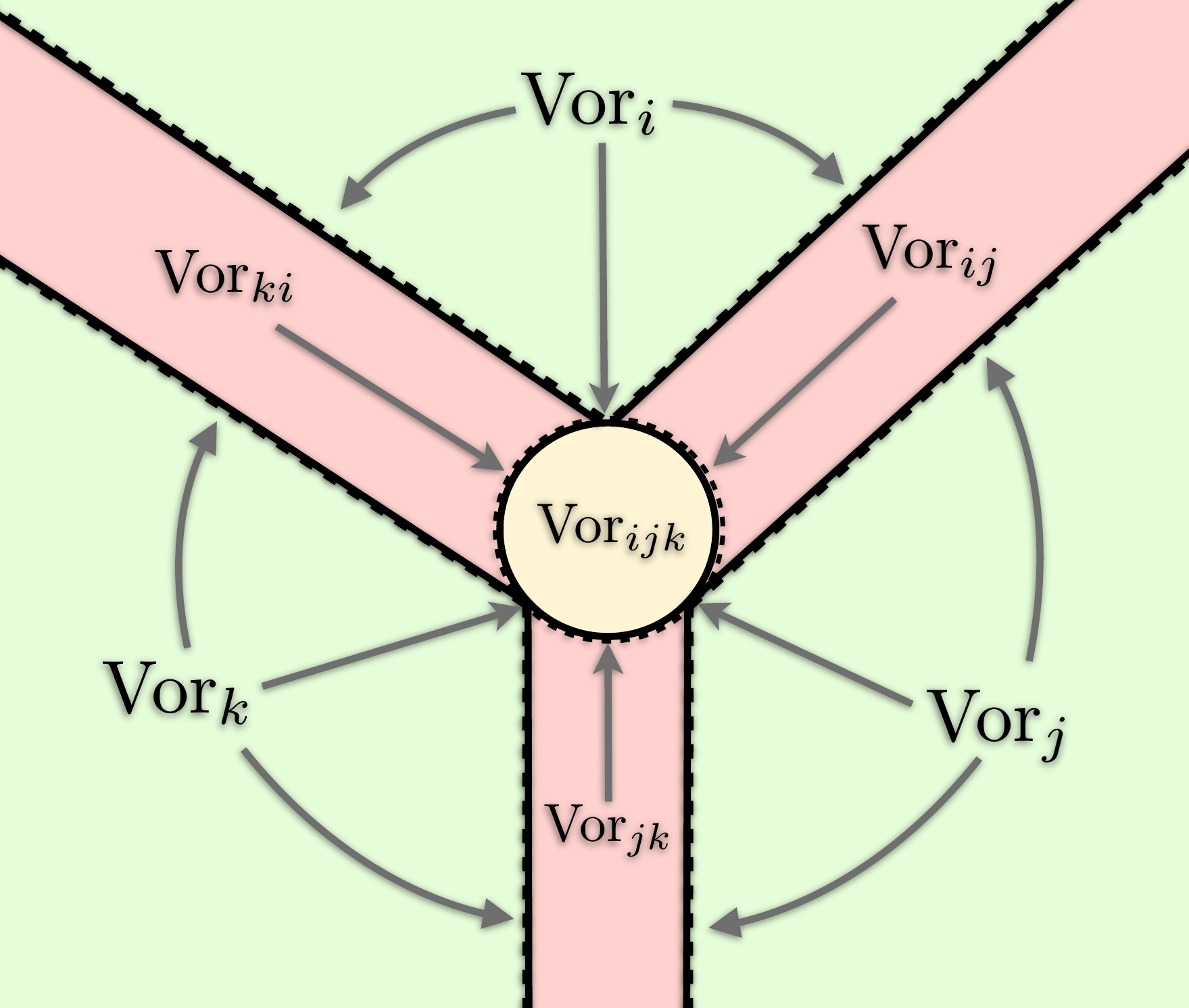}
   \caption{A portion of a Voronoi diagram, with highlighted incidence relations between Voronoi elements. 
   		The incidence relation (definition~\ref{def:incidence}) forms a directed acyclic graph. }
   \label{fig:incidence}
\end{figure}

Consider the following definition of incidence between Voronoi regions (or between connected components of Voronoi regions):

\begin{definition}\label{def:incidence}
Given $I,J\subseteq S$, we say that $\Vor_I$ is incident to $\Vor_J$ (written $\Vor_I \rightsquigarrow \Vor_J$) iff 
	$\overline{\Vor_I}\cap\overline{\Vor_J}\ne\phi$ 
	\emph{and} $I\subset J$. 
\end{definition}

\begin{remark}
By orphan-freedom, and lemma~\ref{lem:connectededges}, both Voronoi regions and edges are connected
	(except for isolated edges, which are defined in section~\ref{sec:propedges}).
For simplicity, in the sequel we refer to connected components of Voronoi vertices simply as ``Voronoi vertices",
	except for the statement of lemma~\ref{lem:vertexincidence}, which makes this distinction explicit. 
\end{remark}

Note that this definition and the one in section~\ref{sec:properties} are equivalent since, 
	for distinct sets $I\ne J$, and 	by the continuity of $D$, 
	$\overline{\Vor_I}\cap\overline{\Vor_J}\subseteq\overline{\Vor_J}$ implies $I\subset J$
	(and viceversa).


%

Given the following substitution rules:
\begin{eqnarray*}
	A,B\rightsquigarrow C ~&\Rightarrow~  A\rightsquigarrow C \text{ and } B\rightsquigarrow C \\
	A \rightsquigarrow B,C ~&\Rightarrow~ A\rightsquigarrow B \text{ and } A\rightsquigarrow C, 
\end{eqnarray*}
	the following are the incidence relations depicted in figure~\ref{fig:incidence}: 
\begin{eqnarray*}
	\Vor_i,\Vor_j &\rightsquigarrow& \Vor_{ij},\Vor_{ijk} \\
	\Vor_j,\Vor_k &\rightsquigarrow& \Vor_{jk},\Vor_{ijk} \\
	\Vor_k,\Vor_i &\rightsquigarrow& \Vor_{ki},\Vor_{ijk} \\
	\Vor_{ij},\Vor_{jk},\Vor_{ki} &\rightsquigarrow& \Vor_{ijk}, 
\end{eqnarray*}
where we often write $\Vor_{ij}$ instead of $\Vor_{\{i,j\}}$ for simplicity. 


\begin{property}\label{prop:boundaryincidence}
	All points in the boundary of a Voronoi element $\Vor_I$ belong to either $\Vor_I$, 
		or to an element 
		that $\Vor_I$ is incident to. 
\end{property}
\begin{proof}
	Let $p\in\partial\Vor_I$, and $J$ be the set of sites that $p$ is equidistant to. 
	Since $p\in\partial\Vor_I$, by the continuity of $D$, $p$ is equidistant to all sites in $I$, 
		and therefore $I\subseteq J$. 
	The property follows from the definition of incidence. 
\end{proof}

\begin{property}
	From the properties of strict set containment, 
		it follows that the incidence relation $\rightsquigarrow$ forms a directed acyclic graph  
	(a cycle $\Vor_I\rightsquigarrow\Vor_J\rightsquigarrow\Vor_K\rightsquigarrow\Vor_I$ would imply $I\subset I$, a contradiction). 
\end{property}

%
	
From property~\ref{prop:boundaryincidence} it follows that closed Voronoi elements 
	are those with zero out-degree in the incidence graph (e.g.\ $\Vor_{ijk}$ in figure~\ref{fig:incidence}), 
and that open Voronoi elements (i.e.\ Voronoi regions) are those with zero in-degree 
	(e.g.\ $\Vor_i,\Vor_j,\Vor_k$ in figure~\ref{fig:incidence}). 


%

\subsubsection{Properties of Voronoi vertices}

The main properties at Voronoi vertices are derived from the two assumptions in section~\ref{sec:setup}. 
Assumptions~\ref{ass:BAA} and~\ref{ass:EGA} are useful when deriving properties of the vertex at infinity, 
	and bounded vertices (all other vertices), respectively. 

Given the set negated gradients $g_1,\dots,g_m$ at a bounded vertex point (eq.~\ref{eq:neggrad}), 
	by assumption~\ref{ass:EGA} they are distinct vertices of their convex hull. 
It is then possible to define ``outward" vectors $d_1,\dots,d_m$ (eq.~\ref{eq:dkdef}) 
	such that eq.~\ref{eq:extremal1} holds. 
This is because, for each $k=1,\dots,m$, 
	eq.~\ref{eq:extremal1} simply requires all gradients other than $g_k$
	to be below the (red dotted) line orthogonal to $d_k$ passing through $g_k$
	(as shown in fig.~\ref{fig:EGA.a} for $d_1$), 
	which is possible because $g_1,\dots,g_m$ are the distinct vertices of $\CH\{g_1,\dots,g_m\}$.
	
Figure~\ref{fig:EGA.b} shows that eq.~\ref{eq:extremal2} holds for the same reason as above. 
Given two gradients that are adjacent vertices of $\CH\{g_1,\dots,g_m\}$
	(for instance $g_1,g_2$), 
	eq.~\ref{eq:extremal2} (in this case with $k=1, k\oplus 1=2$)
	is possible whenever all gradients different from $g_1,g_2$
	are simultaneously below two lines, both passing through $g_1$ and orthogonal to $d_1$ and $d_2$
	(the gray area). 
This holds because the outward directions $d_k$ can be chosen to form an obtuse angle with 
	both segments $g_k,g_{k\oplus 1}$ and $g_k,g_{k\ominus 1}$. 
The same argument applies to eq.~\ref{eq:extremal3}.

\begin{figure}[htbp]
   \centering
	\subfloat[]{\label{fig:EGA.a}\includegraphics[width=2.2in]{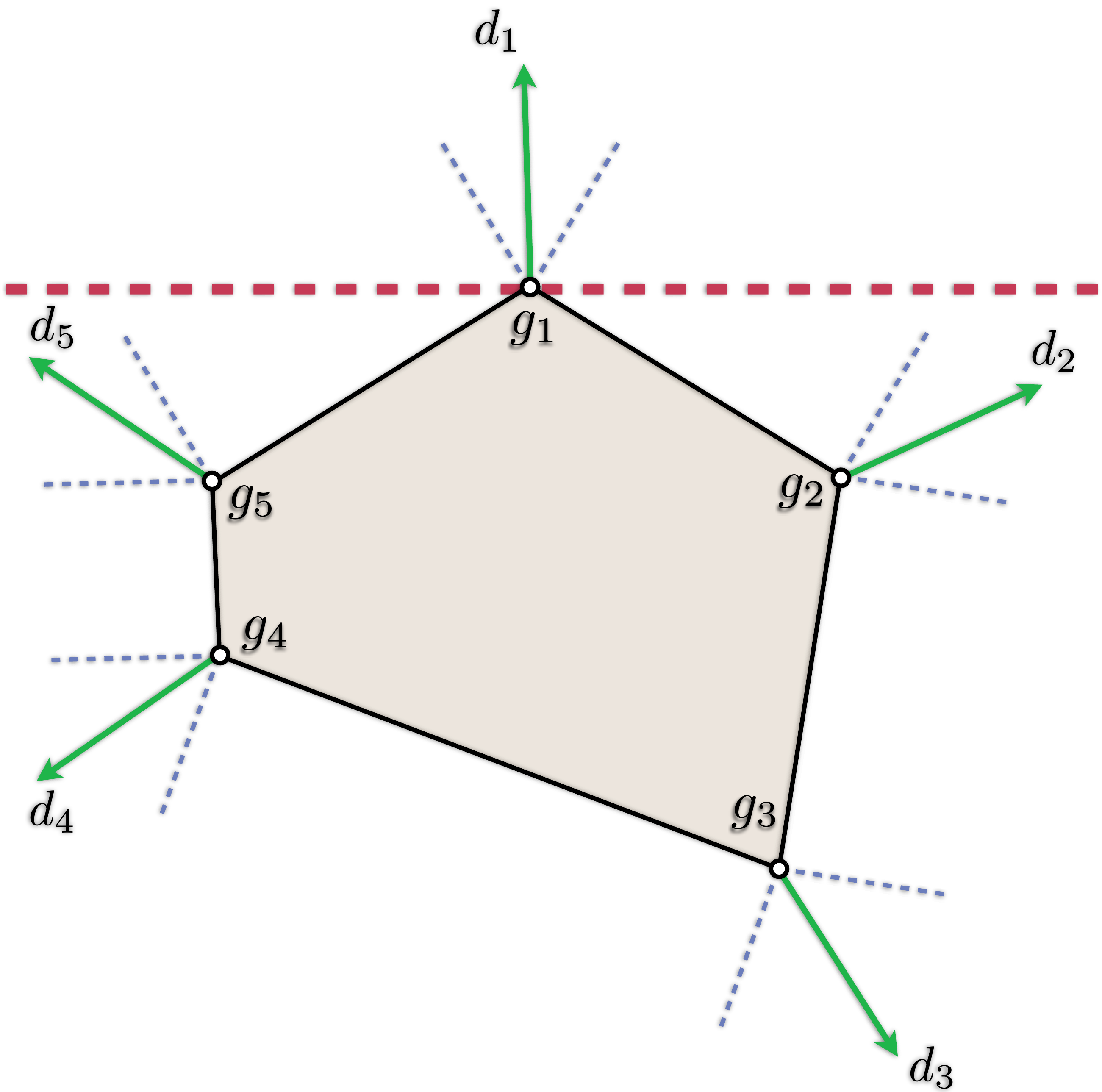}}\quad\quad\quad
	\subfloat[]{\label{fig:EGA.b}\includegraphics[width=2.4in]{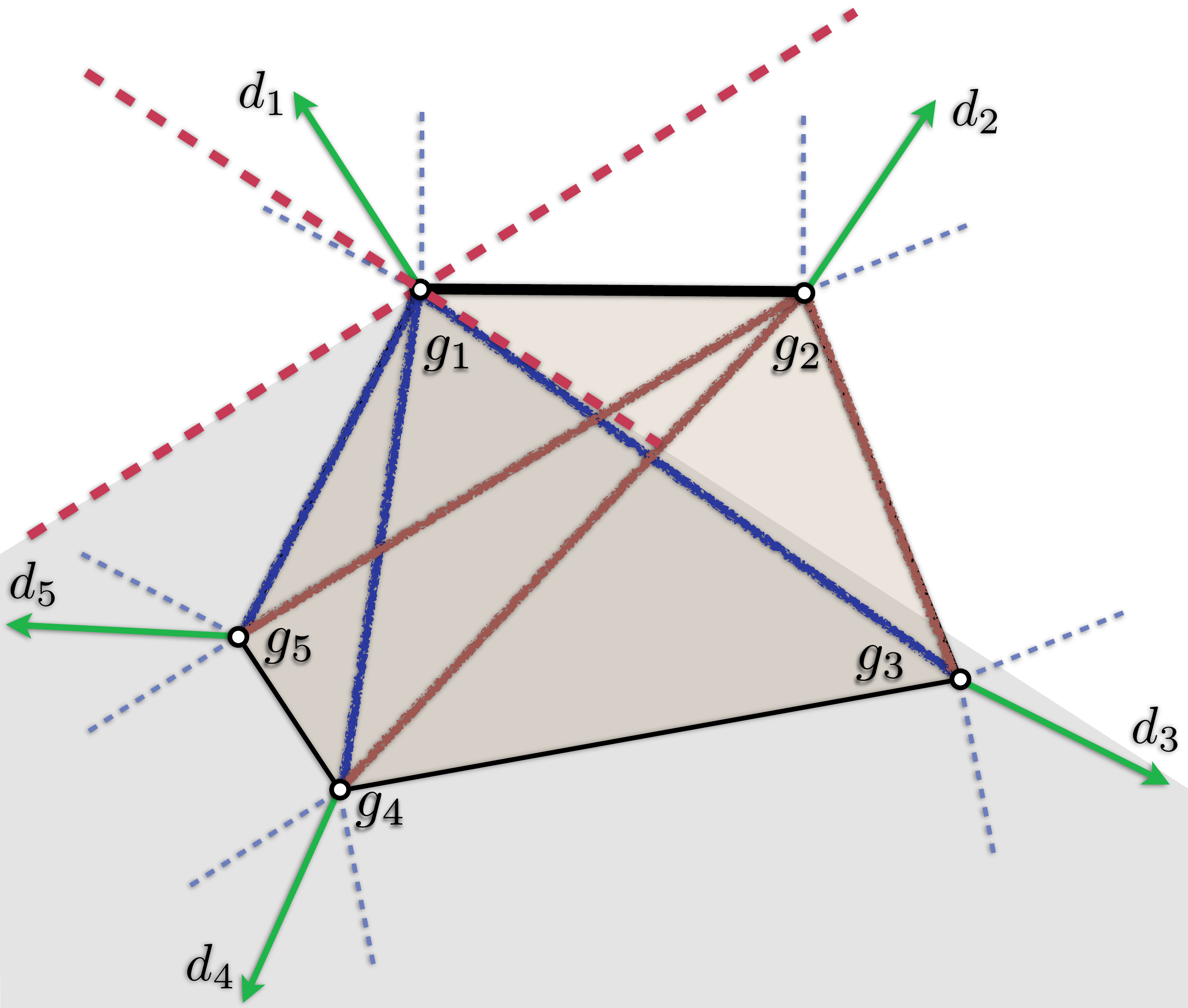}}
   \caption{Diagrams used in the proof of lemma~\ref{lem:vertexincidence}. 
   		Assumption~\ref{ass:EGA} ensures that for all $k$ there is a vector $d_k$ 
				with $\left<d_k,s_k-s_j\right> > 0, j\ne k$ (a), 
			and such that all intermediate direction vectors $d$ between $d_k$ and $d_{k\oplus 1}$ 
				satisfy $\left<d,s_k-s_j\right> >0, j\ne k,k\oplus 1$ (blue lines)
				and $\left<d,s_{k\oplus 1}-s_j\right> >0, j\ne k,k\oplus 1$ (red lines) (b). }
   
\end{figure}

\begin{lemma}[Incidence at Voronoi vertices]\label{lem:vertexincidence}
	A Voronoi vertex $\Vor_I$ 
		is a collection of discrete points, 
		at each of which there is an ordered set of indices $i_1,\dots,i_m$ 
		such that $I=\{i_i,\dots,i_m\}$ and the following incidence relations hold:
	\begin{eqnarray*}
		\Vor_{i_1},\dots,\Vor_{i_m} &\rightsquigarrow& \Vor_I   ~~~~~~~~~~~~~ \text{(region-vertex incidence)}\\
		\Vor_{\{i_1,i_2\}},\Vor_{\{i_2,i_3\}},\dots,\Vor_{\{i_m,i_1\}} &\rightsquigarrow& \Vor_I ~~~~~~~~~~~~~ \text{(edge-vertex incidence)}.
	\end{eqnarray*}
	Additionally, if an edge $\Vor_{jk}$ is incident to a vertex $\Vor_I$, then $\Vor_j,\Vor_k \rightsquigarrow \Vor_{jk}$.\\
	If $\Vor_I$ is the vertex at infinity ($\Vor_\infty$), then $i_1,\dots,i_m$ 
		are the indices of the sites in the boundary of the convex hull $\CHS$, 
		in either clockwise or counter-clockwise order. 
\end{lemma}
\begin{proof}
\noindent{\bf [Bounded vertices, $\Vor_I$]}.
Let $\Vor_I$ be a Voronoi vertex not at the point at infinity and 
	$v$ be a point in $\Vor_I$. 
By the extremal gradient assumption (assumption~\ref{ass:EGA}), 
	the negated gradients 
	\begin{equation}\label{eq:neggrad}
	 g_k\equiv  -\nabla_p \D{s_{i_k}}{p}\big|_v ~,~~ i_k \in I 
	\end{equation}
are distinct vertices of their convex hull. 
Let $i_1,\dots,i_m$ be the indices in $I$ ordered (for instance clockwise) 
	around $\partial \CH\{g_1,\dots,g_m\}$, 
	as shown in figure~\ref{fig:EGA.a}.

Since $g_k$, with $k=1,\dots,m$ are distinct vertices of their convex hull, 
	it is easy to show that there are direction (unit) vectors $d_k$, with $k=1,\dots,m$, 
	such that for all $k$ it holds:
	\begin{equation} \label{eq:extremal1}
	 \left< g_k - g_j, d_k \right> > 0 ~,~~ j\in\{1,\dots,m\}, j \ne k.
	\end{equation}
For instance 
	\begin{equation}\label{eq:dkdef}
	 d_k\equiv \frac{(g_k-g_{k\oplus 1})+(g_k-g_{k\ominus 1})}{\|g_k-g_{k\oplus 1}+g_k-g_{k\ominus 1}\|}. 
	\end{equation}

By the multivariate version of Taylor's theorem~\cite[p.\ 68]{konigsberger2006analysis}, 
	for each $k,j$, and $p\in\mathbb{R}^2$, we may write:
	\[ -\D{s_k}{p} + \D{s_j}{p} = \left<g_k - g_j, p-v\right> + o\left(\|p-v\|\right). \]
For each $k=1,\dots,m$, and $j=1,\dots,m$ with $j\ne k$, let $p-v = \mu d_k$, with $\mu > 0$, 
	and let $\alpha_{k,j} \equiv \left<g_k - g_j, d_k\right> / 2 > 0$. 
It then follows that:
	\[ \left[-\D{s_k}{v+\mu d_k} + \D{s_j}{v+\mu d_k}\right]/\mu = 2\alpha_{k,j} + f(\mu), \]
where $\lim_{\mu\rightarrow 0} f(\mu)=0$. 
Note that, crucially, $f$ depends on $\mu$ but not on the direction $d_k$.

Since $f(\mu)\rightarrow 0$ with $\mu\rightarrow 0$, 
	we can pick constants $\varepsilon_{k,j}>0$ sufficiently small so that 
	for all $\mu < \varepsilon_{k,j}$ it holds $|f(\mu)| < \alpha_{k,j}$, 
	and therefore $\left[-\D{s_{i_k}}{v+\mu d_k} + \D{s_{i_j}}{v+\mu d_k}\right]/\mu  > \alpha_{k,j}$. 
Let $\varepsilon >0$ be the minimum of all $\varepsilon_{k,j}$, with $j,k=1,\dots,m$, and $j\ne k$.

Since $v$ is strictly closest to sites $s_{i_1},\dots,s_{i_m}$, 
	let $\delta$ be small enough so all points $p\in\mathbb{R}^2$ with $\|p-v\|<\delta$ 
	are closest only to sites in $s_{i_1},\dots,s_{i_m}$ (which is possible since $D$ is continuous). 
Consider the set of points in a small circle of radius $0 < \mu < \min\{\delta,\varepsilon\}$ around $v$. 
From the above, we have that at the point $v + \mu d_k$, it holds:
	\[ \left[-\D{s_k}{v+\mu d_k} + \D{s_j}{v+\mu d_k}\right]/\mu > \alpha_{k,j} > 0 ~,~~ j\in\{1,\dots,m\}, j \ne k, \]
from which it follows that $v+\mu d_k$ is strictly closer to $s_{i_k}$ than to any other site. 
Since this is true for all $k=1,\dots,m$ and for all sufficiently small $0<\mu<\min\{\delta,\varepsilon\}$, the incidence relations
	\[ \Vor_{i_1},\dots,\Vor_{i_m} \rightsquigarrow \Vor_I \]
follow. 

Because $g_1,\dots,g_m$ are vertices of $\CH\{g_1,\dots,g_m\}$, 
	it is clear, as shown in figure~\ref{fig:EGA.b}, 
	that for each $k=1,\dots,m$ there are constants $\beta_{k,j},\beta_{k\oplus 1,j} > 0$, with $j\ne k$ and $j\ne k\oplus 1$, 
	such that, for every unit vector $d_{k,k\oplus 1}$ intermediate between $d_k$ and $d_{k\oplus 1}$, 
	it holds:
	\begin{align}
		 \label{eq:extremal2} \left< g_k - g_j, d_{k,k\oplus 1} \right> &> 2\beta_{k,j} > 0 ~,~~~~~ j\ne k, j \ne k\oplus 1 \\
		 \label{eq:extremal3} \left< g_{k\oplus 1} - g_j, d_{k,k\oplus 1} \right> &> 2\beta_{k \oplus 1,j} > 0 ~,~~ j\ne k, j \ne k\oplus 1. 
	\end{align}
Let $\xi_k > 0$ be small enough such that for all $0 < \mu < \xi$, 
	it holds $f(\mu) < \min\{\min_j \beta_{k,j}, \min_j \beta_{k\oplus 1,j}\}$. 
Let $\xi\equiv \min_k \xi_k$, 
	then for all $0 < \mu < \min\{\delta,\varepsilon,\xi\}$, 
	and every point $v+\mu d_k$ it holds:
	\begin{align*}
		-\D{s_k}{v+\mu d_{k,k\oplus 1}} + \D{s_j}{v+\mu d_{k,k\oplus 1}} &> \beta_{k,j} > 0 ~,~~~~~  j\ne k, j \ne k\oplus 1 \\
		-\D{s_{k\oplus 1}}{v+\mu d_{k,k\oplus 1}} + \D{s_j}{v+\mu d_{k,k\oplus 1}} &> \beta_{k\oplus 1,j} > 0 ~,~~~~~  j\ne k, j \ne k\oplus 1, 
	\end{align*}
	and therefore $v+\mu d_{k,k\oplus 1}$ is closest to either $s_k,s_{k\oplus 1}$, or to both. 
For each such $\mu$, and for each  $k=1,\dots,m$, 
	by the intermediate value theorem, 
	there is a direction vector $d$ between $d_k,d_{k\oplus 1}$ such that $v+\mu d$ 
	is in $\Vor_{k,k\oplus 1}$. 
Note that, by the above construction, for every such sufficiently small $\mu$, 
	$\Vor_{k,k\oplus 1}$, with $k=1,\dots,m$, are the only Voronoi edges inside the ball of radius $\mu$ around $v$. 
From this it directly follows that: 
\begin{enumerate}
	\item since all points $v+\mu d$, with unit vector $d$ and sufficiently small $\mu$ 
			have been shown to be in a Voronoi region or edge, 
			$v$ is an isolated point of $\Vor_I$;
		since $v$ is a generic point of $\Vor_I$, it follows that $\Vor_I$ is composed of isolated points;
	\item it holds $\Vor_{\{i_1,i_2\}},\Vor_{\{i_2,i_3\}},\dots,\Vor_{\{i_m,i_1\}} \rightsquigarrow \Vor_I $; and 
	\item if a Voronoi edge $\Vor_{jk}$ is incident to $\Vor_I$, then $\Vor_j,\Vor_k\rightsquigarrow \Vor_{jk}$, 
		since the only edges incident to $\Vor_I$ are $\Vor_{i_k,i_{k\oplus 1}}$, with $k=1,\dots,m$. 
\end{enumerate}

\vspace*{0.1in}\noindent{\bf [Vertex at infinity, $\Vor_\infty$]}. 
Incidence to the vertex at infinity is dealt with in section~\ref{sec:boundary}, 
	where lemma~\ref{lem:VW} shows that the only unbounded elements 
	are of the form $\Vor_I$ where all $i_k\in I$ are vertices of $\CHS$, 
	and lemmas~\ref{boundary_easy} and~\ref{lem:hard} show that, 
	if $s_{i_1},\dots,s_{i_m}$ are the vertices on the boundary of $\CHS$ (whether on an edge or vertex of $\partial\CHS$), 
	ordered around $\partial\CHS$, 
	then $\Vor_{i_1},\dots,\Vor_{i_m}$ and $\Vor_{i_1,i_2},\dots,\Vor_{i_m,i_1}$ are the only unbounded elements
	(and therefore incident to $\Vor_\infty$). 
In this sense we can say that the vertex at infinity $\Vor_\infty$ is the Voronoi vertex $\Vor_{i_1,\dots,i_m}$. 
The proofs in section~\ref{sec:boundary} show that points $p$ in any circle of sufficiently large radius are incident 
	only to sites in $s_{i_1},\dots,s_{i_m}$, 
	that $p$ cannot be incident to more than two sites simultaneously (lemma~\ref{lem:contrad}), 
	and therefore $p$ cannot belong to a Voronoi vertex, 
	and finally that $p$ can \emph{only} be simultaneously closest to two consecutive sites of the form $s_{i_k},s_{i_{k\oplus 1}}$
	(page~\pageref{text:boundary}). 
Note that the relevant proofs of section~\ref{text:boundary} use the bounded anisotropy assumption (assumption~\ref{ass:BAA}), 
	but do not use any result from this section. 
\end{proof}

From the proof of lemma~\ref{lem:vertexincidence}, 
	it is clear that 
	the bounded anisotropy assumption (assumption~\ref{ass:BAA}) 
	is constructed so that lemma~\ref{lem:vertexincidence} holds for the vertex at infinity, 
	while the extremal gradient assumption (assumption~\ref{ass:EGA}) 
	is meant to ensure that  lemma~\ref{lem:vertexincidence} holds for regular (bounded) vertices.
%
%

\subsubsection{Properties of Voronoi edges}\label{sec:propedges}

We begin by considering (isolated) Voronoi edges that are bounded and not incident to any Voronoi vertex. 
Since, as will be shown in lemma~\ref{lem:SCedges}, 
	Voronoi edges are simply connected, 
	it is easy to see that for any Voronoi edge $\Vor_{ij}$
		that is not incident to any bounded Voronoi vertex,
	it can only be $\Vor_i\rightsquigarrow\Vor_{ij}$ or $\Vor_j\rightsquigarrow\Vor_{ij}$, 
	and $\Vor_{ij}$ cannot be involved in any other incidence relation. 
To see this, first note that an isolated component of $\Vor_{ij}$ has, 
	by definition, zero out-degree, and therefore it is closed. 
Because $\Vor_{ij}$ is not incident to the vertex at infinity, it is bounded. 
Since $\Vor_{ij}\rightsquigarrow\Vor_{kl}$ implies 
	that their common boundary belongs to vertex $\Vor_{ijkl}$ (where it may be $k=l$), 
	$\Vor_{ij}$ is not incident to any Voronoi edge. 
$\Vor_{ij}$ cannot be incident to a region $\Vor_k$ with $k\notin\{i,j\}$, 
	or else their common boundary would belong to vertex $\Vor_{ijk}$. 
Finally, we show that it cannot be 
	both $\Vor_i\rightsquigarrow\Vor_{ij}$ and $\Vor_j\rightsquigarrow\Vor_{ij}$. 
Because $\Vor_{ij}$ is closed, simply connected, and bounded, 
	by the continuity of $D$, 
	we can consider a sufficiently small $\varepsilon > 0$ 
	such that every $\varepsilon$-offset of its outer boundary cannot be 
	closest to any site $s_k$ with $k\notin\{i,j\}$. 
If $\Vor_i,\Vor_j\rightsquigarrow\Vor_{ij}$, then 
	there must be $0 < \mu < \varepsilon$ 
	such that the $\mu$-offset $\nu_\mu$ of $\Vor_{ij}$'s outer boundary 
	has at least one point closest to $s_i$, and one point closest to $s_j$, 
and therefore, by continuity of $D$, at least one point equally close to $s_i,s_j$. 
Since all points in $\nu_\mu$ are closest to $s_i,s_j$ only, then $\nu_\mu$ has 
	been shown to have a point in $\Vor_{ij}$, 
	contradicting the fact that $\nu_\mu$ is a $\mu$-offset of $\Vor_{ij}$'s outer boundary, 
	and therefore outside $\Vor_{ij}$. 

Let $\Vor_{ij}$ be an bounded isolated Voronoi edge 
	such that $\Vor_i\rightsquigarrow\Vor_{ij}$. 
Because they are not incident to any Voronoi vertex, 
	bounded isolated edges will not be considered part of the primal Voronoi graph.
For simplicity, we consider all points of an isolated edge $\Vor_{ij}$ to be part of its containing
	Voronoi region (say $\Vor_i$), 
	and therefore to be (by definition) strictly closer to $s_i$ than to any other site. 
This is not just a simplification (which does not affect the final Voronoi graph), 
	but will allow us to prove that Voronoi regions are simply connected. 
	



We begin by proving the following technical lemma. 

%
%

\begin{figure}[htbp]
   \centering
   	\subfloat[]{\label{fig:RSC.a}\includegraphics[width=2.5in]{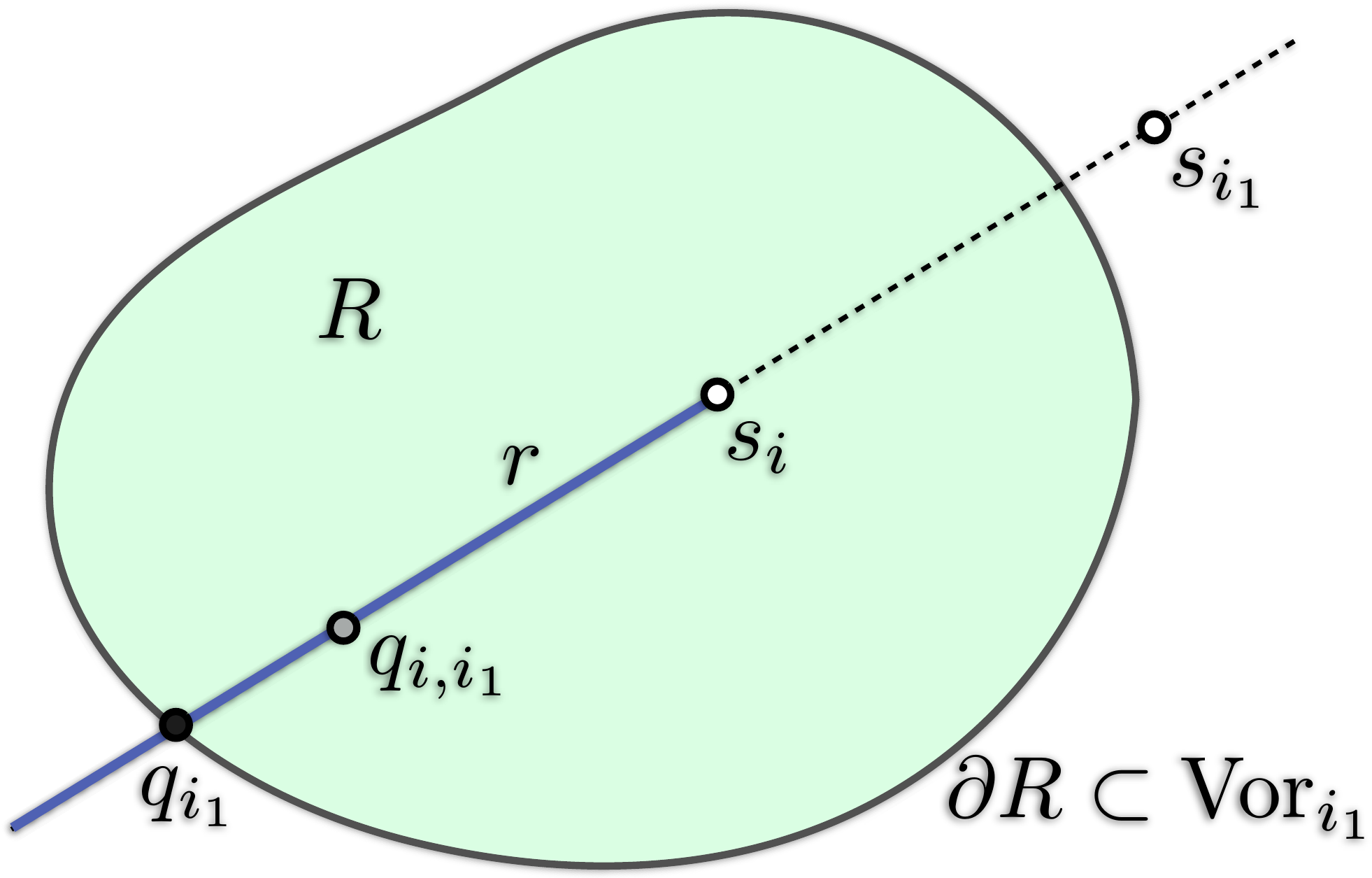}}
	\quad\quad\quad
	\subfloat[]{\label{fig:RSC.b}\includegraphics[width=2.5in]{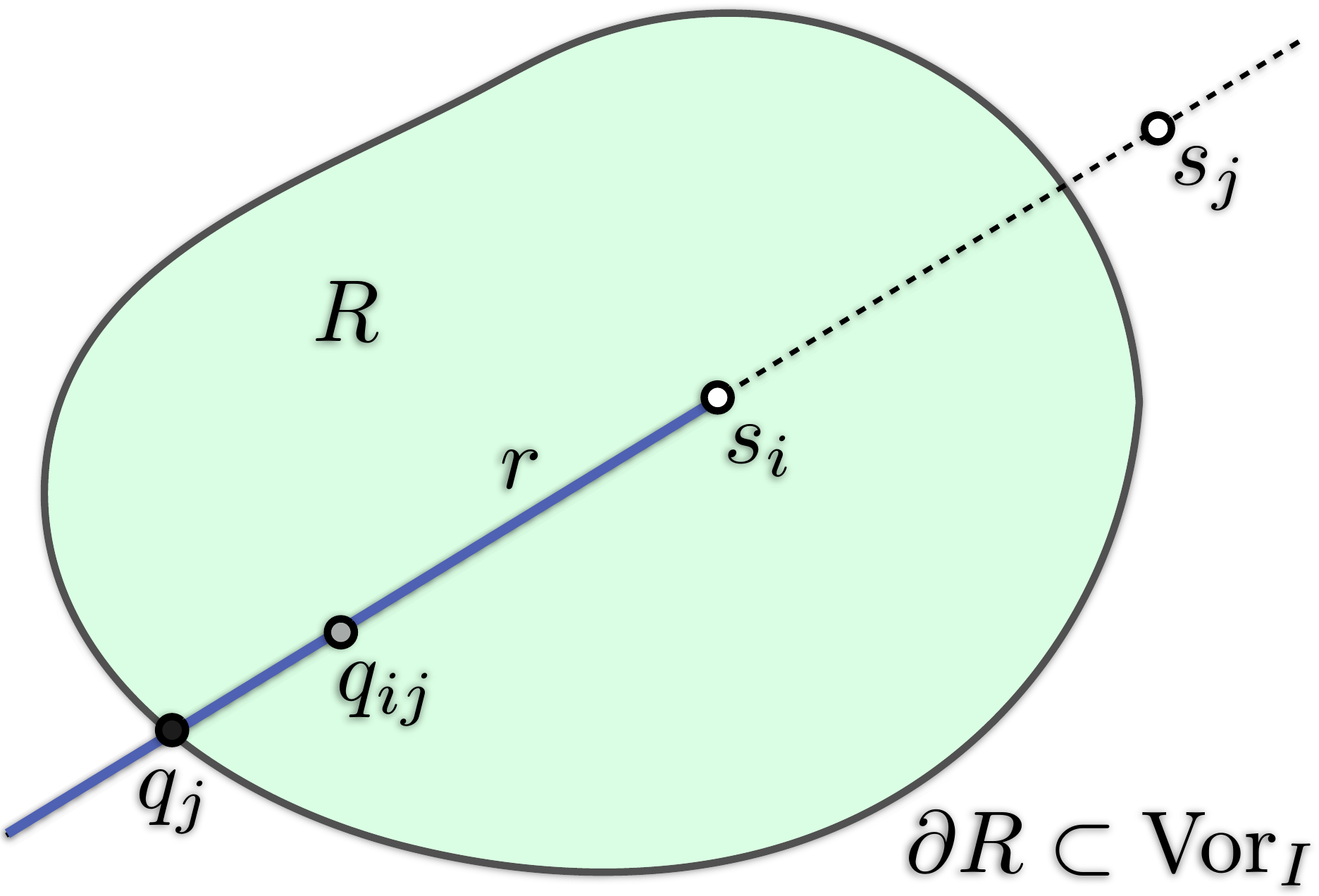}}
   \caption{Diagrams used in the proof of lemma~\ref{lem:RSC}. }
   \label{fig:RSC}
\end{figure}

\begin{lemma}\label{lem:RSC}
	Let the boundary $\partial R$ of $R\subset\mathbb{R}^2$ be a simple, closed path, 
		and $\Vor_I$ be a Voronoi element of an orphan-free diagram. 
	If $\partial R\subseteq\Vor_I$, 
		then $R\subseteq\Vor_I$. 		
\end{lemma}
\begin{proof}
Let $I=i_1,\dots,i_m$, and $\gamma\equiv\partial R$. 
We begin by showing that $R$ does not contain any site $s_i$ whenever $m>1$ or $i\ne i_1$. 

Let $m=1$, and $s_i\in R$ with $i\ne i_1$, as in figure~\ref{fig:RSC.a}. 
Let $r$ be the ray starting from $s_i$ in the direction of $s_i-s_{i_1}$ 
	(note that $s_{i_1}$ may be inside or outside $R$). 
Since $r$ is unbounded and $R$ is bounded, then part of $r$ is outside $R$ and, 
	by the Jordan curve theorem, it must intersect $\gamma$ at some point $q_{i_1}$. 
Since $\gamma\subset\Vor_{i_1}$, 
	$q_{i_1}$ is closest to $s_{i_1}$, while $s_i$ is closest to $s_i$
		(since $\D{s_i}{s_i}=0$ and $\D{\cdot}{s_i}$ is non-negative and convex). 
By the continuity of $D$, there is an intermediate point $q_{i,i_1}$ between $s_i$ and $q_{i_1}$ 
	that is equidistant to $s_{i_1}$ and $s_i$, contradicting lemma~\ref{lem:midpoint}. 

Let $m>1$, and let $s_i$ be any site (figure~\ref{fig:RSC.b}). 
Pick $j\ne i$ among $j\in\{i_1,\dots,i_m\}$, 
	which is always possible because $m>1$. 
The argument is identical in this case, except that, 
	because $\gamma\subset\Vor_I$, then $q_j\in\gamma$ is closest and equidistant to $\{i_1,\dots,i_m\}$, 
	and therefore closer to $s_i$ than to $s_j$, and the same argument holds. 

\vspace*{0.05in}\noindent{\bf [Voronoi regions]}.
We now prove that no point $p\in R$ belongs to a Voronoi region $\Vor_i$ if $m>1$ or $i\ne i_1$. 
Let $p\in R$ belong to $\Vor_i$, with $m>1$ or $i\ne i_1$, we show that this leads to a contradiction. 

We first show that $\Vor_i\subset R$. 
Assume otherwise. Since $\Vor_i$ is open and connected (by the orphan-freedom property), it is path connected. 
Let $\Gamma\subset\Vor_i$ be a simple path from $p$ to a point $q\in\Vor_i$ outside $R$. 
By the Jordan curve theorem, $\Gamma\subset\Vor_i$ intersects $\gamma\subset\Vor_I$, 
	which leads to a contradiction whenever $m>1$ or $i\ne i_1$. 

Since $\Vor_i\subset R$ and, 
	by lemma~\ref{lem:site}, 
	$s_i\in\Vor_i$, then $s_i\in R$, contradicting the fact that $R$ does not contain any site $s_i$ if $m>1$ or $i\ne i_1$. 

\vspace*{0.05in}\noindent{\bf [Voronoi vertices]}.
If $R$ contains a point $p$ that belongs to a Voronoi vertex $\Vor_J$ with $J\ne I$, 
	then $p$ must be in the interior of $R$, since its boundary $\gamma$ is in $\Vor_I$. 
By lemma~\ref{lem:vertexincidence}, $p$ is incident to $\Vor_{j_1},\dots,\Vor_{j_k}$, 
	where $J=j_1,\dots,j_k$ and $k\ge 3$. 
Since $p$ is in the interior of $R$, then there are points 
	$p_{j_i},\dots,p_{j_k}\in R$ that belong to $\Vor_{j_1},\dots,\Vor_{j_k}$, respectively. 
If $m>1$, then this contradicts the fact that $R$ does not have any point in a Voronoi region. 
If $m=1$, since $k\ge 3$, then one of $j_1,\dots,j_k$ must be different from $i_1$, 
	contradicting the fact that $R$ does not have any point in a Voronoi region different from $\Vor_{i_1}$. 

\vspace*{0.05in}\noindent{\bf [Voronoi edges]}.
Let $\Vor'_{ij}$ be a connected component of a Voronoi edge, with $\{i,j\}\ne I$. 
If some point $p\in\Vor'_{ij}$ is in $R$, 
	then $\Vor'_{ij}\subset R$, or else since, by the assumption in section~\ref{sec:assumptions}, 
	$\Vor'_{ij}$ is path connected, there would be a path $\Gamma\subset\Vor'_{ij}$ connecting 
	$p$ to a point of $\Vor'_{ij}$ outside $R$. 
By the Jordan curve theorem $\Gamma\subset\Vor'_{ij}$ would intersect $\gamma\subset\Vor_I$, a contradiction. 

Since we have already discarded isolated Voronoi edges that are not incident to any Voronoi vertex, 
	a Voronoi edge is always incident to a Voronoi vertex and, 
	since $\Vor'_{ij}$ is in the interior of $R$, then its incident Voronoi vertex is in $R$, 
	a contradiction. 

Finally, since we have shown that there cannot be any Voronoi vertices, edges, or regions $\Vor_J$
	with $J\ne I$ in $R$, then it must be $R\subset\Vor_I$. 
\end{proof}

\begin{figure}[htbp]
   \centering
	\subfloat[]{\label{fig:connectededges.a}\includegraphics[width=2.7in]{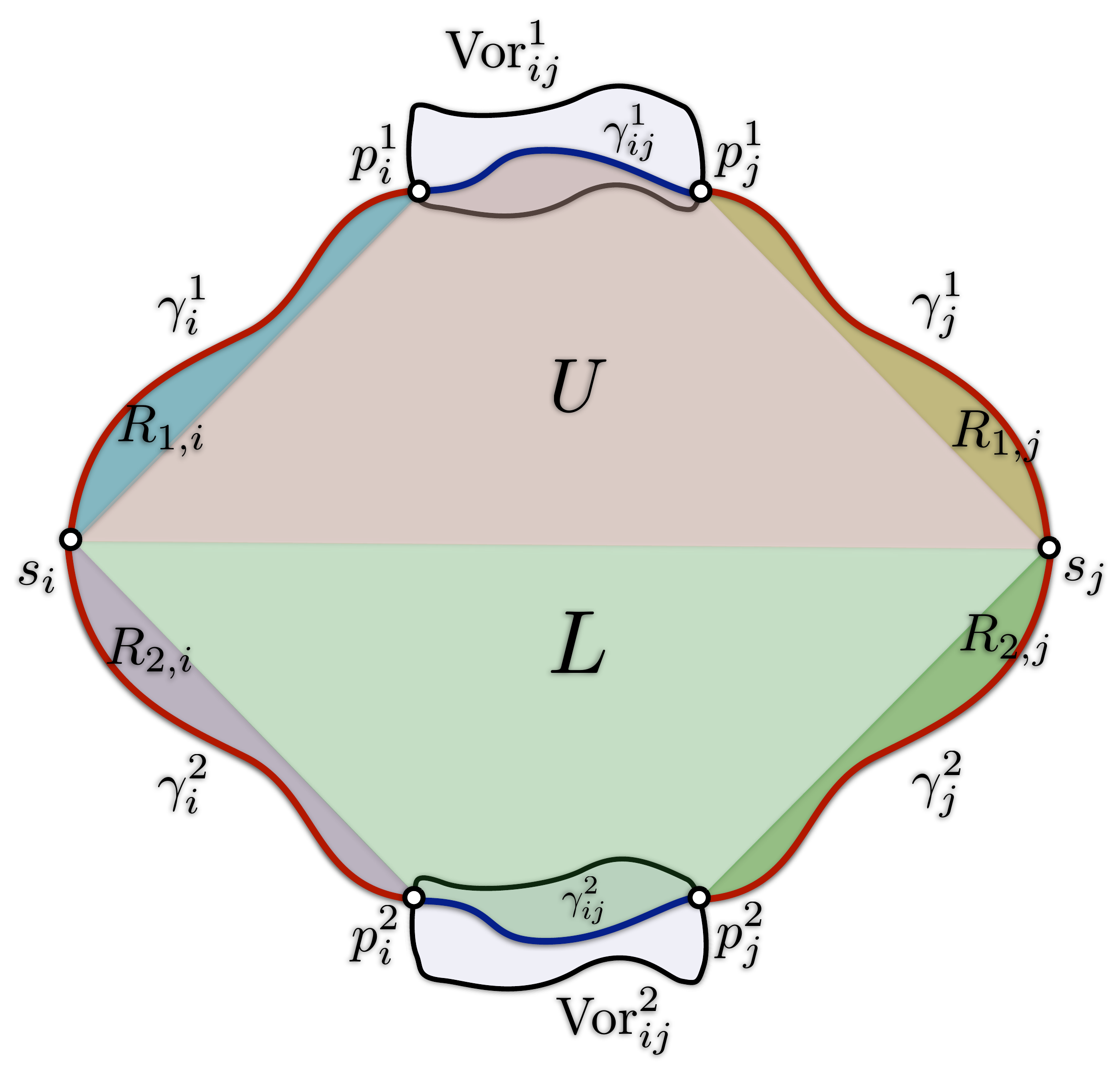}}
	\quad\quad\quad
	\subfloat[]{\label{fig:connectededges.b}\includegraphics[width=2.7in]{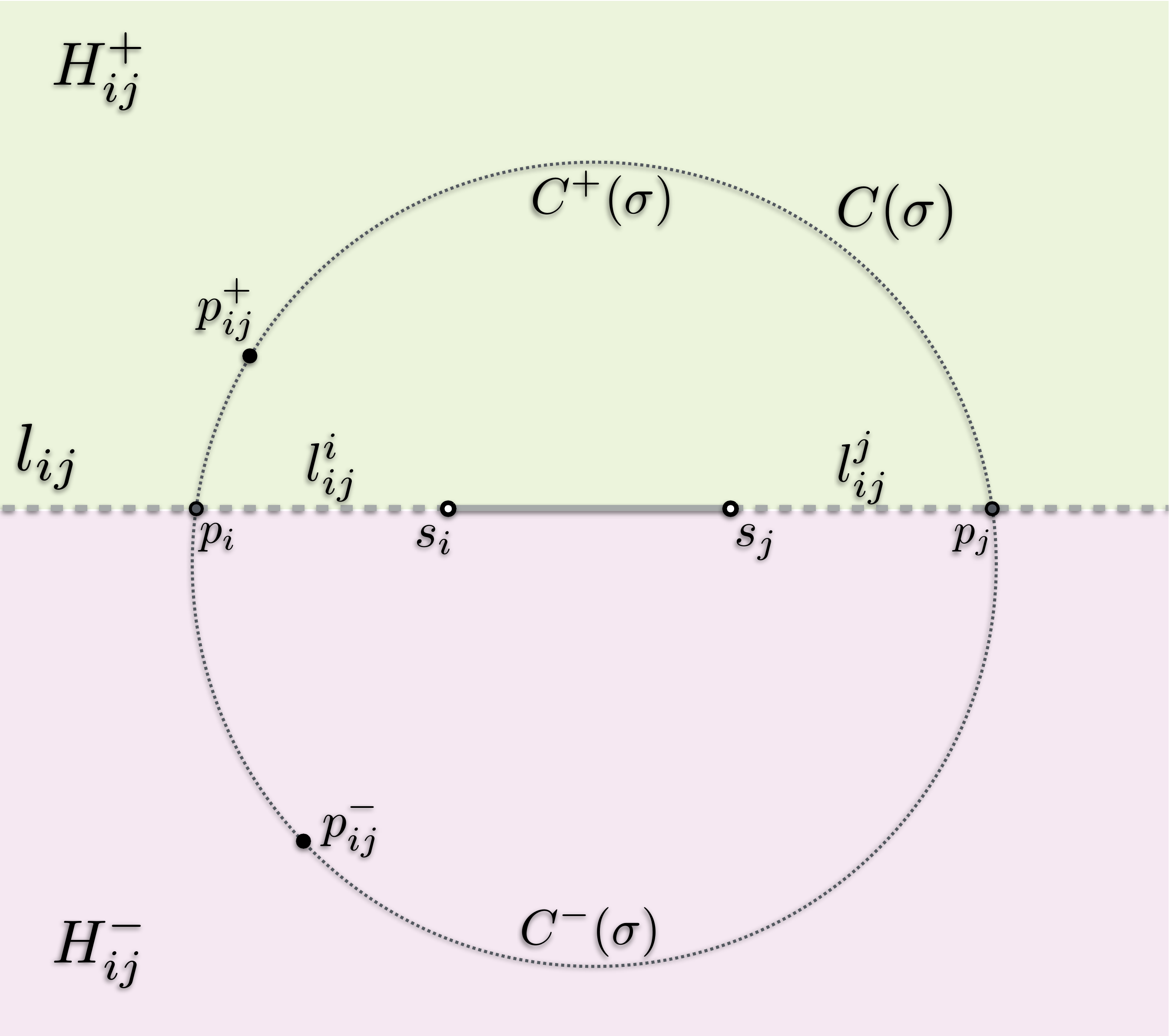}}
   \caption{Diagrams used in the proof of lemmas~\ref{lem:connectededges} (a), 
   			and~\ref{lem:val_ge2} (b).}
\end{figure}

\begin{lemma}\label{lem:connectededges}
	Voronoi edges of an orphan-free diagram are connected. 
\end{lemma}
\begin{proof}
%
%

Let $\Vor^1_{ij},\Vor^2_{ij}$ be two disconnected pieces of a Voronoi edge $\Vor_{ij}$,
	as shown in figure~\ref{fig:connectededges.a}.
Since we have discarded (bounded) isolated edges, we assume that  
	$\Vor_{ij}$ is incident to at least one vertex, and therefore
	by lemma~\ref{lem:vertexincidence},
	 it is $\Vor_i,\Vor_j\rightsquigarrow\Vor^1_{ij}$ and $\Vor_i,\Vor_j\rightsquigarrow\Vor^2_{ij}$.

Since $\Vor_i,\Vor_j$ are incident to both $\Vor^1_{ij},\Vor^2_{ij}$, 
	the boundaries of $\Vor_i$ and $\Vor_{ij}^1$ overlap
	(and likewise $\partial\Vor_i \cap \partial\Vor_{ij}^2, \partial\Vor_j \cap \partial\Vor_{ij}^1, \partial\Vor_j \cap \partial\Vor_{ij}^2 \ne \phi$). 
Let $p^1_i\in\partial\Vor_i\cap\partial\Vor^1_{ij}$ 
	be a point in the common boundary between $\Vor_i$ and $\Vor^1_{ij}$, 
	and $p^1_j\in\partial\Vor_j\cap\partial\Vor^1_{ij}$ be a point in the common boundary between $\Vor_j$ and $\Vor^1_{ij}$, 
	and define the points $p^2_i,p^2_j$ analogously.
Since $\Vor^1_{ij},\Vor^2_{ij}$ are disjoint, 
	it holds $p^1_i\ne p^2_i$ and $p^1_j\ne p^2_j$, 
	and therefore by lemma~\ref{lem:regionpath}
		 there are non-crossing simple paths $\gamma^1_i,\gamma^2_i\subset\Vor_i$ 
		 	from $s_i\in\Sites$ to $p^1_i,p^2_i$, respectively, 	
		 and non-crossing simple paths $\gamma^1_j,\gamma^2_j\subset\Vor_j$ 
		 	from $s_j\in\Sites$ to $p_{j,1},p_{j,2}$, respectively. 
Additionally, since by the assumption in section~\ref{sec:assumptions} $\Vor^1_{ij},\Vor^2_{ij}$ are path connected, 
	there are simple paths $\gamma^1_{ij}\subset\Vor^1_{ij}$ and $\gamma^2_{ij}\subset\Vor^2_{ij}$
	connecting $p^1_i$ to $p^1_j$, and $p^2_i$ to $p^2_j$, respectively. 
Let $\gamma^1$ be the concatenation of paths $\gamma^1_i,\gamma^1_{ij},\gamma^1_j$, 
	and $\gamma^2$ be the concatenation of paths $\gamma^2_i,\gamma^2_{ij},\gamma^2_j$. 
By construction, and since $\Vor^1_{ij},\Vor^2_{ij}$ are disjoint, 
	the simple paths $\gamma^1,\gamma^2$ only meet at their endpoints $s_i,s_j$.

Let $\gamma$ be the simple closed curve resulting from concatenating $\gamma^1,\gamma^2$. 
By the Jordan curve theorem, $\gamma$ divides the plane into an interior ($Int$) and exterior regions, bounded by $\gamma$. 
We first show that $Int$ does not contain any sites (other than $s_i,s_j$). 

\vspace*{0.1in}\noindent{\bf [$Int$ contains no sites]}.
We first divide $Int$ in three parts, as shown in figure~\ref{fig:connectededges.a}: 
\begin{enumerate}
	\item the region $U$ bounded by $\overline{s_i p^1_i}$, $\gamma^1_{ij}$, and $\overline{p^1_j s_j}$,
	\item the region $L$ bounded by $\overline{s_i p^2_i}$, $\gamma^2_{ij}$, and $\overline{p^2_j s_j}$, 
	\item and $R\equiv Int \setminus \left(U \cup L\right)$.
\end{enumerate}
We begin by observing that if $w_{ij}\in \Vor_{ij}$, then the triangle $\triangle{s_i w_{ij} s_j}$ 
	cannot contain any site (other than $s_i,s_j$)
	because 1) $w_{ij}$ is closest and equidistant to $s_i,s_j$,
	and 2) the ball of the first kind $\theta_{w_{ij}}(s_i)$ 
	centered at $w_{ij}$ with $s_i,s_j$ in its boundary (see table~\ref{table:notation}) 
	is convex and therefore contains $\triangle{s_i w_{ij} s_j}$. 
Since the sides of $\triangle{s_i w_{ij} s_j}$ are line segments, and $\theta_{w_{ij}}(s_i)$ is strictly convex, 
	the only points of $\triangle{s_i w_{ij} s_j}$ touching the boundary of $\partial\theta_{w_{ij}}(s_i)$
	are $s_i,s_j$, and therefore a site at any other point in $\triangle{s_i w_{ij} s_j}$ would be strictly closer to $w_{ij}$ than $s_i,s_j$, 
	a contradiction.

Since $U$ can be written as the union of triangles with vertices $s_i,w_{ij},s_j$ with $w_{ij}\in\gamma^1_{ij}\subset\Vor^1_{ij}\subset\Vor_{ij}$, 
	then $U$ does not contain any site other than $s_i,s_j$. 
An analogous argument proves that $L$ does not contain any site other than $s_i,s_j$.

We split the remaining region $R$ into four parts $R_{1,i},R_{1,j},R_{2,i},R_{2,j}$. 
Let $R_{1,i}$ be the part of $R$ bounded by the segment $\overline{s_i p^1_i}$ and the curve $\gamma^1_i$. 
Let $R'_{1,i} \equiv \displaystyle{\cup_{r\in \gamma^1_i} \overline{s_i r}}$ be the union of segments connecting $s_i$ to points in $\gamma^1_i$. 
Clearly, it is $R_{1,i}\subset R'_{1,i}$. We show that $R'_{1,i}$ cannot contain any site other than $s_i$, 
	and thus the same is true of $R_{1,i}$. 

Let $z\in R'_{1,i}$ be a site, and let $r\in\gamma^1$ be the point such that $z\in\overline{s_i r}$. 
Because $r\in\gamma^1_{i}\subset\overline{\Vor_i}$, $r$ is closest and equidistant to $s_i$ (and possibly also to $s_j$), 
	that is: $\D{s_i}{r} \le \D{s_k}{r}$ for all $k=1,\dots,n$. 
Since $z\in\overline{s_i r}$ and $z\ne s_i$, we can write $z=\lambda s_i + (1-\lambda) r$, with $0 \le \lambda < 1$, 
	and therefore by the strict convexity of $\D{\cdot}{r}$ it holds:
	\[ \D{z}{r} = \D{\lambda s_i + (1-\lambda) r}{r} < \lambda\D{s_i}{r} + (1-\lambda)\D{r}{r} = \lambda\D{s_i}{r} < \D{s_i}{r}, \]
where the last equality follows from $\D{r}{r}=0$, and the last inequality follows from the non-negativity of $D$. 
This shows that the site $z$ is \emph{strictly} closer to $r$ than $s_i$, a contradiction. 
Therefore there are no sites in $R'_{1,i}$, and thus no sites in $R_{1,i}\subseteq R'_{1,i}$ either. 
Applying an identical argument to $R_{1,j},R_{2,i},R_{2,j}$ shows that $R$ cannot contain any sites other than $s_i,s_j$.

\vspace*{0.05in}\noindent{\bf [Points in $Int$ can only be closest to $s_i$ and/or $s_j$]}.
We begin by showing that there is no point $p\in Int$ that is 
	\emph{strictly} closer to a site $s_k\notin \{s_i,s_j\}$ than to any other site ($p\in\Vor_k$). 
If $p\in Int$ is closest to $s_k\notin\{s_i,s_j\}$, then we first show that $\Vor_k$ is wholly contained in $Int$. 
Assume otherwise, and pick a point $q\in\Vor_k$ outside $Int$. 
Since Voronoi regions are path-connected, let $\Gamma_{pq}\subset\Vor_k$ be a path connecting $p,q$. 
By the Jordan curve theorem, $\Gamma_{pq}$ 
 	crosses the boundary $\gamma\subset \Vor_i\cup \Vor_j\cup \Vor_{ij}$, 
	contradicting the fact that $\Gamma_{pq}\subset\Vor_k$. 
Since $\Vor_k$ is completely inside $Int$ then, by lemma~\ref{lem:site}, it is $s_k\in\Vor_k\subset Int$, 
	contradicting the fact the $Int$ contains no sites other than $s_i,s_j$, 
	and therefore $\Vor_k\cap Int=\phi$ with $k\notin\{i,j\}$. 

We now show that no point $p\in Int$ can be closest to $s_k\notin\{s_i,s_j\}$, even if it is also simultaneously closest to 
	$s_i$ and/or $s_j$. 
Since $p$ is closest to $s_k$, 
	and the boundary of $Int$ is $\gamma\subset \Vor_i\cup\Vor_j\cup\Vor_{ij}$, then $p$ belongs to the interior of $Int$. 
By definition, $p$ belongs to a Voronoi edge or vertex. If it belongs to a Voronoi vertex and is closest to $s_k\in\Sites$ then, 
	by lemma~\ref{lem:vertexincidence}, and since Voronoi vertices are composed of isolated points, 
	$p$ is incident to $\Vor_k$, a contradiction since $\Vor_k\cap Int=\phi$ whenever $k\notin\{i,j\}$. 
Therefore $Int$ does not contain any Voronoi vertices.

Finally, we show that no point $p\in Int$ can be closest to a site $s_k\notin\{s_i,s_j\}$ 
	and belong to a Voronoi edge $\Vor_E$. 
Since $p$ is in the interior of $Int$, the connected component $\Vor'_E$ of $\Vor_E$ that $p$ belongs to 
	must be fully contained in $Int$, 
	or else by the Jordan curve theorem $\Vor'_E$ would be separated by the boundary $\gamma\subset \Vor_i\cup\Vor_j\cup\Vor_{ij}$ of $Int$. 
Since we have discarded connected components of Voronoi edges not incident to any Voronoi vertex, 
	then $\Vor'_E$ is incident to some vertex $\Vor_I$. 
Since $\Vor'_E$ is in the interior of $Int$, then $\Vor_I$ must be contained in $Int$. 
As we have shown above, $Int$ does not contain any Voronoi vertex, and therefore $p$ cannot be closest to $s_k\notin\{s_i,s_j\}$.

\vspace*{0.1in}\noindent{\bf [$\Vor_{ij}$ is connected]}.
Finally, we show that there is a path in $\Vor_{ij}$ connecting $\Vor^1_{ij}$ to $\Vor^2_{ij}$, 
	and therefore $\Vor_{ij}$ is connected. 
Recall that all points in $Int$ can only be closest to $s_i$ and/or $s_j$, 
	that $\gamma^1,\gamma^2$ are simple paths from $s_i$ to $s_j$, 
	and that, by construction, they do not meet except at their endpoints. 
Clearly, $\gamma^1,\gamma^2$ are path homotopic~\cite[p.\ 323]{munkres2000topology}, 
	for instance via the straight-line homotopy.

We begin by constructing a path homotopy $F$ between $\gamma^1$ and $\gamma^2$
	(a continuous function $F:[0,1]\times [0,1]\rightarrow \mathbb{R}^2$ 
		such that $F(\cdot,0)=\gamma^1(\cdot)$ and $F(\cdot,1)=\gamma^2(\cdot)$)
	contained in $Int$. 
Since $\gamma$ is a Jordan curve, and $Int$ is simply connected, by Carath\'eodory's theorem~\cite{conformal}, 
	there is a homeomorphism $h$ from $\overline{Int}$ to the closed unit disk $D_2$ that maps $\gamma$ to the unit circle. 
Since $\gamma^1,\gamma^2\subset\gamma$ and $D_2$ is convex, the straight-line homotopy $F'$ 
	between $h(\gamma^1)$  and $h(\gamma^2)$ is contained in $D_2$. 
We can now inversely map this homotopy through $h^{-1}$ to obtain a path homotopy $F=h^{-1}\circ F'$ 
	between $\gamma^1$ and $\gamma^2$ which is contained in $Int$ 
	(i.e.\ $F(\cdot,\alpha)\subset\overline{Int}$ with $0\le\alpha\le 1$).

Since every path $F(\cdot,\alpha)$ with $0\le\alpha\le 1$ starts at $s_i$ and ends at $s_j$, 
	and $D$ is continuous, there is $0< t_\alpha < 1$ such that $F(t_\alpha,\alpha)\in\overline{Int}$ is equidistant to $s_i,s_j$. 
Since we have shown above that all points in $\overline{Int}$ are closest to $s_i$ and/or $s_j$, 
	then $F(t_\alpha,\alpha)\in\Vor_{ij}$ for $0\le\alpha\le 1$. 
By the continuity of $D$ and $F$, is it possible to choose $t_\alpha$ to be continuous with $\alpha$, 
	and such that the path $\Phi:[0,1]\rightarrow\mathbb{R}^2$ 
	with $\Phi(\alpha)=F(t_\alpha,\alpha)$ is $\Phi([0,1])\subset\overline{Int}\cap\Vor_{ij}$. 
Since the path $\Phi$ is defined to start at $\Vor^1_{ij}$ and end at $\Vor^2_{ij}$, then 
	$\Vor^1_{ij}$ and $\Vor^2_{ij}$ are connected, 
	and therefore $\Vor_{ij}$ must be connected. 
\end{proof}

\begin{lemma}\label{lem:SCedges}
	Voronoi edges of orphan-free diagrams are simply connected. 
\end{lemma}

\begin{proof}

Recall that, by the assumption in section~\ref{sec:assumptions}, 
	connected Voronoi edges are also path connected.

Let $\Vor_{ij}$ be a Voronoi edge, and $\gamma\subset\Vor_{ij}$ be a simple path not contractible to a point. 
By the Jordan curve theorem, $\gamma$ divides the plane into an exterior (unbounded), 
	and an interior (bounded) region $R$. 
By lemma~\ref{lem:RSC}, $R\subset\Vor_{ij}$, and therefore $\gamma$ is contractible to a point. 

%
%
%
%
\end{proof}

\subsubsection{Properties of Voronoi regions}

\begin{lemma}\label{lem:regionSC}
Voronoi regions of orphan-free diagrams are simply connected.
\end{lemma}

\begin{proof}
Let $\Vor_i$ be a Vornoi region, which must be connected since the diagram is orphan-free. 
Since $\Vor_i$ is open, it is path connected~\cite[p.\ 158]{munkres2000topology}.

Assume that $\Vor_i$ is not simply connected, 
	and therefore has a closed simple path $\gamma\subset\Vor_i$ that is not contractible to a point. 
By the Jordan curve theorem the path $\gamma$ separates the plane into an exterior and an interior region $R$. 
By lemma~\ref{lem:RSC}, $R\subset\Vor_{i}$, and therefore $\gamma$ is contractible to a point. 
\end{proof}

\begin{lemma}\label{lem:regionpath}
For every Voronoi region $\Vor_i$ of an orphan-free Voronoi diagram, 
	there is a collection of simple paths 
	connecting the site $s_i$ to each point in the boundary of $\Vor_i$, 
	 such that:
	\begin{enumerate}
	\item all paths are contained in $\overline{\Vor_i}$,
	\item paths intersect the boundary $\partial\Vor_i$ only at the final endpoint, and 
	\item two paths meet only at the starting point $s_i$.  
	\end{enumerate}
\end{lemma}
\begin{proof}
By the assumption in section~\ref{sec:assumptions}, the boundary of Voronoi regions are simple closed paths. 
Since a Voronoi region $\Vor_i$ is also simply connected (lemma~\ref{lem:regionSC}), we may use Carath\'eodory's theorem~\cite{conformal}
	to map $\overline{\Vor_i}$ to the closed unit disk $D_2$ through a homeomorphism $h$  
	that maps the boundary $\partial\Vor_i$ to the unit circle. 
Since, by lemma~\ref{lem:site}, $s_i$ is an interior point of $\Vor_i$, 
	then $s'_i\equiv h(s_i)$ is an interior point of $D_2$. 
We now simply construct a set of straight paths 
	from $s'_i$ to each point in the unit circle. 
These paths are contained in $D_2$, and meet only at the starting point. 
We map them back through $h^{-1}$ to obtain the desired set of paths. 
\end{proof}

\subsection{Voronoi edges are incident to two and only two Voronoi vertices}

\begin{lemma}\label{lem:val_ge2}
No Voronoi edge is incident to just one Voronoi vertex.
\end{lemma}
\begin{proof}
Let $\Vor_{ij}$ be a Voronoi edge incident to just one Voronoi vertex $\Vor_I$. 
By lemma~\ref{lem:vertexincidence}, it is $\Vor_i\rightsquigarrow\Vor_{ij}$, 
	and therefore $\Vor_{ij}$ has a common boundary with $\Vor_i$. 
Recall from property~\ref{prop:boundaryincidence} that the boundary $\partial\Vor_i$ belongs 
	to Voronoi edges and vertices to which $\Vor_i$ is incident. 
Since, by lemma~\ref{lem:vertexincidence}, Voronoi vertices are isolated points, 
	and two Voronoi edges $\Vor_{ij},\Vor_{kl}$ can only meet at a Voronoi vertex $\Vor_I$ (with $\{i,j,k,l\}\subset I$), 
	we can enumerate an alternating sequence of Voronoi edges and vertices 
	$\left[\dots,\Vor_{ij},\Vor_I,\Vor_{kl},\Vor_K,\dots\right]$ in clockwise order around $\partial\Vor_i$, 
	in which every edge is incident to the previous and next vertices in the sequence. 
Therefore, a Voronoi edge can only be incident to one Voronoi vertex if the sequence is $\left[\Vor_{ij},\Vor_I\right]$. 

If $\Vor_I$ is not the vertex at infinity, then we can show that the above is not possible 
	with an argument identical to the proof of lemma~\ref{lem:RSC} (figure~\ref{fig:RSC}). 
Note that $\Vor_{ij}\rightsquigarrow\Vor_I$ implies $\{i,j\}\subset I$, 
	and therefore all points in $\partial\Vor_i$ are equidistant to $s_i,s_j$. 
Let $\gamma\equiv\partial\Vor_i$, and consider the ray $r$ from $s_i$ in the direction $s_i-s_j$ 
	which, since $r$ is unbounded and $\Vor_i$ is bounded (since it is not incident to $\Vor_\infty$), 
	it must cross $\gamma$ at some point $q$. 
Since $q\in\gamma$, $q$ is equidistant to $s_i,s_j$, contradicting lemma~\ref{lem:midpoint}. 

If $\Vor_\infty$ is the vertex at infinity, then $\Vor_{ij}$ is not incident to any Voronoi vertex, and is unbounded. 
Therefore, $\Vor_{ij}$ does not cross any Voronoi edge, 
	or else $\Vor_{ij}$ would be incident to their intersection point (a Voronoi vertex). 
Recall from lemma~\ref{lem:midpoint} 
	that 	$\Vor_{ij}$ can never intersect the supporting line $L_{ij}$ of $s_i,s_j$ outside the segment $\overline{s_i,s_j}$. 
Let $L^i_{ij}$ ($L^j_{ij}$) be the ray starting at $s_i$ ($s_j$) with direction $s_i-s_j$ ($s_j-s_i$), 
	as shown in figure~\ref{fig:connectededges.b}. 
It can be easily shown that every point in $L^i_{ij}$ ($L^j_{ij}$) is strictly closer to $s_i$ ($s_j$) than to $s_j$ ($s_i$). 
Since, regardless of the choice of origin, 
	\emph{every} origin-centered circle $C(\sigma)$ of sufficiently large radius $\sigma$ intersects $L_{ij}$ at 
	exactly one point $p_i$ in $L^i_{ij}$, and one point $p_j$ in $L^j_{ij}$,
	the following holds. 
Let $L_{ij}$ divide $\mathbb{R}^2$ into two half spaces $H^{+}_{ij},H^{-}_{ij}$, 
	and let $C^{+}(\sigma)\equiv C(\sigma)\cap H^{+}_{ij}$ and
		    $C^{-}(\sigma)\equiv C(\sigma)\cap H^{-}_{ij}$. 
Since $p_i$ ($p_j$) is closer to $s_i$ ($s_j$) than to $s_j$ ($s_i$), 
	and $p_i,p_j$ are the endpoints of $C^{+}(\sigma), C^{-}(\sigma)$, 
	by the continuity of $D$, 
	there are points $p^{+}_{ij}\in C^{+}_{ij}$ and $p^{-}_{ij}\in C^{-}_{ij}$ equidistant to $s_i,s_j$. 
Since $\Vor_{ij}$ does not intersect any Voronoi element, then $p^{+}_{ij},p^{-}_{ij}$ are also closest to $s_i,s_j$. 
Because this holds for all sufficiently large $\sigma$, 
	then both $\Vor_{ij}\cap H^{+}_{ij}$ and $\Vor_{ij}\cap H^{-}_{ij}$ are unbounded, 
	contradicting lemma~\ref{lem:contrad}, 
	which states that every point $p^{-}_{ij}\in H^{-}_{ij}$ that is sufficiently far from the origin
	and equidistant to $s_i,s_j$ (and therefore its closest point in $L_{ij}$ lies in $\overline{s_i,s_j}$) 
	is closer to a site in $S\setminus\{s_i,s_j\}$ than to $s_i,s_j$. 
\end{proof}

\begin{figure}[htbp]
   \centering
   	\subfloat[]{\label{fig:closureedge.a}\includegraphics[width=2.7in]{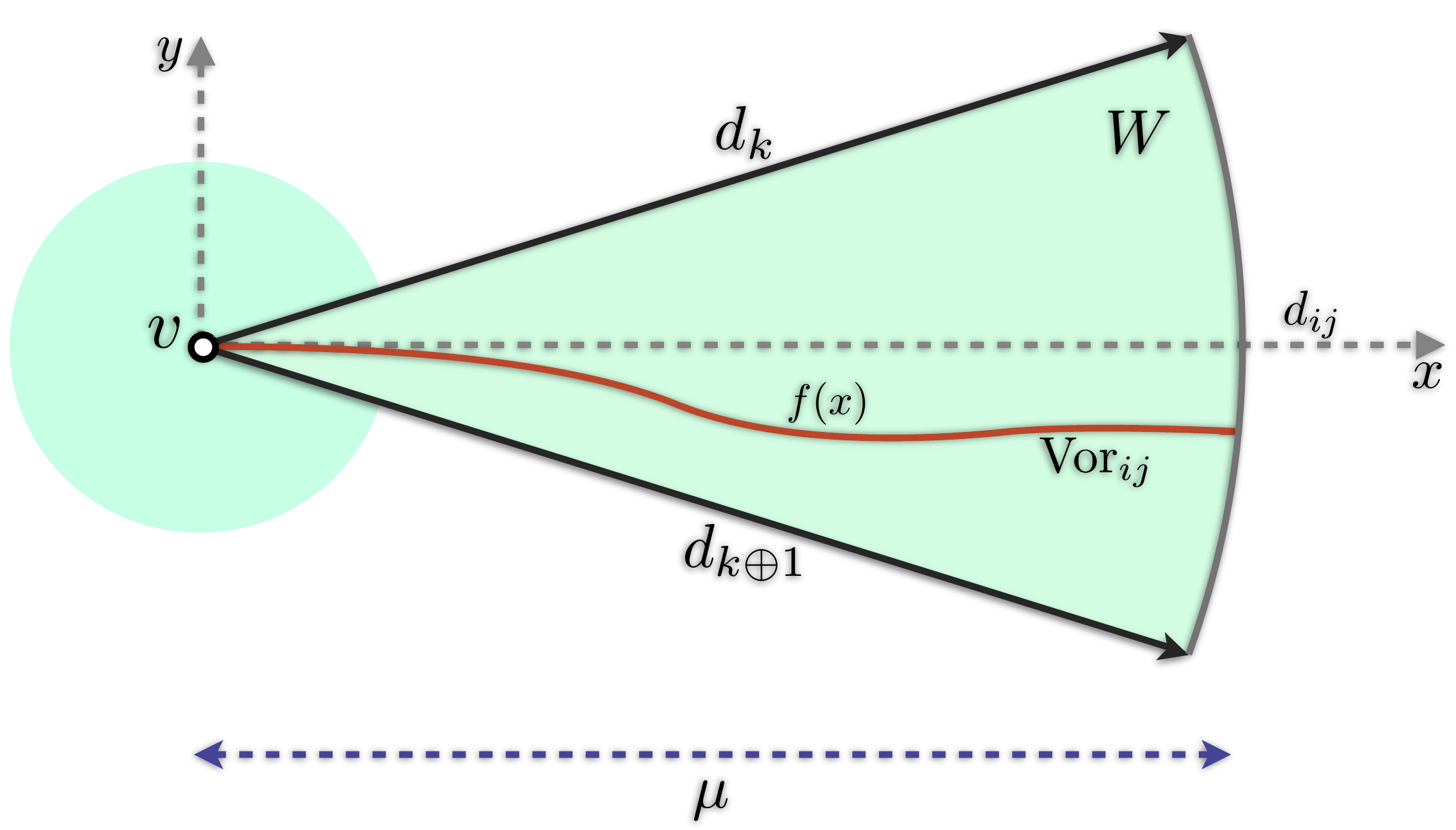}}\quad\quad
	\subfloat[]{\label{fig:closureedge.b}\includegraphics[width=2.7in]{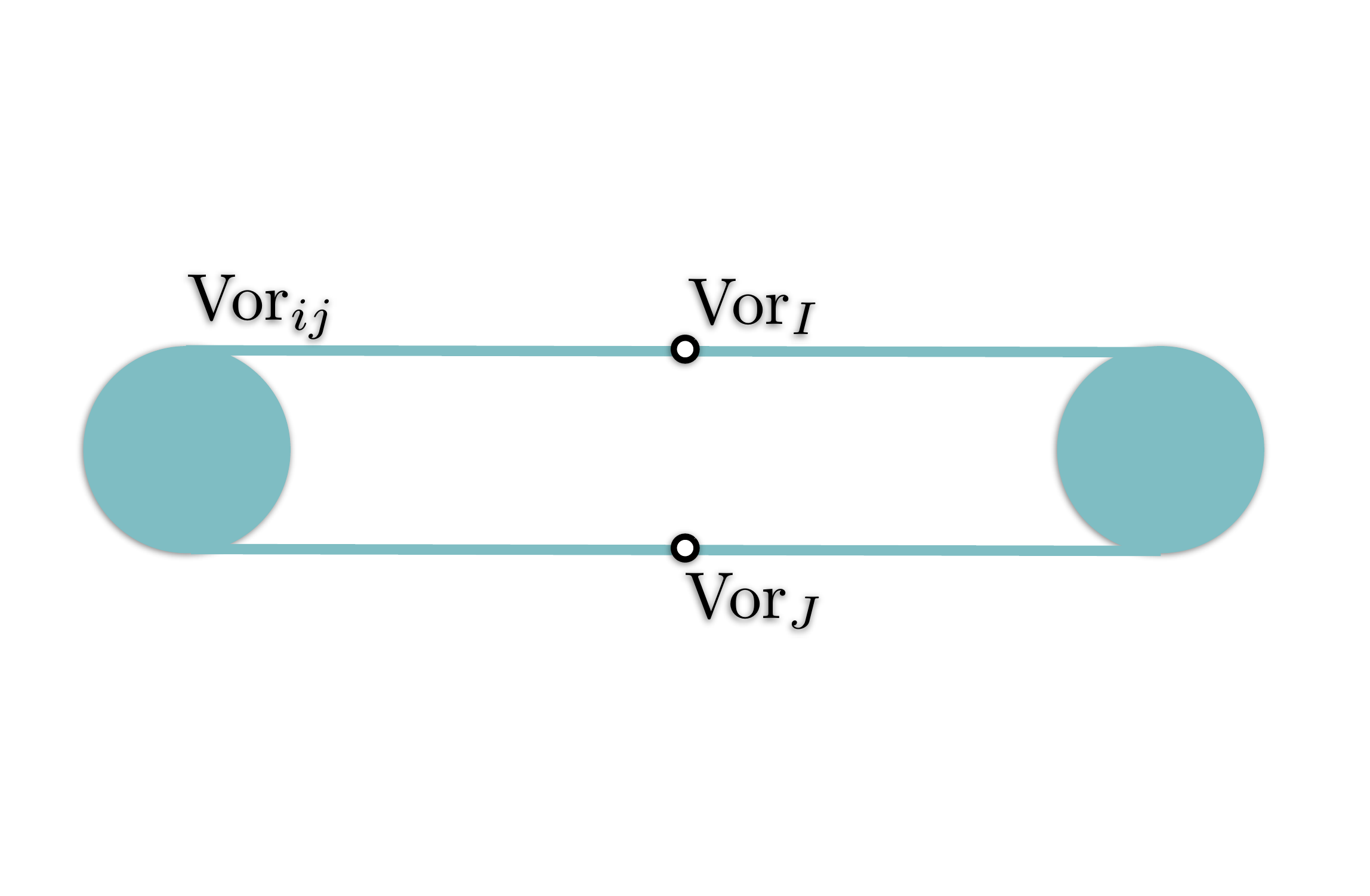}}
   \caption{From assumption~\ref{ass:EGA}, it follows that every Voronoi edge $\Vor_{ij}$, 
   			in the close vicinity of a Voronoi vertex can be written as the graph of a function $f(x)$ with an endpoint at $v$ (a). 
		Figure (b) shows a hypothetical Voronoi edge $\Vor_{ij}$ that breaks assumption~\ref{ass:EGA}, 
			for which lemma~\ref{lem:closureedge} does not hold. }
   \label{fig:closureedge}
\end{figure}

\begin{lemma}\label{lem:closureedge}
Let $\Vor_{ij}$ be a Voronoi edge. 
For every $r\in\Vor_{ij}$ and $q\in\overline{\Vor_{ij}}$ there is a simple path $\gamma:[0,1]\rightarrow\overline{\Vor_{ij}}$ 
	such that $\gamma(0)= r$, $\gamma(1)= q$, and $\gamma((0,1))\subset\Vor_{ij}$. 
%
%
%
%
\end{lemma}

\begin{proof}
\noindent{{\bf [Case $q\in\Vor_{ij}$]}}. 
Recall that connected components Voronoi edges are assumed to be path-connected (section~\ref{sec:assumptions}). 
Since Voronoi edges are connected (lemma~\ref{lem:connectededges}), they are path-connected. 
Therefore, if $q\in\Vor_{ij}$, there is always a path $\gamma:[0,1]\rightarrow\Vor_{ij}$ connecting $r,q$. \\

\noindent{{\bf [Case $q\in\partial\Vor_{ij}\setminus\Vor_{ij}$]}}. 
In this case, by property~\ref{prop:boundaryincidence}, $q$ must belong to a Voronoi element of higher order than $\Vor_{ij}$
	(a Voronoi vertex $\Vor_I$), to which $\Vor_{ij}$ is incident (with $\Vor_{ij}\rightsquigarrow\Vor_I$). 
Since, by lemma~\ref{lem:vertexincidence}, Voronoi vertices are composed of isolated points, 
	then $q$ is a connected component of $\Vor_I$ (possibly the vertex at infinity). 
Consider separately whether $q$ is the vertex at infinity. \\

\noindent{{\bf [Case $q\in\partial\Vor_{ij}\setminus\Vor_{ij}$ and $q$ is not the vertex at infinity]}}.
%
%
Recall that the proof of lemma~\ref{lem:vertexincidence}
	defines an ordering of $I=i_1,\dots,i_m$, and a set of associated direction vectors $d_1,\dots,d_m$. 
Let $g_k=\nabla_p\D{s_i}{p}\big|_{p=q}$, with $k=1,\dots,m$, 
	and let $d_{ij}$ be the unit vector orthogonal to $g_i-g_j$
	in the direction outgoing from $\CH\{-g_1,\dots,-g_m\}$
	(which exists since, by assumption~\ref{ass:EGA}, it is $g_i\neq g_j$). 
We assume, without loss of generality, that the coordinate representation of $d_{ij}$ 
	is $[\left(d_{ij}\right)_x, \left(d_{ij}\right)_y] = [1,0]$. 
Since $D\in\mathcal{C}^1$ and 
	$g_i\ne g_j$, by the implicit function theorem, 
	there is an open $L_2$ ball $B_2(q;\xi)$ around $q$ in which the implicit equation $ \D{s_i}{p} = \D{s_j}{p} $
	can be written as $y=f(x)$, with $f'(0)=0$, 
	as shown in figure~\ref{fig:closureedge.a}. 
	
Since $\Vor_{ij}$ is incident to $\Vor_I$ at $q$, 
	there is $k\in\{1,\dots,m\}$ such that $i=i_k, j=i_{k\oplus 1}$. 
Choose $0 < \mu < \xi$ to be sufficiently small for the conditions of the proof of lemma~\ref{lem:vertexincidence} to apply
	(in particular $\mu < \min\{\delta,\varepsilon,\xi\}$, as defined in the proof). 
Let $W$ be a circular wedge contained in the $L_2$ ball $B_2(q;\mu)$, 
	and bounded by the rays $q+\mu d_k$ and $q+\mu d_{k\oplus 1}$ which, 
	aside from $q$, only contains points strictly closer to $\{s_i,s_j\}$ than to all other sites.  
From the definition of $d_k,d_{k\oplus 1}$ it is clear that the segment $\overline{q, q+\mu d_{ij}}$ 
	is contained in $W$.

Since $\mu<\xi$, and inside $W$ all points with the exception of $q$ are closest only to $s_i,s_j$, 
	the implicit equation $ \D{s_i}{p} = \D{s_j}{p} $ represents the set of points in $W\cap \Vor_{ij}$. 
Since $ \D{s_i}{p} = \D{s_j}{p} $ can be written in coordinates as $y=f(x)$ inside $W$, 
	it is clear that, inside $W$, $\Vor_{ij}$ is a simple curve, 
	and that this is the only part of $\Vor_{ij}$ incident to $q$. 

Given $r\in\Vor_{ij}$, find any point $v\in W\cap\Vor_{ij}$ that is closer to $q$ than $r$. 
Because $v\in\Vor_{ij}$, there is a simple path $\gamma_1\subset\Vor_{ij}$ connecting $r$ to $v$ and, 
	because $v$ is in $W$, there is also a simple path 
	$\gamma_2\subset\Vor_{ij}$ from $v$ to $q$
	(part of the curve $y=f(x)$ of figure~\ref{fig:closureedge}). 
Finally, because $v$ is closer to $q$ than $r$ is, 
	the paths $\gamma_1$ and $\gamma_2$ do not cross, and therefore the 
	concatenation of $\gamma_1$ and $\gamma_2$ 
	meets the requirements of the lemma.


%

%
%

\noindent{{\bf [Case $q\in\partial\Vor_{ij}\setminus\Vor_{ij}$ and $q$ is the vertex at infinity]}}.
Since $\Vor_{ij}\rightsquigarrow\Vor_{\infty}$ then, by definition, $\Vor_{ij}$ is unbounded. 
Let $r_0\equiv r$ and, for each $k\in\mathbb{N}$, let $r_k\in\Vor_{ij}$ be at distance $\|r-r_k\|_2=k$. 
One can always find such a sequence of points because $\Vor_{ij}$ is unbounded and path-connected
	(if there is no $r_k\in\Vor_{ij}$ at distance $\|r-r_k\|_2=k$ then the circle with center at $r$ and radius $k$ would disconnect $\Vor_{ij}$). 
Let $\gamma'_k:[0,1]\rightarrow\Vor_{ij}$ be paths connecting $r_{k-1}$ to $r_k$, 
	and $\gamma':\mathbb{R}^+\rightarrow\Vor_{ij}$ be the concatenation of $\gamma'_1,\gamma'_2,\dots$, 
	where $\gamma'(k+t)\equiv\gamma'_k(t)$, with $k\in\mathbb{N}$ and $t\in[0,1]$. 

Define $\gamma:[0,1]\rightarrow\Vor_{ij}\cup\Vor_\infty$ as $\gamma(t) \equiv \gamma'(1/(1-t))$. 
Consider $\gamma$ on the Riemann sphere, transformed through a stereographic projection. 
Since $\gamma'$ is continuous and $\gamma$ has an accumulation point at the point at infinity (north pole on the sphere), 
	it is continuous on the sphere. 
If $\gamma$ is not simple, it can be appropriately cut and reparametrized until it is
	(i.e.\ by tracing the path and, upon arrival to a point $c$ where the path crosses itself, 
		cutting out the next portion up to the highest $t$ for which $\gamma(t)=c$, 
		and proceeding this way to the end of the path).

\end{proof}

Note that for lemma~\ref{lem:closureedge} to hold it is crucial that edges $\Vor_{ij}$ are incident to vertices $\Vor_I$ 
	as a curve arriving at $v\in\Vor_I$ from a single direction, 
	as illustrated in figure~\ref{fig:closureedge.a}.  
To see that assumption~\ref{ass:EGA} is required, 
	consider figure~\ref{fig:closureedge.b}, 
	which depicts an edge $\Vor_{ij}$ incident to two vertices $\Vor_I,\Vor_J$
	which do not satisfy assumption~\ref{ass:EGA}, 
	in which \emph{every} path connecting the two disks passes through either $\Vor_I$ or $\Vor_J$, 
	and therefore for which lemma~\ref{lem:closureedge} does not hold.

\begin{figure}[htbp]
   \centering
	\subfloat[]{\label{fig:tree.a}\includegraphics[width=3.0in]{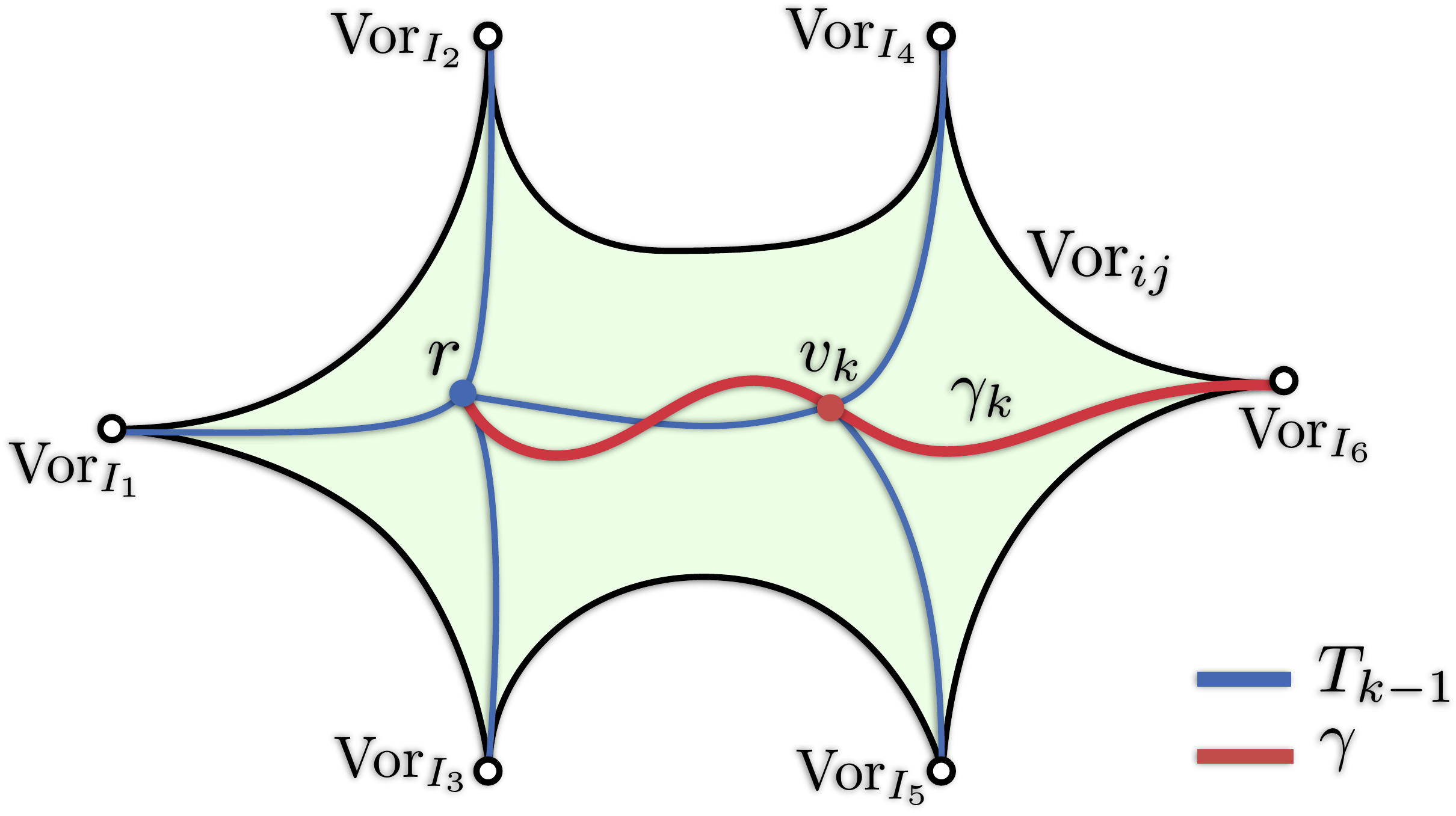}}
	\quad
	\subfloat[]{\label{fig:tree.b}\includegraphics[width=3.0in]{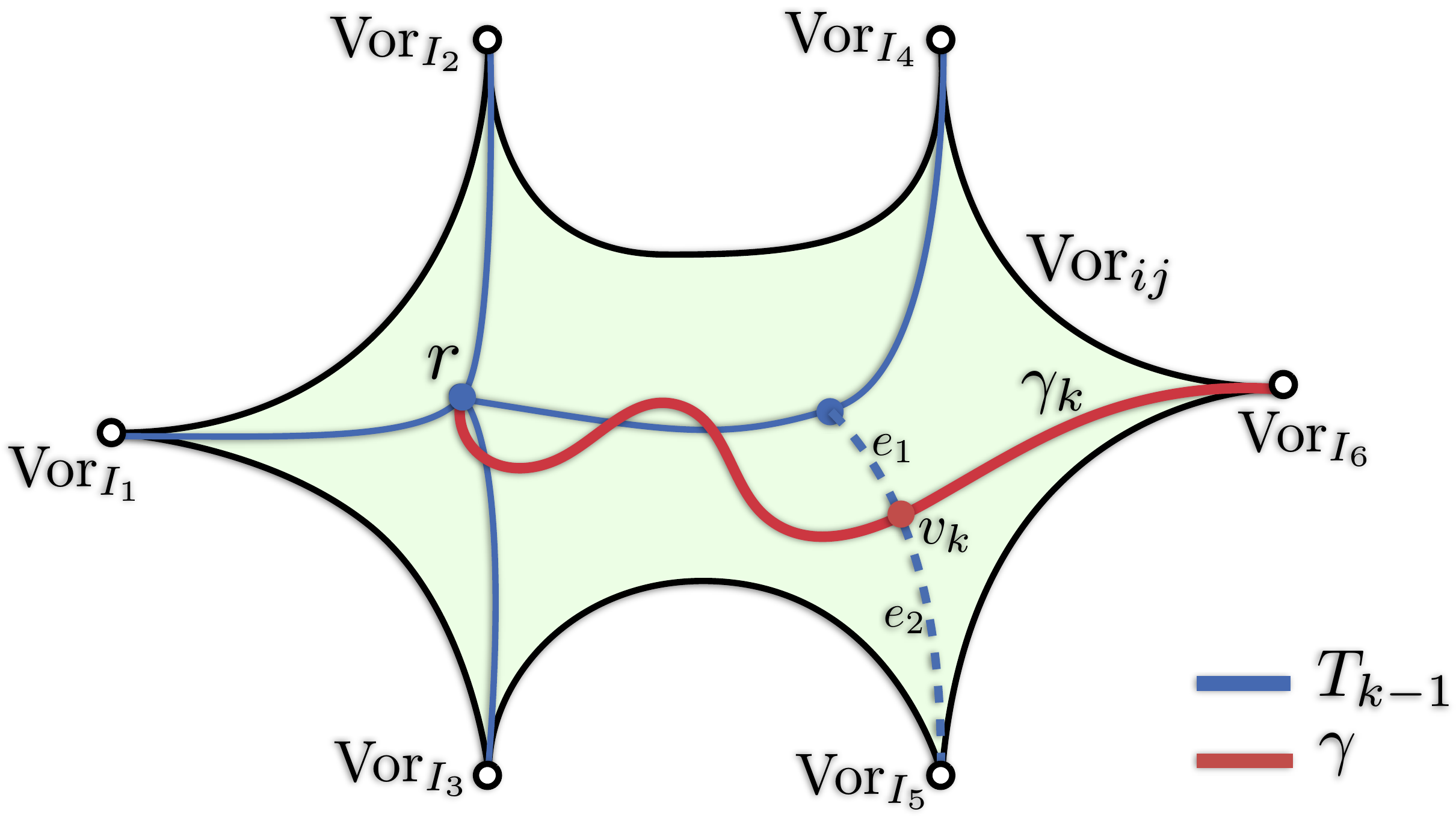}}
   \caption{The construction of a tree (blue) inside an edge $\Vor_{ij}$ (green region), with root $r$ and leafs 
   			at its incident Voronoi vertices $\Vor_{I_1},\dots,\Vor_{I_m}$.}
   \label{fig:tree}
\end{figure}

\begin{lemma}\label{lem:tree}
In an orphan-free diagram, 
	for every Voronoi edge $\Vor_{ij}$ that is incident to Voronoi vertices $\Vor_{I_1},\dots,\Vor_{I_m}$, 
	there is an embedded tree graph in $\overline{\Vor_{ij}}$ whose leafs are $\Vor_{I_1},\dots,\Vor_{I_m}$.
\end{lemma}

\begin{proof}

Unless otherwise specified, we assume in this proof that all paths are simple, contained in $\overline{\Vor_{ij}}$, 
	parametrized over the unit interval $[0,1]$, 
	and that, using lemma~\ref{lem:closureedge},
	there is a path connecting any two points in $\overline{\Vor_{ij}}$ that does not intersect a Voronoi vertex
	(expect perhaps at the endpoints). 
We use throughout the fact that Voronoi edges are path connected 
	(lemma~\ref{lem:SCedges} and section~\ref{sec:assumptions}).

If $m=1$, pick a point $r\in\Vor_{ij}$ as root and, 
	using lemma~\ref{lem:regionpath},
	 consider a simple path $\gamma_{r,1}\subset\overline{\Vor_{ij}}$ connecting $r$ to $\Vor_{I_1}$, 
	then the tree with vertex set $V=\{r,\Vor_{I_1}\}$, and edge set $E=\{\gamma_{r,1}\}$ meets the requirements of the lemma.

%
%

For each $k\ge 2$, assume that there is an embedded tree graph $T_{k-1}\subset\overline{\Vor_{ij}}$ 
	with $\Vor_{I_1},\dots,\Vor_{I_{k-1}}$ as leafs. 
We construct a new embedded tree $T_k$ as follows (figure~\ref{fig:tree}). 
Let $r\in\Vor_{ij}$ be the root of $T_{k-1}$, 
	and let $\gamma$ be a simple path connecting $r$ to $\Vor_{I_k}$ which, 
	making use of lemma~\ref{lem:closureedge}, is chosen such that it does not intersect 
	any Voronoi vertex (other than the final endpoint). 
Let 
	\[ t_k \equiv \max \{t\in [0,1] ~:~ \gamma(t)\in T_{k-1} \}, \]
	which always exists because $T_{k-1}$ is closed 
	and $\gamma(0)=r\in T_{k-1}$. 
Let $v_k\equiv \gamma(t_k)$ be the ``last" point along $\gamma$ that belongs to $T_{k-1}$.
Because $\gamma(1)=\Vor_{I_k}\notin T_{k-1}$ 
	then it must be $t_k <1$. 
Additionally, $v_k$ cannot be a Voronoi vertex, 
	since $\gamma$ doesn't intersect Voronoi vertices except at the final endpoint $\Vor_{I_k}$.

Let $\gamma_k$ be the path $\{\gamma(t) ~:~ t\in [t_k,1]\}$, 
	that is, the part of $\gamma$ from $v_k$ to $\Vor_{I_k}$. 
We construct a new tree graph $T_k\subset\overline{\Vor_{ij}}$ as follows. 
Begin by setting $T_k$ equal to $T_{k-1}$. 
We then insert a new vertex $\Vor_{I_k}$ into $T_k$. 
Next, we proceed differently depending on whether $v_k\in T_{k-1}$ 
	is a vertex, or it belongs to an edge of $T_{k-1}$
(note that, since $v_k$ is not a Voronoi vertex, it cannot be a leaf vertex of $T_{k-1}$). 

If $v_k$ is an internal vertex of $T_{k-1}$, 
	as in figure~\ref{fig:tree.a},
	then we add a new edge $\gamma_k$ to $T_k$ connecting vertices $v_k$ and $\Vor_{I_k}$. 
Since, by construction, $\gamma_k$ does not cross any edge in $T_k$, the tree graph remains embedded.

If, on the other hand, $v_k$ belongs to an edge $e$ of $T_{k-1}$ connecting vertices $v_1,v_2$, 
	as shown in figure~\ref{fig:tree.b}, 
	then:
	\begin{enumerate}
	\item we insert a new (internal) vertex $v_k$ into $T_k$;
	\item we split $e$ into two edges: 
		$e_1$ and $e_2$, connecting $v_1,v_k$, and $v_k,v_2$, respectively; 
	\item we insert a new edge $\gamma_k$ connecting vertices $v_k$ and $\Vor_{I_k}$. 
	\end{enumerate}
Note that 
	the edge $e$ is split into two edges that represent the same set of points, 
	and therefore, since $\gamma_k$ didn't cross any edges of $T_{k-1}$, 
	then $\gamma_k$ does not cross any edge of $T_k$. 
Hence, since $T_{k-1}$ is an embedded tree graph, 
	the new tree $T_k$ is also embedded
	 and has $\Vor_{I_1},\dots,\Vor_{I_k}$ as leafs. 

The lemma follows by induction on $m$.  
\end{proof}

\begin{figure}[!h]
   \centering
\includegraphics[width=3.5in]{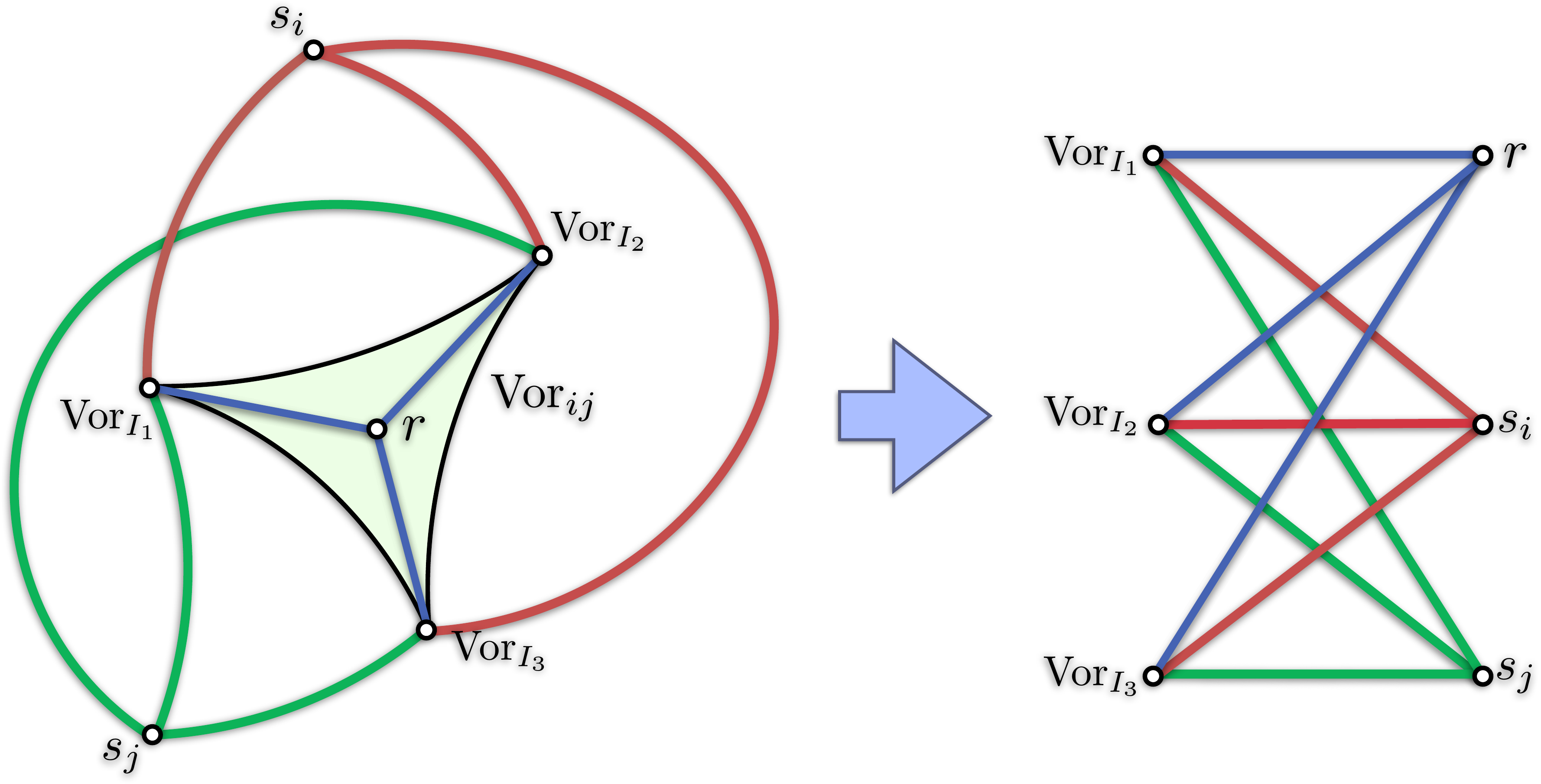}
   \caption{By assuming that a Voronoi edge $\Vor_{ij}$ is incident to three Voronoi vertices $\Vor_{I_1},\Vor_{I_2},\Vor_{I_3}$, 
   		we can construct a planar embedding of the non-planar graph $K_{3,3}$, a contradiction. 
		The more general figure~\ref{fig:planarity.a} 
			 further illustrates the proof of lemma~\ref{lem:val_le2}. }
   \label{fig:K33}
\end{figure}

The final lemma of this section can be used in conjunction with lemma~\ref{lem:val_ge2} 
	to establish that Voronoi edges are incident to exactly two Voronoi vertices. 
We sketch here the argument that shows that a Voronoi edge $\Vor_{ij}$ cannot 
	be incident to three vertices $\Vor_{I_1},\Vor_{I_2},\Vor_{I_3}$ (figure~\ref{fig:K33}). 
The general case in the proof of lemma~\ref{lem:val_le2} follows a similar argument. 
We first use lemma~\ref{lem:tree} to build a tree inside $\Vor_{ij}$ with leafs at $\Vor_{I_1},\Vor_{I_2},\Vor_{I_3}$, 
	and show that it can be collapsed into a star-graph with a vertex $r\in\Vor_{ij}$, and non-crossing edges 
		$(r,\Vor_{I_1}), (r,\Vor_{I_2}), (r,\Vor_{I_3})$, as shown in the figure. 
The incidence rules of lemma~\ref{lem:vertexincidence}, as well as lemma~\ref{lem:regionpath} allows us 
	to construct six non-crossing edges from $s_i$ and $s_j$, to $\Vor_{I_1},\Vor_{I_2},\Vor_{I_3}$, respectively. 
We have just constructed an embedding of a graph which can be easily shown to be the non-planar graph $K_{3,3}$, 
	thereby reaching a contradiction.

\begin{figure}[!h]
   \centering
	\subfloat[]{\label{fig:planarity.a}\includegraphics[width=2.8in]{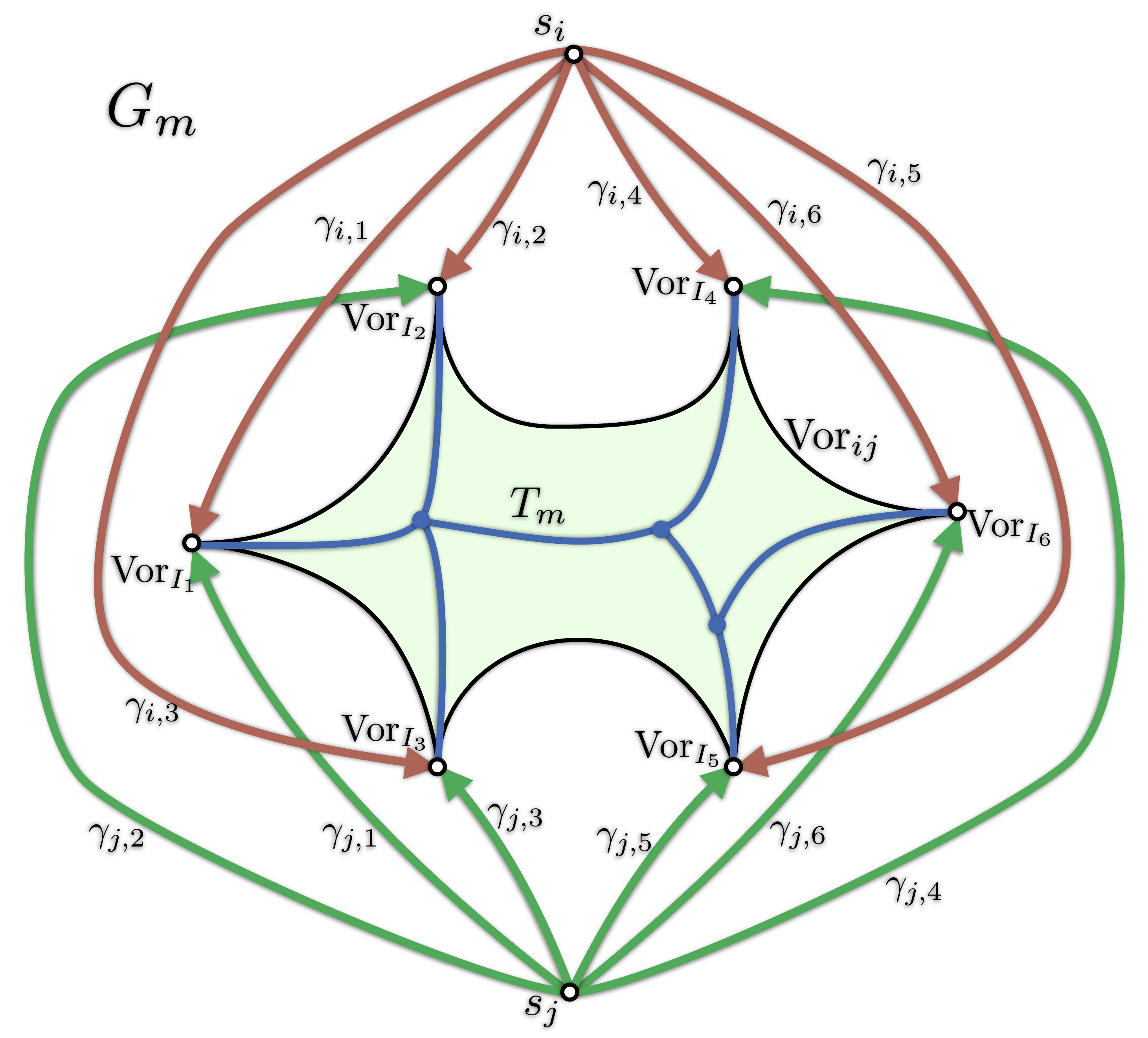}}
	\quad\quad
	\subfloat[]{\label{fig:planarity.b}\includegraphics[width=2.8in]{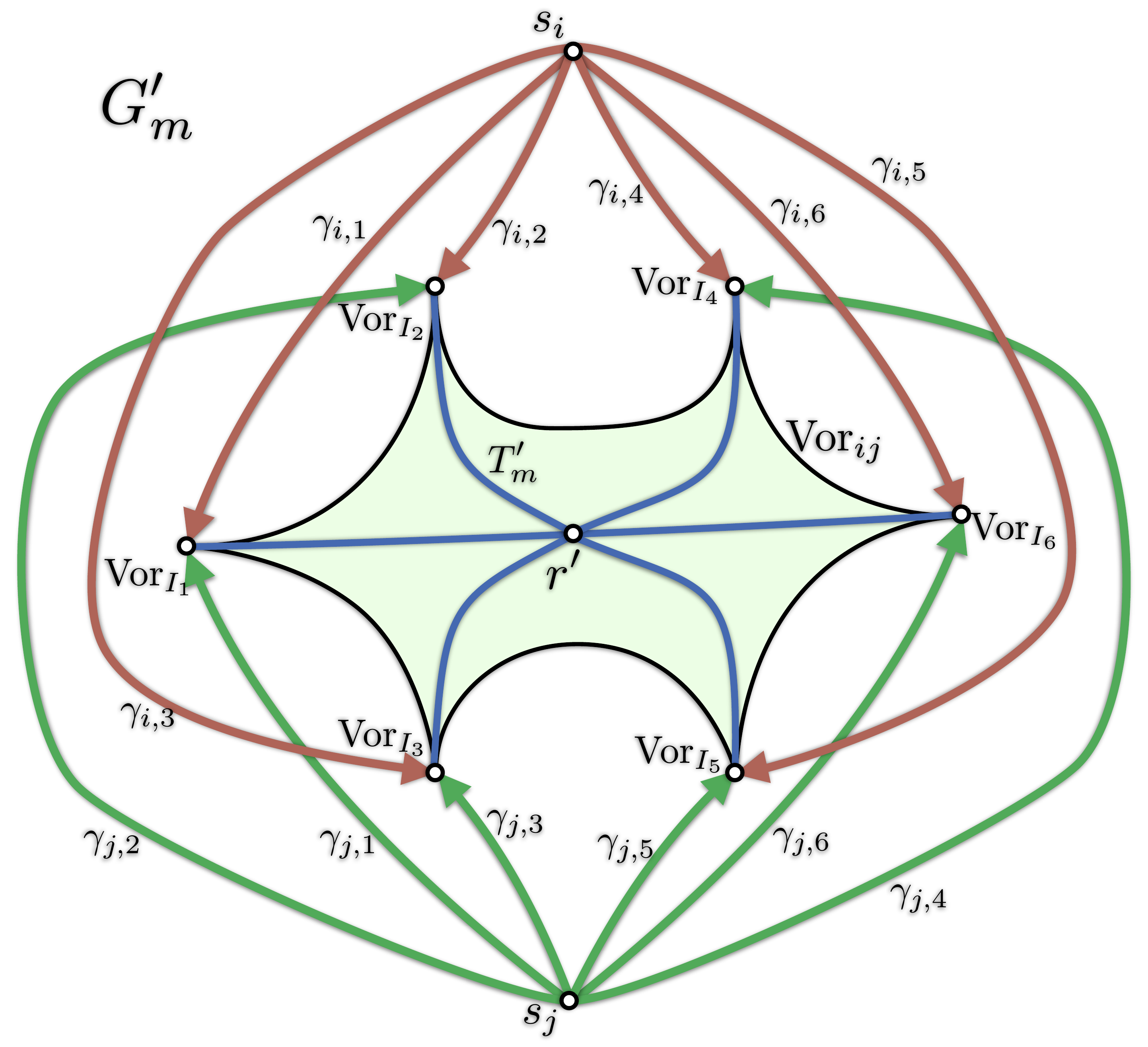}}
   \caption{Every Voronoi edge $\Vor_{ij}$ (green region) incident to $m$ Voronoi vertices 
   			allows the construction of an embedded planar graph $G_m$ connecting 
   			a tree $T_m$ inside $\Vor_{ij}$, to the sites $s_i,s_j$ (a). 
		This graph has a minor $G'_m$ obtained from $G_m$ by contracting edges of $T_m$. 
		$G'_m$ can be shown not to be planar for $m>2$, and therefore Voronoi edges are incident to no more than two Voronoi vertices. 
		 }
   \label{fig:planarity}
\end{figure}

\begin{lemma}\label{lem:val_le2}
Voronoi edges of an orphan-free diagram are incident to no more than two Voronoi vertices.
\end{lemma}

\begin{proof}
	Let $\Vor_{ij}$ be a Voronoi edge incident to Voronoi vertices $\Vor_{I_1},\dots,\Vor_{I_m}$. 
%
%
%
	Since $\Vor_{ij} \rightsquigarrow \Vor_{I_1},\dots,\Vor_{I_m}$, 
		and Voronoi vertices are of higher order ($|I_k| > 2$) than Voronoi edges, 
		by the definition of incidence (definition~\ref{def:incidence}), it is $\{i,j\} \subset \Vor_{I_k}$, with $k=1,\dots,m$. 
	We prove the lemma on the sphere $\mathbb{S}^2$, where any of the Voronoi vertices may be the vertex at infinity. 
	Note also that some of the sets $I_k$ with $k=1,\dots,m$ may be equal, 
	since Voronoi vertices have not yet been shown to be connected.

By lemma~\ref{lem:vertexincidence}, 
	the vertices $\Vor_{I_k}$ are isolated points
	(possibly the point at infinity), 
	and $\Vor_i,\Vor_j\rightsquigarrow \Vor_{I_1},\Vor_{I_2},\dots,\Vor_{I_m}$. 
We begin by assuming that $m>2$, and build an embedded planar graph $G_m$ (figure~\ref{fig:planarity.a}). 
We then show that $G_m$ can only be planar if $m\le 2$, reaching a contradiction.

By lemma~\ref{lem:tree}, there is an embedded tree graph $T_m\subset\overline{\Vor_{ij}}$
	with $\Vor_{I_1},\dots,\Vor_{I_m}$ as leafs. 
We begin by setting $G_m$ equal to $T_m$. 
We then insert the vertices $s_i$ and $s_j$ in $G_m$ (as shown in figure~\ref{fig:planarity.a}). 
Since $\Vor_i,\Vor_j \rightsquigarrow \Vor_{I_1},\Vor_{I_2},\dots,\Vor_{I_m}$, 
	by lemma~\ref{lem:regionpath}, 
	there are non-crossing paths $\gamma_{i,k}\subset\overline{\Vor_i}$, with $k=1,\dots,m$, 
	connecting $s_i$ to $\Vor_{I_k}$ 
	and non-crossing paths $\gamma_{j,k}\subset\overline{\Vor_j}$, with $k=1,\dots,m$, 
	connecting $s_j$ to $\Vor_{I_k}$. 
We insert the above paths $\gamma_{i,k},\gamma_{j,k}$, $k=1,\dots,m$, as edges of $G_m$. 
Aside from all paths $\gamma_{i,k}$ ($\gamma_{j,k}$) only crossing
	 at their starting point,
	all paths $\gamma_{i,k}$ ($\gamma_{j,k}$) are, by lemma~\ref{lem:regionpath},
	 contained (except for their final endpoint) 
	in the interior of $\Vor_i$ ($\Vor_j$), 
	and therefore they can only cross an edge of $T_m\subset\overline{\Vor_{ij}}$ at an endpoint. 
$G_m$ is therefore embedded in $\mathbb{S}^2$, and so it is a planar graph.

Recall that the minors of a graph are obtained by erasing vertices, erasing edges, or contracting edges, 
	and that minors of planar graphs are themselves planar~\cite[p.\ 269]{bondy2008graph}. 
We now construct an appropriate minor $G'_m$ of the planar graph $G_m$, 
	shown in figure~\ref{fig:planarity.b}, 
	and prove that it is non-planar whenever $m > 2$, creating a contradiction. 

Clearly, every tree $T_m$ satisfying the conditions of lemma~\ref{lem:tree} has a minor $T'_m$ 
	directly connecting the root to each leaf $\Vor_{I_k}$, $k=1,\dots,m$ (see figure~\ref{fig:planarity.b}), 
	which is obtained by successively contracting every edge of $T_m$ that connects two internal vertices. 
We apply the same sequence of edge contractions to obtain $G'_m$ from $G_m$, as shown in figure~\ref{fig:planarity}.

Let $r'$ be the root of $T'_m$, and $\gamma_{r',k}$ be edges from $r'$ to $\Vor_{I_k}$, with $k=1,\dots,m$. 
The minor $G'_m$ has vertex set
	\[ V=\{s_i,s_j,r', ~\Vor_{I_1},\dots,\Vor_{I_m}\},\]
and edge set
	\[ E=\{ \gamma_{i,1},\dots,\gamma_{i,m}, ~ \gamma_{j,1},\dots,\gamma_{j,m}, ~ \gamma_{r',1},\dots,\gamma_{r',m}\},\]
and therefore $G'_m$ has $v=m+3$ vertices and $e=3m$ edges.  
Since (as is easily verified) every cycle in $G'_m$ has length four or more, and $G'_m$ is planar, 
	then it holds $2e \ge 4f$, where $f$ is the number of faces. 
Using Euler's identity for planar graphs, $v-e+f=2$~\cite{bondy2008graph}, 
	and the fact that $2e \ge 4f$, $v=m+3$, and $e=3m$, it follows that $m \le 2$, 
	and therefore $G'_m$ is not planar whenever $m> 2$
(for instance, $G'_3$ is the utility graph $K_{3,3}$). 
%

Since $m> 2$ leads to a contradiction, it follows that every Voronoi edge is incident to at most two Voronoi vertices.
\end{proof}

\subsection{Primal Voronoi graph and dual Delaunay triangulation}\label{primaldual}

We use the results in this section to construct a graph from the 
	incidence relations of an orphan-free Voronoi diagram, 
	and dualize it into a planar embedded graph. 

Let the \emph{primal Voronoi graph} $\tilde{P}=(\tilde{P}_V, \tilde{P}_E)$
	of an orphan-free Voronoi diagram be defined as follows. 
The vertices $\tilde{P}_V$ are the connected components of Voronoi vertices. 
Since, by lemma~\ref{lem:vertexincidence}, 
	Voronoi vertices are composed of isolated points, then $\tilde{P}_V$ is a collection of isolated points. 
By lemmas~\ref{lem:val_ge2} and~\ref{lem:val_le2}, 
	Voronoi edges that are incident to some Voronoi vertex are incident to exactly two Voronoi vertices. 
For each Voronoi edge $\Vor_{ij}$ incident to some Voronoi vertex, 
	we include in $\tilde{P}_E$ an edge connecting the vertices in $\tilde{P}_V$ 
	corresponding to the connected components of Voronoi vertices that $\Vor_{ij}$ is incident to. 
By lemma~\ref{lem:closureedge}, for each such Voronoi edge $\Vor_{ij}$ there is a simple path in $\Vor_{ij}$
	connecting the two Voronoi vertices incident to $\Vor_{ij}$, and therefore $\tilde{P}$ is an embedded planar graph.

\begin{figure}[htbp]
   \centering
	\subfloat[]{\label{fig:planar.a}\includegraphics[width=2.3in]{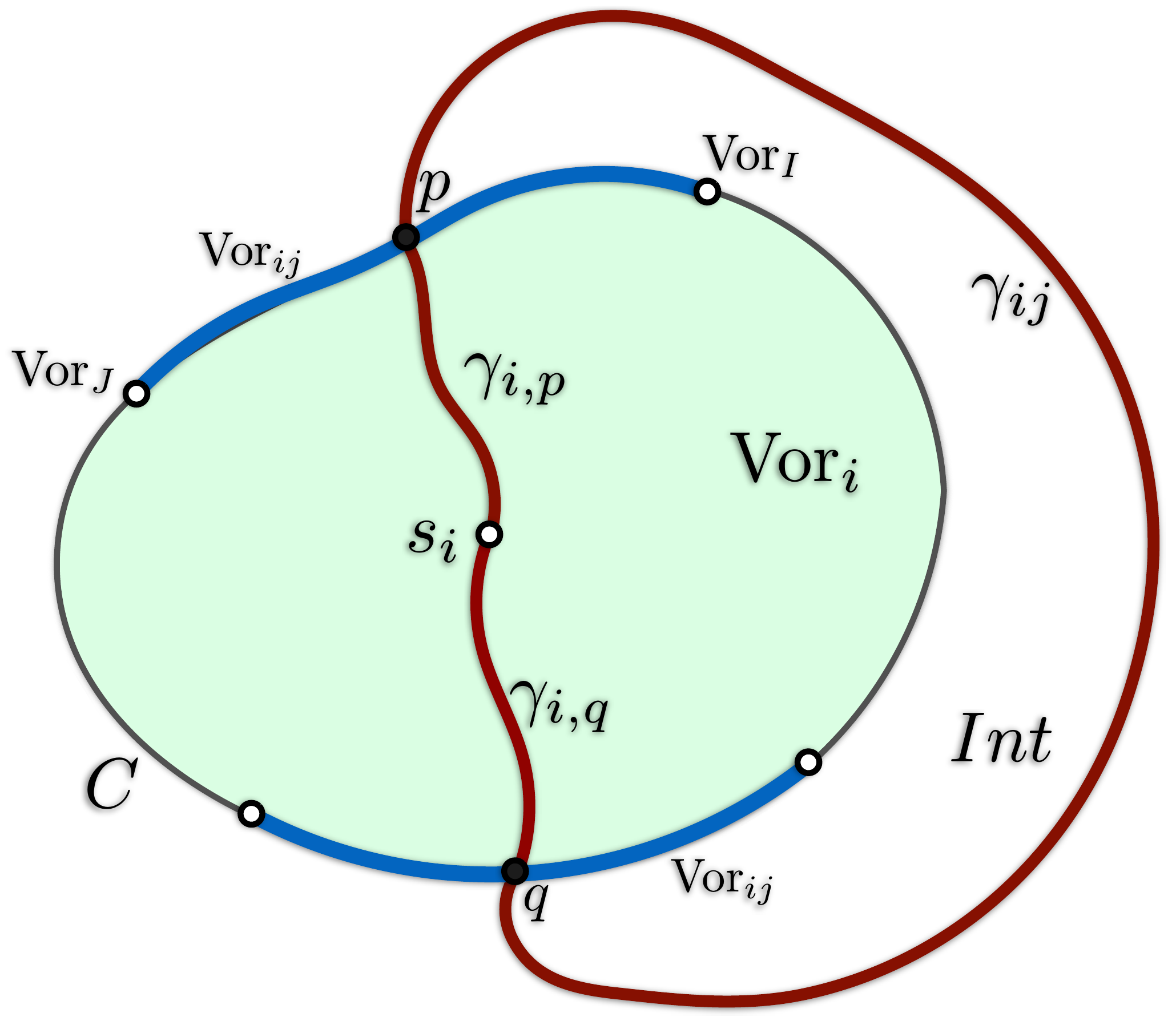}}
	\quad\quad
	\subfloat[]{\label{fig:planar.b}\includegraphics[width=2.3in]{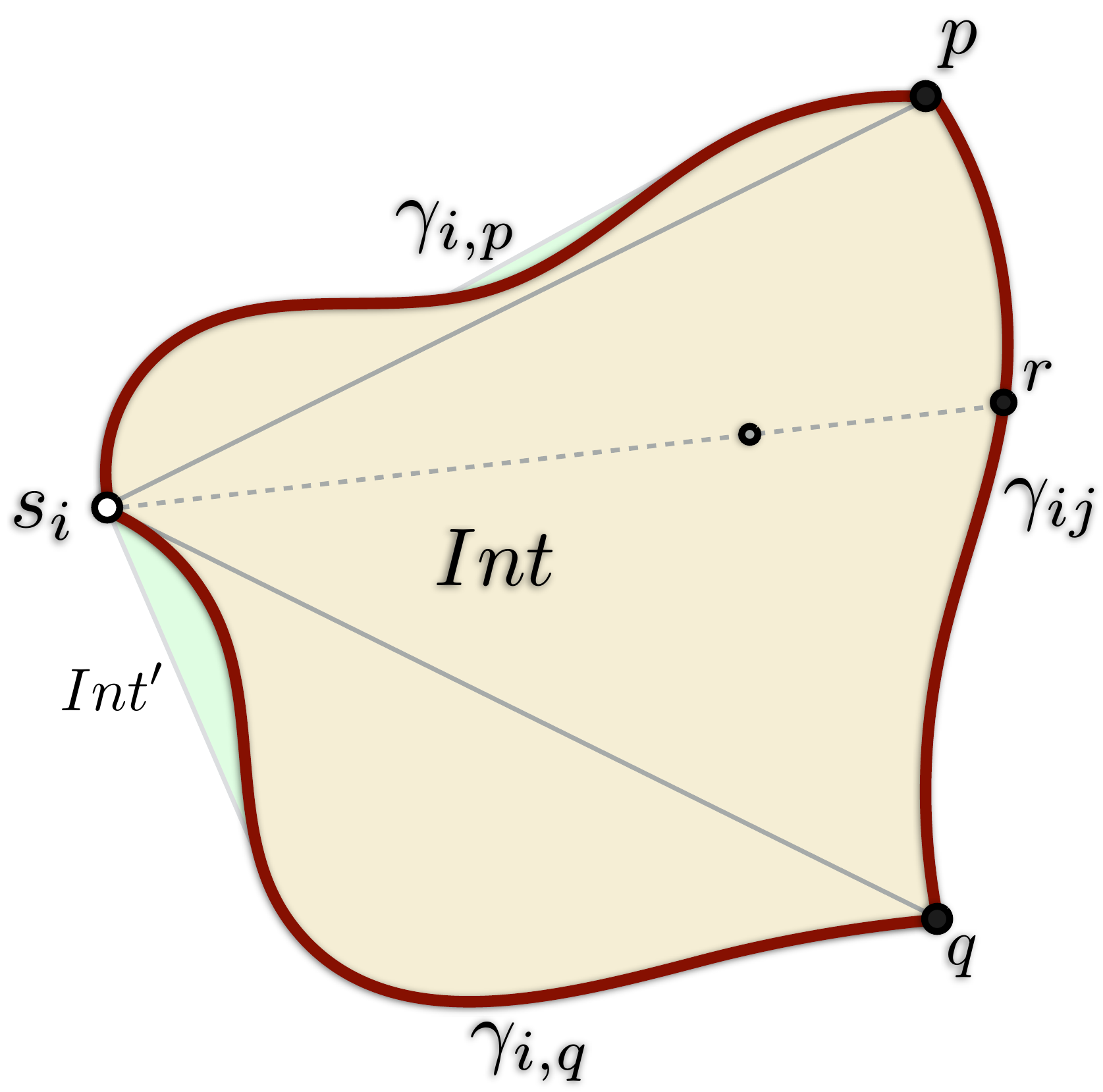}}
   \caption{Diagrams used in the proof of theorem~\ref{th:simpleplanar}. }
\end{figure}

\vspace*{0.1in}\noindent{\bf Theorem~\ref{th:simpleplanar}}. \emph{
Let $\tilde{G}=(\tilde{V}, \tilde{E})$ be the dual of the primal Voronoi graph corresponding to an orphan-free Voronoi diagram, 
	then $\tilde{G}$ is a simple, connected, planar graph. 
}
\begin{proof}
The dual graph $\tilde{G}$ is constructed by dualizing $\tilde{P}$ 
	and using the natural embedding described in~\cite[p.\ 252]{bondy2008graph}, 
	in which dual vertices are placed inside primal faces (at the sites in this case), 
	and dual edges cross once their corresponding primal edges. 
From this construction, $\tilde{G}$ is an embedded planar graph~\cite[p.\ 252]{bondy2008graph}, 
	and is connected by virtue of being the dual of a planar graph~\cite[p.\ 253]{bondy2008graph}. 

We show that $\tilde{G}$ is simple (edges have multiplicity one, and there are no loops: edges incident to the same vertex). 
Edges of $\tilde{G}$ are one-to-one with edges of $\tilde{P}$. 
In turn, edges of $\tilde{P}$ correspond to Voronoi edges, and these are, 
	by lemma~\ref{lem:connectededges}, 
	connected. 
Therefore the edges of $\tilde{G}$ have multiplicity one. 

Since loops and cut edges (those whose removal disconnects the graph) 
	are duals of each other~\cite[p.\ 252]{bondy2008graph}, 
	we now show that $\tilde{P}$ has no cut edges, and therefore $\tilde{G}$ has no loops. 

By~\cite[p.\ 86]{bondy2008graph}, 
	an edge of $\tilde{P}$ is a cut edge iff it belongs to no cycle of $\tilde{P}$. 
To every edge of $\tilde{P}$ corresponds an Voronoi edge $\Vor_{ij}$ that is incident to two Voronoi vertices. 
By lemma~\ref{lem:vertexincidence}, 
	$\Vor_{ij}$ is incident to at least one Voronoi region $\Vor_i$. 
We next show that the Voronoi elements in the boundary of every Voronoi region $\Vor_i$ form a cycle, 
	and therefore $\Vor_{ij}$ belongs to a cycle, so it cannot be a cut edge. 

Clearly, the boundary $\partial \Vor_i$ of $\Vor_i$ is composed of Voronoi edges and Voronoi vertices, 
	since $\Vor_i\rightsquigarrow\Vor_j$ is not possible because $\{i\}\not\subset\{j\}$ (see section~\ref{sec:incidence}). 
Let $C=\left[\Vor_{i,j_1},\Vor_{i,j_1,j_2}, \Vor_{i,j_2}, \dots\right]$ be the sequence of elements around the boundary of $\Vor_i$, 
	with $\Vor_{ij}\in C$. 
	We show that $C$ is a cycle. 

\vspace*{0.05in}\noindent{\bf [$C$ has no repeated Voronoi vertices]}. 
By the assumption of section~\ref{sec:assumptions}, 
	Voronoi regions have boundaries that are simple closed curves (in $\mathbb{S}^2$). 
Note that, because vertices are isolated points, there are no repeated vertices in $C$ since 
	the boundary of $\Vor_i$ is a simple curve. 

\vspace*{0.05in}\noindent{\bf [$C$ has no repeated Voronoi edges]}. \\
Let $\Vor_{ij}$ appear twice in $C$ as $\left[\dots, \Vor_I, \Vor_{ij}, \Vor_J, \dots, \Vor_{ij}, \dots\right]$, 
where $\Vor_I,\Vor_J$ are Voronoi vertices, as in figure~\ref{fig:planar.a}. 
Let $p,q$ be two points in each of the two common boundaries between $\Vor_i$ and $\Vor_{ij}$. 
By lemma~\ref{lem:regionpath}, there are simple paths $\gamma_{i,p},\gamma_{i,q}\subset\Vor_i$ from $s_i$ to $p,q$, respectively, 
	which only meet at the initial endpoint (figure~\ref{fig:planar.a}). 
Since $\Vor_{ij}$ is simply connected, we can consider a simple path $\gamma_{ij}\subset\Vor_{ij}$ connecting $p,q$. 
Let $\gamma$ be the simple closed path obtained by concatenating $\gamma_{i,p},\gamma_{ij},\gamma_{i,q}$ 
	which, by the Jordan curve theorem divides the plane into a bounded region $Int$, and an unbounded region. 
Since it must be $\Vor_I\in Int$ or $\Vor_J\in Int$, 
	assume without loss of generality that $\Vor_I\in Int$, and note that it cannot be $\Vor_I=\Vor_\infty\in Int$, 
	since $Int$ is bounded. 
We show that $\Vor_I\in Int$ is not possible, 
	and therefore that $C$ 
	has no repeated elements.

Let $\Vor_I\in Int$ and let $k\in I$ be $k\ne i,j$, which always exists because $|I|\ge 3$. 
By lemma~\ref{lem:vertexincidence}, there is a point $p\in\Vor_k\cap Int$. 
Since $\Vor_k$ is path connected, and the boundary of $Int$ is $\gamma\subset\Vor_{i}\cup\Vor_{ij}$, then $\Vor_k\subset Int$, 
	and therefore $s_k\in Int$. 
We show that $Int$ cannot contain any sites other than $s_i$, reaching a contradiction.  

Recall that the boundary $\gamma$ of $Int$ is the concatenation of 
	$\gamma_{i,p}\subset\Vor_i$, $\gamma_{ij}\subset\Vor_{ij}$, and $\gamma_{i,q}\subset\Vor_i$, 
	and that $s_i\in\gamma$, as in figure~\ref{fig:planar.b}. 
Let $Int'$ be the union of segments from $s_i$ to every point in $\gamma$:
\[ Int'\equiv \displaystyle{\left(\cup_{r\in\gamma_{i,p}} \overline{s_i, r}\right) \cup
				     \left(\cup_{r\in\gamma_{ij}} \overline{s_i, r}\right) \cup
				     \left(\cup_{r\in\gamma_{i,q}}\overline{s_i,r}\right)}. \]
Since it is clearly $Int\subset Int'$, it suffices to show that $Int'$ does not contain any site $s_k$ different from $s_i$.
Every segment of the form $\overline{s_i,r}$ 
	with $r\in\gamma_{i,p}\subset\Vor_i$ or $r\in\gamma_{i,q}\subset\Vor_i$ cannot contain a site $s_k$ 
	or else, by the convexity of $D$, $r$ would be closer to $s_k$ than to $s_i$. 
Similarly, every segment of the form $\overline{s_i,r}$ with 
	$r\in\gamma_{ij}$ cannot contain a site $s_k$, or else by the convexity of $D$, $r$ would be closer to $s_k$ than to $s_i,s_j$.

Since every Voronoi edge $\Vor_{ij}$ is part of a cycle, it cannot be a cut edge, and therefore its dual has no loops. 
\end{proof}


\section{Embeddability of the Delaunay triangulation}\label{sec:dual}


Let $\tilde{G}=(\tilde{V}, \tilde{E})$ be the dual of the primal Voronoi graph corresponding to an orphan-free Voronoi diagram, 
	as defined in section~\ref{sec:planar}. 
By theorem~\ref{th:simpleplanar}, $\tilde{G}$ is simple and planar with vertices at the sites. 
Let $G=(\Sites,E,F)$ be the planar graph obtained by replacing curved edges by straight segments. 
Recall from section~\ref{sec:planar} that, 
	while Voronoi regions and edges are connected, Voronoi vertices may have multiple connected components, 
	and therefore $G$ can have duplicate faces in $F$. 
We only show after this section that faces have multiplicity one by virtue of $G$ being embedded.

\vspace*{0.08in}\noindent{\bf Faces with more than three vertices}. 
Every face $f\in F$ is dual to a Voronoi element $\Vor_I$ of order $|I|=k\ge 3$, 
	to which corresponds (proposition~\ref{prop:ECB}) a convex ball $B(c;\rho)$, with $c\in\Vor_I$, 
	that circumscribes the sites $(s_i)_{i\in I}$ incident to $f$. 
Due to the planarity of $G$, we can assume the sites $(s_i)_{i\in I}$ to be ordered around $f$. 
In order to find whether a point $p\in\mathbb{R}^2$ belongs to $f$, 
	we simply triangulate $f$ in a fan arrangement: 
		$\tau_1=\{s_{i_1},s_{i_2},s_{i_k}\}; \tau_2=\{s_{i_2},s_{i_3},s_{i_k}\}; \dots$, 
	and consider that $p\in f$ iff it lies in any of the resulting $\tau_j$. 
Note that this arrangement does not interfere with the original edges in $E$ (other than creating new ones), 
	all new edges are incident to two faces (they are not in the topological boundary of $G$), 
	and most importantly, 
	every $\tau_j$, with $j=1,\dots,k-2$ satisfies the empty circum-ball property with the same {witness} ball $B(c;\rho)$ 
	as $f$. 
We assume in the sequel that $G$ has been triangulated in this way. 
The fact that this triangulated $G$ will be shown to be embedded will imply that every face $f$ is in fact convex. 

For convenience in the remainder of this section 
we name $W=\{w_i\in \Sites : i=1,\dots,m\}$ the sites that are part of the boundary of
the convex hull $\CHS$, and order them 
in clock-wise order around $\CHS$.

\subsection{Boundary}\label{sec:boundary}

In this section, we assume that the divergence $D$ satisfies the bounded anisotropy assumption~\ref{ass:BAA},
 and conclude that the boundary of the dual triangulation of an orphan-free diagram is the same as the boundary of the convex hull of the sites (and in particular it is simple and closed). 

The vertices in the \emph{topological boundary} of $G$ are those 
	whose corresponding primal regions are unbounded, 
while topological boundary edges are those connecting topological boundary vertices. 
For convenience, we call $B\subseteq E$ the set of topological boundary edges of $G$. 

%
The boundary $\mathcal{B}$ of the convex hull is a simple circular chain 
$\mathcal{B} = \{(w_i, w_{i\oplus 1}) : i=1,\dots,m\}$. 
We prove that 
it is $B=\mathcal{B}$ 
(loosely speaking: the topological, and geometric boundaries of $G$ are the same and coincide with the boundary of $\CHS$), 
which implies that $G$ covers
the convex hull of the sites, and its topological boundary edges form 
a simple, closed polygonal chain. 
All the proofs of this section are in Appendix B. 

\begin{lemma}[$B\subseteq\mathcal{B}$]\label{boundary_easy}
 To every topological boundary edge of $G$ corresponds a segment in the boundary of $\CHS$.
\end{lemma}

We now turn to the converse claim: that to every segment 
in $\mathcal{B}$ corresponds one in $B$. 
Since $B$ is the set of boundary edges of $G$, whose corresponding primal edges 
are unbounded, the claim is equivalent to proving that, to every segment
in $\mathcal{B}$ 
corresponds a boundary edge $(w_i,w_j)\in E$ of $G$ whose corresponding primal edge $\Vor_{ij}$ is unbounded. 

The proof proceeds as follows. First, assume without loss of generality  that the origin is in the interior of $\CHS$. 
Let  $C(\sigma)=\{x\in\mathbb{R}^2 : \|x\|=\sigma\}$ be an origin-centered circle 
	of radius $\sigma$ large enough so that lemmas~\ref{lem:VW} and~\ref{lem:contrad} hold in $C(\sigma)$. 
We define two  functions: 
\begin{eqnarray}\label{eq:pinu}
	&\nu_\sigma &: \partial\CHS\rightarrow C(\sigma), ~~~~~~~~~~~~~\nu_\sigma: r\mapsto \sigma \cdot r / \|r\|, \\
	&\pi &: C(\sigma)\rightarrow\partial\CHS,  
\end{eqnarray}
$\nu_\sigma$  simply projects points in the boundary of $\CHS$ out to their closest point in $C(\sigma)$ 
(using the natural metric; note that $\sigma$ can always be chosen large enough so this projection is unique). 
$\pi$ is constructed as follows.

\begin{figure}[htbp]
   \centering
   	\includegraphics[width=2.7in]{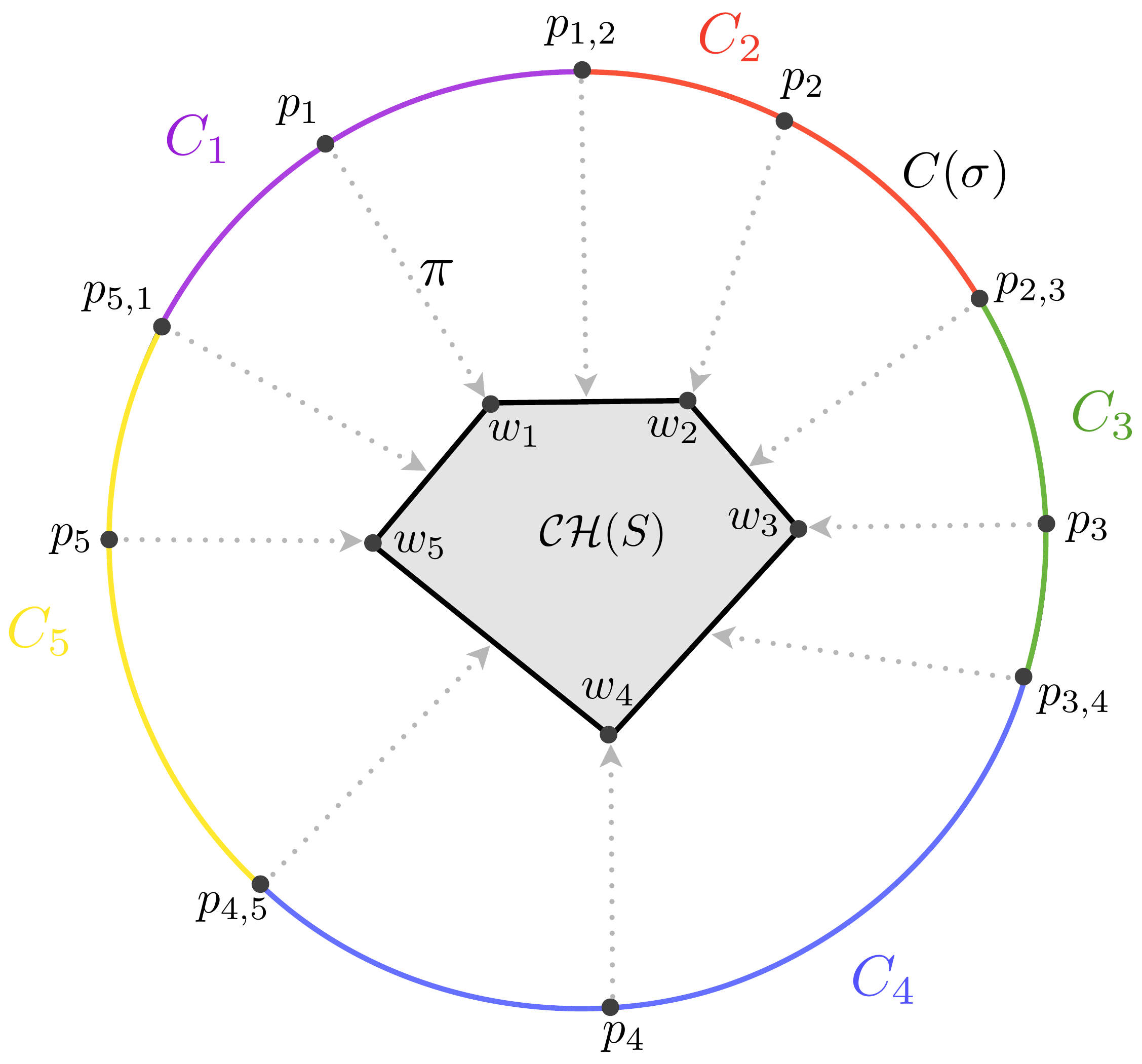} 
   \caption{The construction of the projection function $\pi:C(\sigma)\rightarrow \partial\CHS$. 
   Note that in this case the region-to-site-index function is simply $j(i)=i$, but this cannot be assumed in general. }
   \label{fig:pi}
\end{figure}

\refstepcounter{foo}\thefoo\label{text:boundary}
Consider the situation illustrated in figure~\ref{fig:pi}. 
By lemma~\ref{lem:VW}, all points in $C(\sigma)$ are closer to $W$ than to any interior site $\Sites\setminus W$. 
We split $C(\sigma)$ into 
	a sequence $(C_j)$ of \emph{connected} parts closest 
	to the same boundary site $w_{i(j)}$ 
	(the function $i(\cdot)$ is used to map part indices to the index of their closest site). 
By the convexity of balls, 
	adjacent regions \emph{must} be closest to (circularly) consecutive sites in $W$ 
	(e.g.~if regions $C_1,C_2$ had $i(1)=1$ and $i(2)=3$, 
	 by the continuity of $D$, the point $p$ where $C_1,C_2$ meet would be closest to $w_1,w_3$; 
	 however, since the sites $w_i$ are in cyclic order around $\partial\CHS$, 
	 $p$ would be closer to $w_2$ than to $w_1,w_3$, a contradiction). 
Pick one point $p_j$ for each region $C_j$, and let $\pi(p_j) \equiv w_{i(j)}$. 
For each pair of consecutive regions $C_j, C_{j\oplus 1}$ meeting at $p_{j,j\oplus 1}$, 
	let $\pi(p_{j,j\oplus 1}) \equiv (w_{i(j)} + w_{i(j\oplus 1)})/2$ (the midpoint of two consecutive boundary sites). 
The remaining values of $\pi$ are filled using simple linear interpolation. 
By construction, the following holds:
\begin{property}\label{prop:pi}
$i\cdot$ $\pi:C(\sigma)\rightarrow\partial\CHS$ is continuous. \\
\hspace*{0.91in}$ii\cdot$ Given $p\in C(\sigma)$ and consecutive boundary sites $w_i, w_{i\oplus 1}$, 
	then $p\in\Vor_{i, i\oplus 1}$ iff $\pi(p)=(w_i + w_{i\oplus 1})/2$. 
\end{property}

By the convexity of $\CHS$, $\nu_\sigma$ is continuous in $\partial\CHS$. 
Note that, because $\CHS$ is assumed to contain the origin then, 
	as shown in figure~\ref{fig:pinu}, $\nu_\sigma$ projects 
every point $\pi(p)\in (w_i,w_j)$ lying on a segment of $\partial\CHS$, 
\emph{outwards} from the convex hull (and on the \emph{empty} side of
$(w_i,w_j)$); that is, so that $\nu_\sigma(\pi(p))\in H^{+}_{ij}\cap C(\sigma)$
(i.e.~$\nu_\sigma(\pi(p))$ is in the empty half-space of $(w_i,w_j)$). 

\begin{figure}[htbp]
   \centering
   	\includegraphics[width=2.7in]{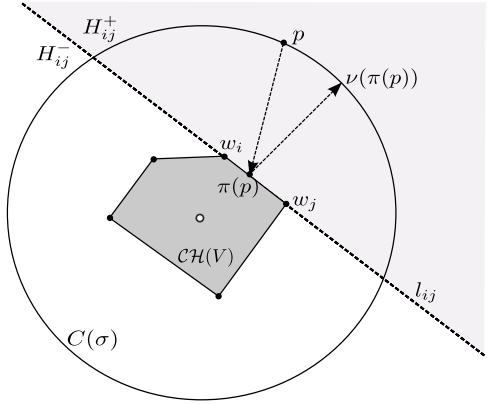} 
   \caption{The construction for the proof of lemma~\ref{lem:hard}.}
   \label{fig:pinu}
\end{figure}

The claim now reduces to showing that for each segment 
	$(w_i,w_j)$ of $\partial\CHS$, and for \emph{every} sufficiently large $\sigma$, 
there is $p\in C(\sigma)$ 
with $p\in\Vor_{ij}$ 
(i.e.~$\pi(p)=(w_i+w_j)/2$). 
Since this implies that $\Vor_{ij}$ is unbounded, 
	it means that the corresponding edge $(w_i,w_j)$ is in $B$ (the topological boundary of $G$).

The proof is by contradiction. 
Lemma~\ref{lem:Sn} uses Brouwer's fixed point theorem to show that, for every segment $(w_i,w_j)$ of $\mathcal{B}$, 
if there were no $p\in C(\sigma)$ closest to $w_i,w_j$, then 
the function $\nu_\sigma\circ\pi:C(\sigma)\rightarrow C(\sigma)$ 
(in fact a slightly different but related function) 
would have a point $q\in C(\sigma)$ such that
$\nu_\sigma(\pi(q))=-q$, that is, such that 
$q$ is ``behind" the segment $(w_i,w_j)\in\partial\CHS$ to which it is
closest ($q\in H^{-}_{ij}$). 
On the other hand, lemma~\ref{lem:contrad} shows that, for all sufficiently large circles $C$, no
point $q\in C(\sigma)$ can be closest to a segment
$(w_i,w_j)\in\partial\CHS$ it is \emph{behind} of, creating a contradiction. 

%
%
%
%
%

The next Lemma is used to create a contradiction, and relies on assumption~\ref{ass:BAA}. 
Lemma~\ref{lem:Sn} is the key lemma in this section, and is a simple application of Brouwer's fixed point theorem. 

\begin{lemma}\label{lem:contrad}
There is $\rho > 0$ such that, for any segment $(w_i,w_j)\in\mathcal{B}$ with supporting line $L_{ij}$, 
	 every $p\in H^{-}_{ij}$ with $\|p\| > \rho$ whose closest point in $L_{ij}$ belongs to $\overline{w_i w_j}$ is 
	closer to a site in $\Sites\setminus\{w_i,w_j\}$ than to $L_{ij}$.
\end{lemma}

\begin{lemma}\label{lem:Sn}
	Every continuous function $F:\mathbb{S}^n\rightarrow\mathbb{S}^n$ that is not onto has a fixed point. 
\end{lemma}

\begin{lemma}[$B\supseteq\mathcal{B}$]\label{lem:hard}
 To every segment  in the boundary of
$\CHS$ corresponds a boundary edge of $G$.
\end{lemma}

Finally, since we have shown that the topological boundary of the dual triangulation is the 
	same as the boundary of the convex hull of the sites, we can conclude that:
	
\begin{corollary}\label{cor:boundary}
The topological boundary of the dual of an orphan-free Voronoi diagram 
	is the boundary of the convex hull $\CHS$, and is therefore simple and closed. 
\end{corollary}

%
%
%
%
%
%
%
%
%

\subsection{Interior}\label{sec:interior}

This section concludes the proof of Theorem~\ref{th:main} by showing that, 
	if the topological boundary of $G$ is simple and closed, then $G$ must be embedded.
The main argument in the proof 
uses proposition~\ref{prop:ECB} and~\ref{cor:boundary}, 
as well as the theory of discrete one-forms on graphs, 
to show that there are no
``edge fold-overs" in $G$ 
(edges whose two incident faces are on the same side of its supporting line), 
and uses this to conclude that the interior of $G$ is a single ``flat sheet", and therefore it is embedded. 

The following definition, from~\cite{1form},  assumes that, for each edge $(s_i,s_j)\in E$ of $G$, 
we distinguish the two opposing half-edges $(s_i,s_j)$ and $(s_j,s_i)$.

\begin{definition}[Gortler et al.\ \cite{1form}]\label{def:1form}
A non-vanishing (discrete) one-form $\xi$  is an assignment of a real value
$\xi_{ij} \neq 0$ to each half edge $(s_i,s_j)$ in $G$, such that 
$\xi_{ji} = -\xi_{ij}$. 
\end{definition}

We can construct a non-vanishing one-form
over $G$ as follows. 
Given some unit direction vector $n\in\mathbb{S}^1$
(in coordinates $n=\left[n_1,n_2\right]^t$), 
we assign a real
value $z(v) = n^t v$ to each vertex $v$ in $G$, and define 
$\xi_{ij} \equiv z(s_i) - z(s_j)$, which clearly satisfies 
$\xi_{ji} = -\xi_{ij}$. The one-form, denoted by $\xi^n$, 
is non-vanishing if, for all edges $(s_i,s_j)\in E$, 
it is $\xi_{ij} = n^t (s_i - s_j) \neq 0$, 
that is, 
if $n$ is not orthogonal to any edge. 
The set of edges has finite cardinality $|E| \le |\Sites| (|\Sites|-1)/2$,  
so \emph{almost all} directions $n\in\mathbb{S}^1$ generate a non-vanishing one-form $\xi^n$.

Since $G=(\Sites,E,F)$ is a planar graph with a well-defined face structure,
there is, for each face $f\in F$, a cyclically ordered set
$\partial f$ of half-edges around the face. 
Likewise, for each vertex $v\in \Sites$, the set $\delta v$ of cyclically ordered
(oriented) half-edges emanating from each vertex is well-defined. 

\begin{definition}[Gortler et al.\ \cite{1form}]
Given 
non-vanishing one-form $\xi$,
the index of vertex $v$ with respect to $\xi$ is \[\mathbf{ind}_{_{\xi}}(v) \equiv 1 - \mathbf{sc}_{_{\xi}}(v) / 2, \]
	where $\mathbf{sc}_{_{\xi}}(v)$ is the number of sign changes of $\xi$ 
	when visiting the half-edges of $\delta v$ in order. 
The index of face $f$ is \[\mathbf{ind}_{_{\xi}}(f) \equiv 1 - \mathbf{sc}_{_{\xi}}(f)/2\] where $\mathbf{sc}_{_{\xi}}(f)$ is the number
of sign changes of $\xi$ as one visits the half-edges of $\partial f$ in order. 
\end{definition}

Note that, by definition, it is always $\mathbf{ind}_{_{\xi^n}}(v) \le 1$. 
A discrete analog of the Poincar\'e-Hopf index theorem relates 
the two indices above:

\begin{theorem}[Gortler et al.\ \cite{1form}]\label{lem:ph}
For any non-vanishing one-form $\xi$, it is 
\[ \displaystyle{\sum_{v\in \Sites} \mathbf{ind}_{_{\xi}}(v) + \sum_{f\in F} \mathbf{ind}_{_{\xi}}(f) = 2} \]
\end{theorem}

Note that this follows from Theorem 3.5 of~\cite{1form} because the unbounded, 
outside face, which is not in $G$, is assumed in this section to be closed and simple
(corollary~\ref{cor:boundary}), and therefore has null index. 
Note that the machinery from~\cite{1form} to deal with degenerate cases
isn't needed here because vertices, by definition, cannot coincide ($\Sites$ is not a multiset). 
All  proofs in this section, except for that of theorem~\ref{th:main}, are  in Appendix C. 

The one-forms $\xi^n$ constructed above satisfy the following property:

\begin{lemma}\label{lem:non-negative}
	Given a non-vanishing one-form $\xi^n$, the sum of indices of interior vertices ($\Sites\setminus W$) of $G$ is non-negative. 
\end{lemma}

The next two lemmas 
relate the presence of edge fold-overs and 
the ECB 
property (proposition~\ref{prop:ECB})  to the indices of vertices in $G$.

\begin{lemma}\label{lem:index-1}
If $G$ has an edge fold-over, then there is $n\in\mathbb{S}^1$ and non-vanishing one-form $\xi^n$ such
that $\mathbf{ind}_{_{\xi^n}}(v) < 0$ for some interior vertex $v\in \Sites\setminus W$. 
\end{lemma}

\begin{lemma}\label{lem:index1}
Given $n\in\mathbb{S}^1$ and non-vanishing one-form $\xi^n$, if $G$ has an interior vertex $v\in \Sites\setminus W$ with index
$\mathbf{ind}_{_{\xi^n}}(v)=1$, then there is a face $f$ of
$G$ that does not satisfy the empty circum-ball 
property (proposition~\ref{prop:ECB}). 
\end{lemma}

The above provides the necessary tools to prove the following key lemma. 

\begin{lemma}\label{lem:ef}
$G$ 
has no edge fold-overs. 
\end{lemma}

Finally, the absence of edge fold-overs, 
together with a simple and closed boundary, 
is sufficient to show that $G$ is 
embedded.

\begin{lemma}\label{lem:interior}
If its (topological) boundary is simple and closed, 
	then the straight-line dual of an orphan-free Voronoi diagram, 
	with vertices at the sites, 
	is an embedded triangulation. 
\end{lemma}

\section{Proof-of-concept implementation}\label{sec:implementation}

\begin{figure}[ht]
\centering
\subfloat[]{\includegraphics[height=2.6cm]{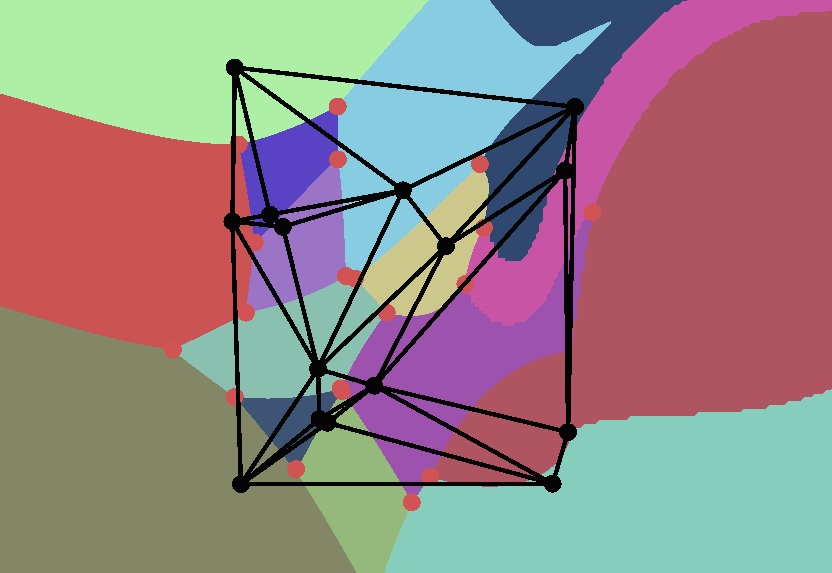}\label{fig:img_a}}
\subfloat[]{\includegraphics[height=2.6cm]{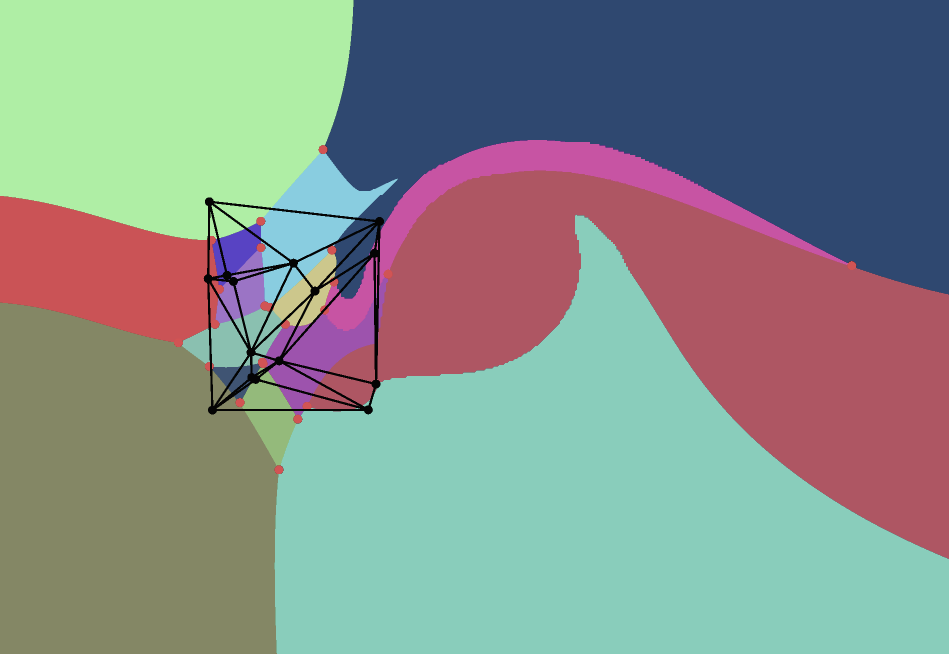}\label{fig:img_b}}
\subfloat[]{\includegraphics[height=2.6cm]{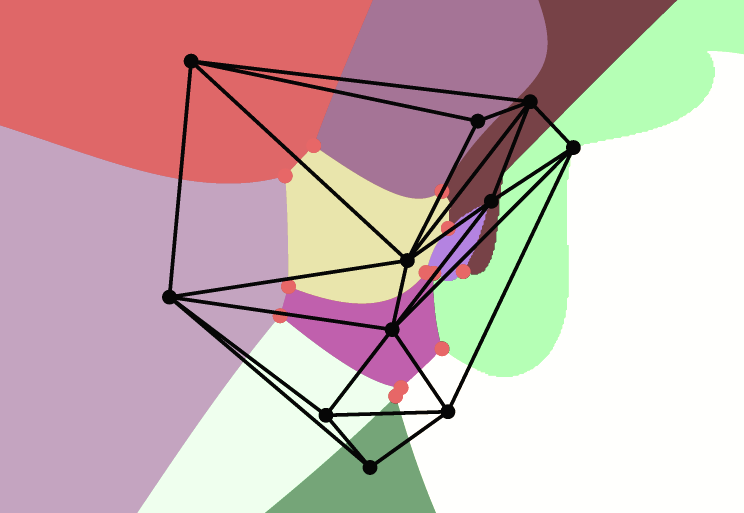}\label{fig:img_c}}
\subfloat[]{\includegraphics[height=2.6cm]{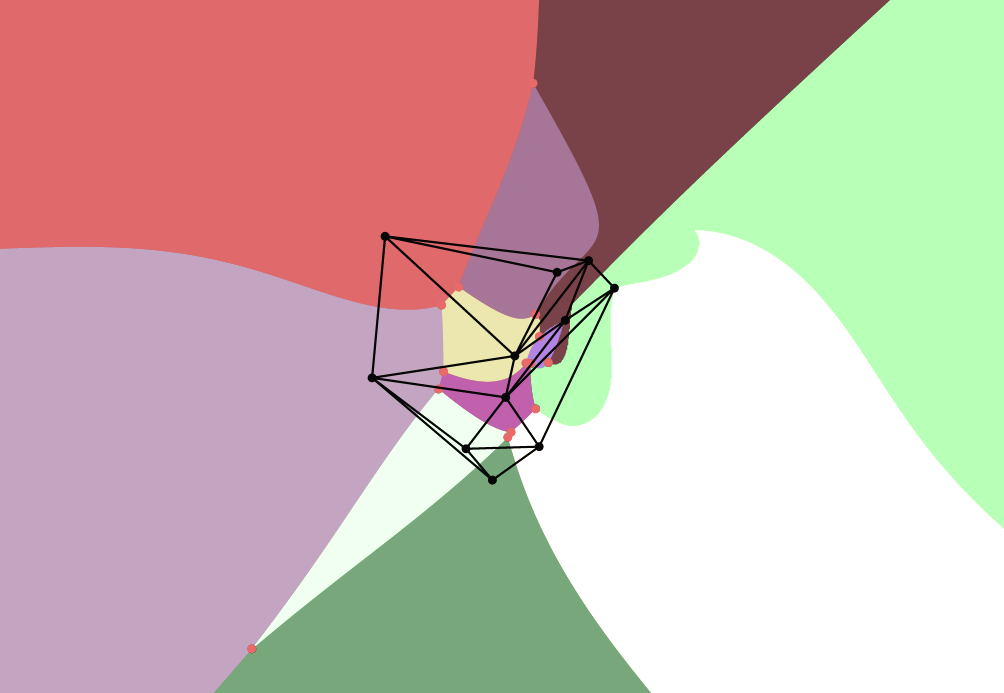}\label{fig:img_d}}\quad
\subfloat[]{\includegraphics[height=2.6cm]{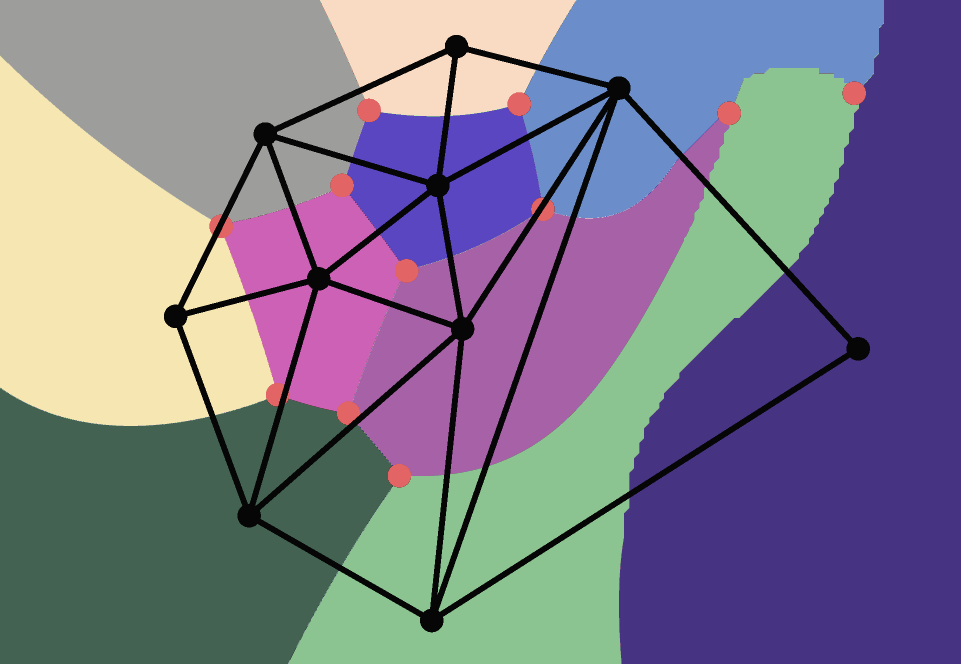}\label{fig:img_e}}
\subfloat[]{\includegraphics[height=2.6cm]{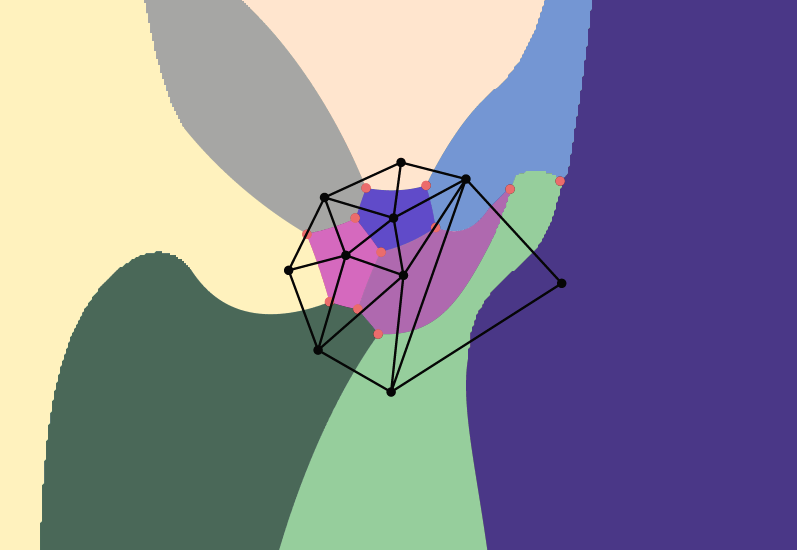}\label{fig:img_f}}
\subfloat[]{\includegraphics[height=2.6cm]{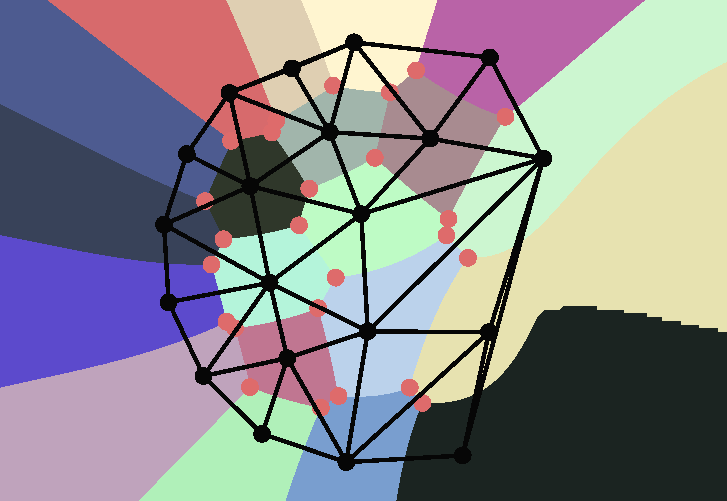}\label{fig:img_g}}
\subfloat[]{\includegraphics[height=2.6cm]{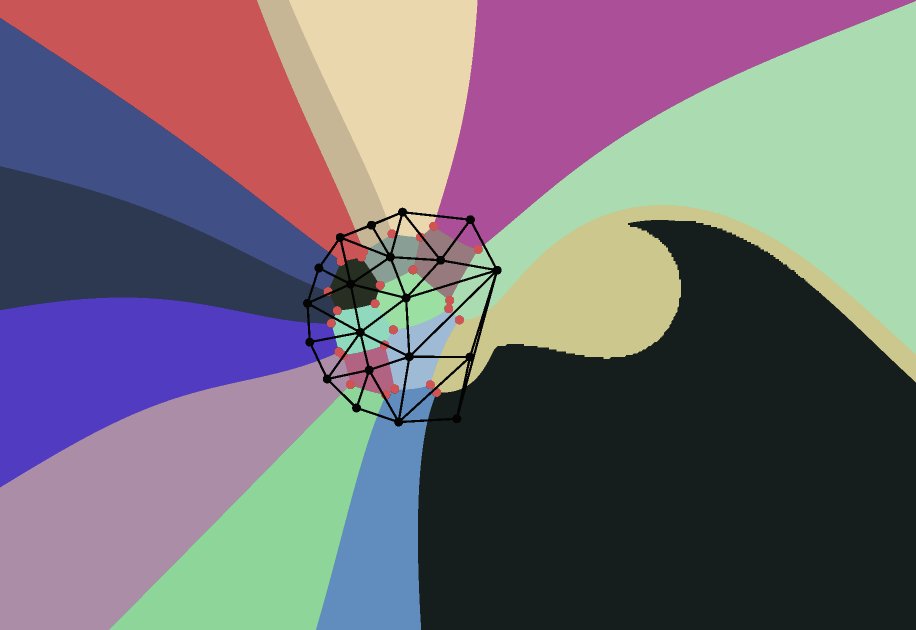}\label{fig:img_h}}
\caption{
Anisotropic Voronoi diagrams, and their duals generated by our
proof-of-concept implementation. 
Voronoi vertices are marked as red dots, while dual vertices (sites) and edges are drawn
in black.}
\label{fig:test}
\end{figure}

Though not aiming for an efficient implementation, 
we tested a simple proof-of-concept that constructs anisotropic Voronoi diagrams
(using a quadratic divergence $D_Q$ of the type discussed in section~\ref{sec:DQ}) 
and their duals
(figure~\ref{fig:test}). 
A closed-form metric, which 
has bounded ratio of eigenvalues 
(and therefore by lemma~\ref{lem:DQgamma} satisfies assumption~\ref{ass:BAA}), 
is discretized on a fine regular grid, and linearly interpolated inside grid elements, resulting in a
continuous metric. The sites are generated randomly (figures~\ref{fig:img_a}
and~\ref{fig:img_b}), or using a combination of random, and equispaced
points forming an (asymmetric) $\epsilon$-net~\cite{avd} (remaining figures). 

The primal diagram was obtained using front propagation from the sites
outwards, until fronts meet at Voronoi edges. 
The runtime is proportional to the grid size, since every grid-vertex is visited exactly six times (equal to their valence), 
	and so linear in the resolution of the sampled divergence $D_Q$. 

The implementation does not 
guarantee the correctness of the diagram unless it \emph{is} orphan-free, and serves to verify the claims of the paper since well-behave-ness of the dual is predicated on that of the primal. 

The two main claims of the paper (that orphan-freedom is sufficient to ensure well-behavedeness of 
both the dual and the primal) are clearly illustrated in these examples. 
In all examples, 
the dual covers the convex hull of the vertices
(corollary~\ref{cor:boundary}), is a
single cover, embedded with straight edges without edge crossings
(lemma~\ref{lem:interior}), 
and has no degenerate faces 
(since, by proposition~\ref{prop:ECB}, the vertices of a face lie on the boundary of a strictly convex ball). 
By focusing on the primal diagrams (second and fourth column), further claims in
the paper become apparent, namely that Voronoi regions 
(Voronoi elements of order one according to definition~\ref{def:VorI}) are simply connected (lemma~\ref{lem:regionSC}), 
and Voronoi edges (order two), and vertices (order three or higher) are connected (corollary~\ref{cor:VorI}). 

\section{Conclusion and open problems}

We studied the properties of duals of orphan-free Voronoi diagrams with respect to divergences, for the
purposes of constructing triangulations on the plane. 
The main result (Theorems~\ref{th:main}) is that
the dual, with straight edges and vertices at the sites, is embedded
and covers the convex hull of the sites, mirroring similar results for
ordinary Voronoi diagrams and their duals.
Additionally, the primal is composed of connected elements (corollary~\ref{cor:VorI}). 


\begin{figure}[ht]
\centering
\includegraphics[width=5.6cm]{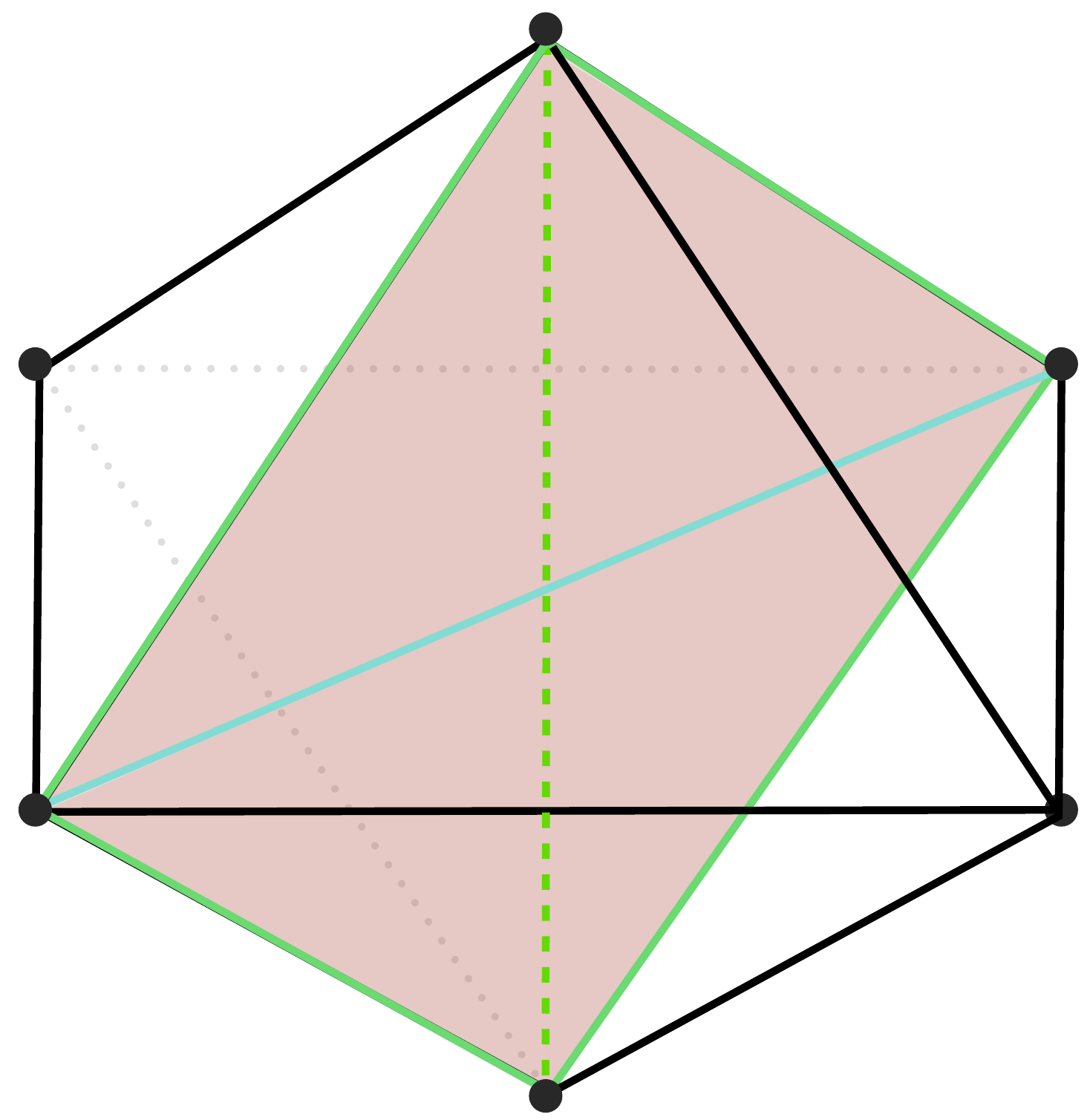}
\caption{
	The main proof of this paper does not work as is in higher dimensions. 
This arrangement of tetrahedra is not embedded: 
	the red tetrahedra has been ``inverted" 
	(the green dotted edge is \emph{behind} the solid blue edge), 
	``invading" the two front tetrahedra (closest to the viewer), 
	as well as the two back tetrahedra (farthest from the viewer). 
However, it does not violate the ECB condition (proposition~\ref{prop:ECB}). 
}
\label{fig:3dbreaks}
\end{figure}
Perhaps the most important outstanding question is whether these results
extend to higher dimensions. 
The proofs in Secs.~\ref{sec:boundary} and~\ref{sec:interior},
except for lemma~\ref{lem:regionSC}, can be trivially extended to n dimensions. 
Section~\ref{sec:boundary} has been written only for the two-dimensional case,
but a similar construction, and the same argument would work in higher
dimensions (lemma~\ref{lem:Sn} being a hint of this). 
It is the argument in section~\ref{sec:interior}, and described in figure~\ref{fig:outline}, that becomes problematic. 
While the ECB property is shown to be sufficient to prevent fold-overs in the
triangulation, it is not sufficient in 
higher dimensions. In particular, fixing the boundary to be simple and convex, 
there are simple arrangements of tetrahedra in $\mathbb{R}^3$ that contain 
face fold-overs but do not break the ECB property. 
In particular, the arrangement of tetrahedra of figure~\ref{fig:3dbreaks} is not embedded: 
	the red tetrahedra has been ``inverted" 
	(indicated by the green dotted edge being \emph{behind} the solid blue edge); 
	its interior overlaps that of the two front tetrahedra (closest to the viewer), 
	as well as the two back tetrahedra (those farthest from the viewer). 
However, this arrangement does not break the ECB condition (proposition~\ref{prop:ECB}, which holds in any dimension), 
	and therefore the same argument used in this work would not create a 
	contradiction in higher dimensions. 

\bibliographystyle{plain}
\bibliography{vddw8}

\begin{thebibliography}{10}

\bibitem{powerdiag}
Franz Aurenhammer.
\newblock Power diagrams: Properties, algorithms and applications.
\newblock {\em SIAM J. Comput.}, 16(1):78--96, February 1987.

\bibitem{Aurenhammer:1991}
Franz Aurenhammer.
\newblock Voronoi diagrams - a survey of a fundamental geometric data
  structure.
\newblock {\em ACM Comput. Surv.}, 23(3):345--405, September 1991.

\bibitem{Bregman}
Jean-Daniel Boissonnat, Frank Nielsen, and Richard Nock.
\newblock Bregman voronoi diagrams.
\newblock {\em Discrete Comput. Geom.}, 44(2):281--307, September 2010.

\bibitem{bondy2008graph}
A.~Bondy and U.S.R. Murty.
\newblock {\em Graph Theory}.
\newblock Graduate Texts in Mathematics. Springer, 2008.

\bibitem{DBLP:conf/imr/CanasG06}
Guillermo~D. Canas and Steven~J. Gortler.
\newblock On asymptotically optimal meshes by coordinate transformation.
\newblock In Philippe~P. P{\'e}bay, editor, {\em International Meshing
  Roundtable}, pages 289--305. Springer, 2006.

\bibitem{avd}
Guillermo~D. Canas and Steven~J. Gortler.
\newblock Orphan-free anisotropic {V}oronoi diagrams.
\newblock {\em Discrete {\&} Computational Geometry}, 46(3):526--541, 2011.

\bibitem{conformal}
C.~Carathéodory.
\newblock Über die gegenseitige beziehung der ränder bei der konformen
  abbildung des inneren einer jordanschen kurve auf einen kreis.
\newblock {\em Mathematische Annalen}, 73(2):305--320, 1913.

\bibitem{DMG}
Siu-Wing Cheng, Tamal~K. Dey, and Jonathan~Richard Shewchuk.
\newblock {\em Delaunay Mesh Generation}.
\newblock Chapman and Hall / CRC computer and information science series. CRC
  Press, 2013.

\bibitem{CsiszarTutorial}
I.~Csisz\'ar and P.C. Shields.
\newblock Information theory and statistics: A tutorial.
\newblock {\em Foundations and Trends in Communications and Information
  Theory}, 1(4):417--528, 2004.

\bibitem{DAzevedo}
E.~F. D'Azevedo and R.~B. Simpson.
\newblock On optimal triangular meshes for minimizing the gradient error.
\newblock {\em Numerische Mathematik}, 59(4):321--348, July 1991.

\bibitem{DW}
Qiang Du and Desheng Wang.
\newblock Anisotropic centroidal {V}oronoi tessellations and their
  applications.
\newblock {\em SIAM Journal of Scientific Computing}, 26(3):737--761, 2005.

\bibitem{Aurenhammer13}
D.T.~Lee F.~Aurenhammer, R.~Klein.
\newblock {\em Voronoi Diagrams and {D}elaunay Triangulations}.
\newblock World Scientific Publishing Company, Singapore, 2013.

\bibitem{power}
H.~Edelsbrunner F.~Aurenhammer.
\newblock An optimal algorithm for constructing the weighted {V}oronoi diagram
  in the plane.
\newblock {\em Pattern Recognition}, 17:251--257, 1984.

\bibitem{1form}
Steven~J. Gortler, Craig Gotsman, and Dylan Thurston.
\newblock Discrete one-forms on meshes and applications to 3d mesh
  parameterization.
\newblock {\em Computer Aided Geometric Design}, 23:83--112, 2006.

\bibitem{konigsberger2006analysis}
K.~K{\"o}nigsberger.
\newblock {\em Analysis 2}.
\newblock Number v. 2 in Springer-Lehrbuch. Physica-Verlag, 2006.

\bibitem{LS}
Francois Labelle and Jonathan~R. Shewchuk.
\newblock Anisotropic {V}oronoi diagrams and guaranteed-quality anisotropic
  mesh generation.
\newblock In {\em SCG '03: Proceedings of the Nineteenth Annual Symposium on
  Computational Geometry}, pages 191--200, New York, NY, USA, 2003. ACM.

\bibitem{matousek2002lectures}
J.~Matou{\v{s}}ek.
\newblock {\em Lectures on Discrete Geometry}.
\newblock Graduate Texts in Mathematics. Springer New York, 2002.

\bibitem{Milnor}
John~W. Milnor.
\newblock {\em Topology from the Differentiable Viewpoint}.
\newblock Princeton University Press, 1997.

\bibitem{Moore1918}
Robert~L. Moore.
\newblock A characterization of jordan regions by properties having no
  reference to their boundaries.
\newblock {\em Proceedings of the National Academy of Sciences}, 4(12):pp.
  364--370, 1918.

\bibitem{munkres2000topology}
J.R. Munkres.
\newblock {\em Topology}.
\newblock Featured Titles for Topology Series. Prentice Hall, Incorporated,
  2000.

\bibitem{rockafellar1997convex}
R.T. Rockafellar.
\newblock {\em Convex Analysis}.
\newblock Convex Analysis. Princeton University Press, 1997.

\bibitem{triangle}
J.R. Shewchuk.
\newblock What is a good linear element? interpolation, conditioning, and
  quality measures.
\newblock {\em Eleventh International Meshing Roundtable}, pages 115--–126,
  September 2002.

\end{thebibliography}

\section*{Appendix A: Bounded anisotropy condition}\label{app:gamma}

\noindent{\bf lemma~\ref{lem:DFgamma}} (Bounded anisotropy for Bregman divergences).
\emph{
If $F\in\mathcal{C}^2$ and there is $\gamma > 0$ such that the Hessian of $F$ has ratio of eigenvalues bounded by $\lambda_{\text{min}}/\lambda_{\text{max}}\ge \gamma$,
	then assumption~\ref{ass:BAA} holds. 
}\begin{proof}
Consider the situation described in figure~\ref{fig:gamma}, 
	in a coordinate system 
	with the y-axis along $L_{pq}$. 

Let $\rho\equiv D_F(m\parallel c) = F(m)-F(c)-\langle m-c,\nabla F(c)\rangle$. 
Because $F\in\mathcal{C}^2$ and the ball $B(c;\rho)$ it tangent to the y-axis at $m$, 
	it is 
\begin{eqnarray*}
	0 &=& \langle \hat{y}, \nabla_x \D{x}{c}\bigg|_{x=m}\rangle 
	   = \langle \hat{y}, \nabla_x \left[ F(x) - F(c) - \langle x-c, \nabla F(c) \rangle \right]\bigg|_{x=m} \rangle \\
	   &=& \langle \hat{y}, \nabla F(m) - \nabla F(c) \rangle.
\end{eqnarray*}
Since $D_F(m\parallel c) - D_F(r\parallel c) = F(m)-F(r) - \langle m-r,\nabla F(c)\rangle$, 
we can obtain the value of $D_F(r\parallel c)$ by integration from $m$, first 
	along the y-axis from $(m_x,m_y)$ to $(m_x,r_y)$, 
	then along the x-axis from $(m_x,r_y)$ to $(r_x,r_y)$. 

Let $\delta_x\equiv r_x-m_x$ and $\delta_y\equiv r_y-m_y$, 
	and assume that $\delta_x> 0$ and $\delta_y\ge 0$ without loss of generality, 
	since $r\not\in L_{pq}$, $r$ is on the same side of $L_{pq}$ as $c$, 
	and we have freedom in choosing the sign of the axis. 

For assumption~\ref{ass:BAA} to hold it must be $D_F(r\parallel c) < D_F(m\parallel c)$. This holds whenever
\[
	\delta_x |\nabla F(m)-\nabla F(c)| > 
		\lambda_{\text{max}} \delta_x^2 / 2 + \lambda_{\text{max}}\delta_y^2 / 2 + \lambda_{\text{max}} \delta_x\delta_y
\]
or equivalently
\[
	\lambda_{\text{min}}\delta_x \|m - c\| > 
		\lambda_{\text{max}} \delta_x^2 / 2 + \lambda_{\text{max}}\delta_y^2 / 2 + \lambda_{\text{max}} \delta_x\delta_y
\]
which reduces to 
\[
	\|m-c\| > \gamma^{-1} \left( \delta_x^2/2 + \delta_y^2/2 + \delta_x\delta_y\right)/\delta_x
\]
and is always satisfied whenever $\|c\| \ge \max\{\|p\|,\|q\|\} + \gamma^{-1} \left( \delta_x^2/2 + \delta_y^2/2 + \delta_x\delta_y\right)/\delta_x$. 
Note that this bound is finite because $\gamma>0$ and $\delta_x>0$, $\delta_y\ge 0$. 
\end{proof}

%
%

\noindent{\bf lemma~\ref{lem:DQgamma}} (Bounded anisotropy for quadratic divergences).
\emph{
If there is $\gamma > 0$ such that $Q$ has ratio of eigenvalues bounded by $\lambda_{\text{min}}/\lambda_{\text{max}}\ge \gamma$,
	then assumption~\ref{ass:BAA} holds. 
}\begin{proof}
The proof of this lemma can be reduced to that of lemma~\ref{lem:DFgamma}.
Given $c\in\mathbb{R}^2$, we let \[ F(\cdot)\equiv D^2_Q(\cdot\parallel c)=(\cdot-c)^t Q(c) (\cdot-c)/2, \]
whose Hessian is $H_F\equiv Q(c)$. Since $Q$ has eigenvalues bounded from below by $\gamma$, 
	the conditions of the proof of lemma~\ref{lem:DFgamma} hold. 
Note that this definition of $F(\cdot)$ is \emph{per choice of} $c$, 
	and therefore we are not defining a real Bregman divergence this way, but simply choosing 
	a different $F$ for each $c$ as to satisfy the conditions of the proof. 

\end{proof}

\begin{figure}[htbp]
   \centering
	\subfloat[]{\label{fig:norm.a}\includegraphics[width=2.3in]{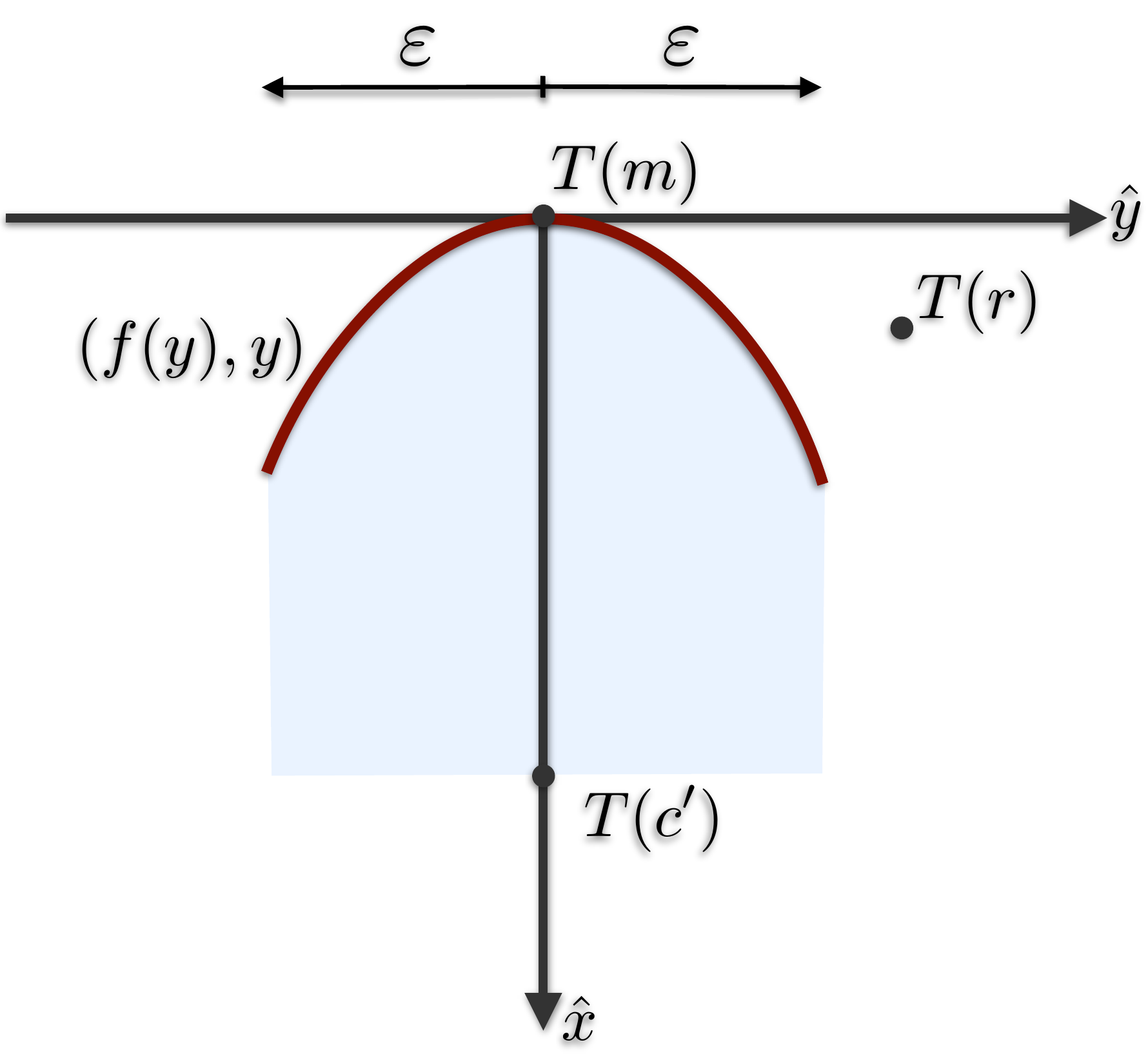}}\quad\quad
	\subfloat[]{\label{fig:norm.b}\includegraphics[width=2.5in]{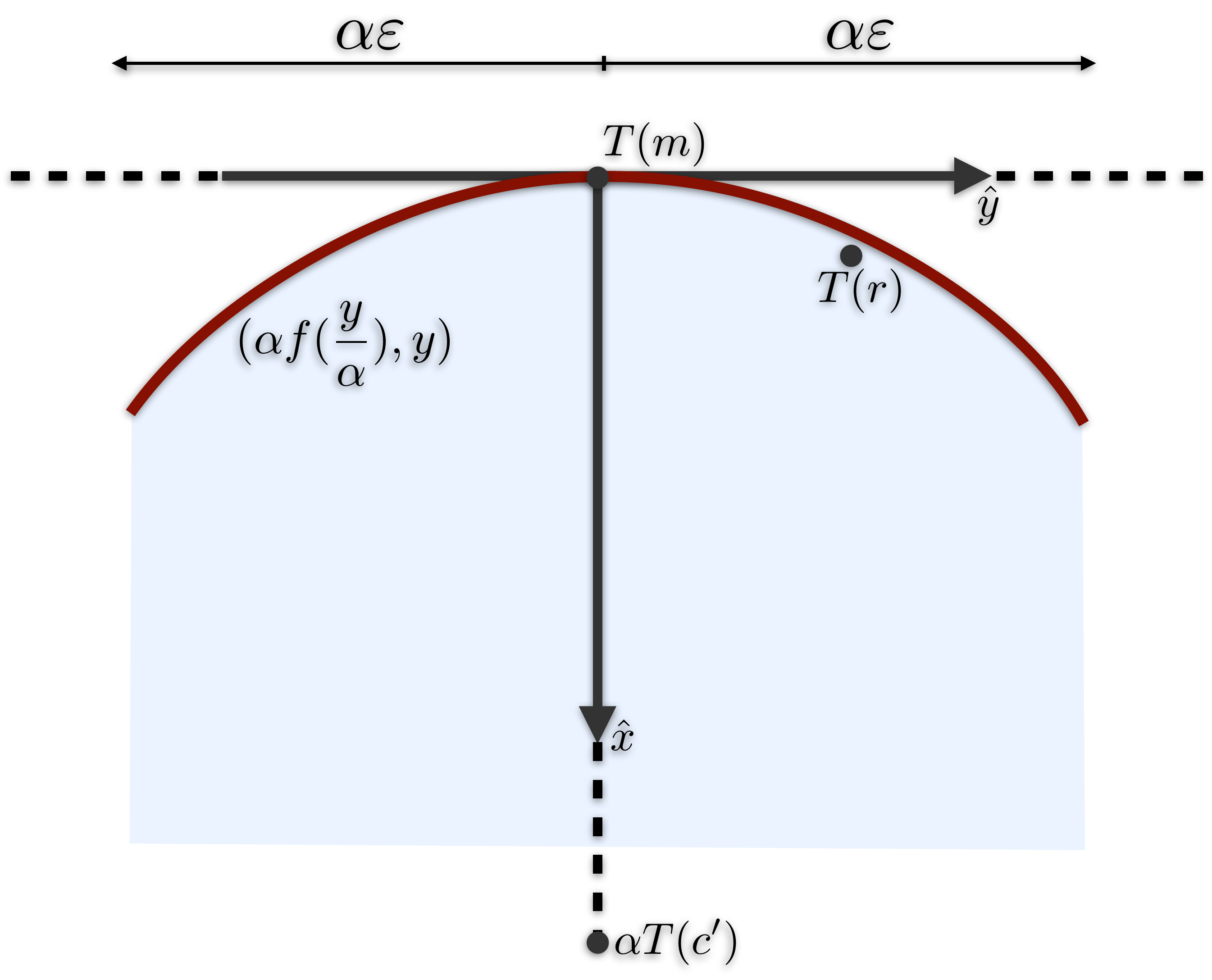}}
   \caption{Diagrams used in the proof of lemma~\ref{lem:Lpgamma}. 
   	By making $\alpha$ large enough, we can ensure that $T(r)$ falls in the blue shaded region, and 
		therefore $\|r-c\|_K < \|m-c\|_K$, where $T(c)\equiv \alpha T(c')$. 
   }
\end{figure}

\noindent{\bf Lemma~\ref{lem:Lpgamma}} (Bounded anisotropy for normed spaces)\emph{
	Distances derived from strictly convex $\mathcal{C}^1$ norms satisfy assumption~\ref{ass:BAA}.
}
\begin{proof}
Let $\|\cdot\|_K$ be a strictly convex $\mathcal{C}^1$ norm, whose unit ball is the symmetric convex body $K$. 
Let $p,q\in\mathbb{R}^2$ with supporting line $L_{pq}$, and $r\notin L_{pq}$ be given. 
For any $c\notin L_{pq}$ with closest point $m\in\overline{pq}$ in $L_{pq}$, define $\pi_{m,c} \equiv (c-m) / \|c-m\|_K$. 
Defining $m$ to be the origin, let $T_{\pi_{m,c}}$ be a linear transformation that maps the $L_{pq}$ direction into the ${y}$-axis, 
	and $\pi_{m,c}$ into the $x$ axis. 
The fact that $c\notin L_{pq}$ implies that $T_{\pi_{m,c}}$ is non-singular. 
Choose the sign of the $y$-axis so that $\lambda_{max}(T_{\pi_{m,c}}) \ge \lambda_{min}(T_{\pi_{m,c}}) > 0$ 
	are the maximum and minimum eigenvalues of $T$, 	respectively. 

Consider the following statements:
\begin{enumerate}[i]
\item For all pairs $(m,c)$, there is a sufficiently large $\mu_{\pi_{m,c}} > 0$ such that whenever 
	$\|m-c\| > \mu_{\pi_{m,c}}$ then $\|r-c\|_K < \|m-c\|_K$. 
\item For all pairs $(m,c)$, there is a sufficiently large $\mu^T_{\pi_{m,c}} > 0$ such that whenever
	$\|T(m) - T(c)\| > \mu^T_{\pi_{m,c}}$ then $\|r-c\|_K < \|m-c\|_K$. 
\end{enumerate}
	
\vspace*{0.05in}\noindent{\bf [Reducing assumption~\ref{ass:BAA} to statement (i)]}. 	
Given (i), and since both $\overline{pq}$ and $\partial K$ are compact, we can define:
	\[\mu \equiv \max\{ \|p\|, \|q\| \} + \displaystyle{\max_{\substack{ \|\pi_{m,c}\|_{_K} = 1 \\ m\in \overline{pq} }} \mu_{\pi_{m,c}}, }\]
from which it follows that whenever $\|c\| > \mu$, it holds
	\[ \|m - c\| \ge \|c\| - \|m\| \ge \|c\| - \max\{ \|p\|, \|q\| \} >  
			\displaystyle{\max_{\substack{ \|\pi_{\tilde{m},\tilde{c}}\|_K = 
			1 \\ \tilde{m}\in \overline{pq} }} \mu_{\pi_{\tilde{m},\tilde{c}}}} \ge \mu_{\pi_{m,c}}, \]
	and therefore $\|r-c\|_K < \|m-c\|_K$, thereby satisfying assumption~\ref{ass:BAA}. 

\vspace*{0.05in}\noindent{\bf [Reducing statement (i) to statement (ii)]}. 	
Assume (ii) is true and let $\mu_{\pi_{m,c}} \equiv \lambda_{max}(T^{-1}) \mu^T_{\pi_{m,c}}$. 
	Whenever $\|m-c\| > \mu_{\pi_{m,c}}$, it holds:
\begin{eqnarray*}
	\|T(m) - T(c)\| &=& \|T(m-c)\| \ge \lambda_{min}(T) \|m-c\| > \lambda_{min}(T)\mu_{\pi_{m,c}} \\
			&=& \lambda_{min}(T) \lambda_{max}(T^{-1}) \mu^T_{\pi_{m,c}} = \mu^T_{\pi_{m,c}},
\end{eqnarray*}
	and therefore by (ii) it is $\|r-c\|_K < \|m-c\|_K$.

\noindent{\bf [Proof of statement (ii)]}. 
Consider the situation depicted in figure~\ref{fig:norm.a}, which shows a portion of the plane transformed by $T$. 
Given $c'=m+\pi_{m,c}$, consider the set of points at distance $\|m-c'\|_K$ from $c'$ (red line). 
First note that, because we have temporarily chosen $m$ as the origin, then $m=T(m)$, and $T(m)+\alpha [T(c')-T(m)] = \alpha T(c')$. 
Because $\|\cdot\|_K$, there is an open interval $(-\varepsilon,\varepsilon)$ and 
	a function $f\in\mathcal{C}^1$ such that $(f(y),y)$, with $y\in(-\varepsilon,\varepsilon)$ 
	are the coordinates of the points (in $T$-space) at distance $\|m-c'\|_K$ from $c'$. 

Because $m$ is the point closest to $c'$ in $L_{pq}$, then, 
	in $T$-space, $(f(y),y)$ is tangent to the $y$-axis at $y=0$, 
	and therefore $f(0)=f'(0)=0$, from which it follows that
	\[ \displaystyle{\lim_{y\rightarrow 0} (f(y)-f(0)) / (y-0)} =   \displaystyle{\lim_{y\rightarrow 0} f(y)/y} = 0. \]

By a simple calculation, it is simpe to show that moving $T(c')$ further down along the $x$ axis to $\alpha T(c')$
	(figure~\ref{fig:norm.b}), 
	scales the red curve of figure~\ref{fig:norm.a} by a factor $\alpha$, so that it becomes 
	$(\alpha f(y\alpha), y)$ with $y\in(-\alpha\varepsilon, \alpha\varepsilon)$. 

Given $T(r)=(r^T_x,r^T_y)$ in coordinates, with $r^T_x > 0$ and $r^T_y \ge 0$, without loss of generality, 
	then, from the figure and the expression for the curve $(\alpha f(y\alpha), y)$, it is clear that it is 
	possible to choose $\alpha$ large enough so that $T(r)$ is below the curve $(\alpha f(y\alpha), y)$, 
	and therefore $r$ is closer (with respect to $\|\cdot\|_K$) to $\alpha T(c')$ then $T(m)$. 
By setting $\mu^T_{\pi_{m,c}} \equiv \alpha \|T(c')-T(m)\|$, statement (ii) follows. 

In particular, we simply choose $\alpha$ large enough so that
\begin{itemize}
	\item $[\alpha T(c')]_x > r^T_x$; 
	\item $\alpha T(c')$ is far enough from $T(m)$ so the line $\overline{[\alpha T(c')] T(r)}$ crosses the $y$-axis between 
		$(-\alpha\varepsilon, \alpha\varepsilon)$, which is clearly possible for sufficiently large $\alpha$; 
	\item if $r^T_x > 0$, then it is a simple calculation to show that 
		we can ensure that $T(r)$ is ``below" the curve $(\alpha f(y\alpha), y)$ as follows:
		1) choose a small enough $\delta$ such that $f(\delta)/\delta < r^T_x/r^T_y$, 
		which is always possible because $\lim_{y\rightarrow 0} f(y)/y = 0$, 
		and 2) enforcing $\alpha > r^T_y / \delta$. 
\end{itemize}

\end{proof}

\section*{Appendix B: dual triangulation (boundary)}\label{app:boundary}

Let $s_i,s_j\in\Sites$ be two sites, we denote by $H^{+}_{ij},H^{-}_{ij}$ 
	the two open half-spaces on either side of their supporting line $L_{ij}$. 
The set $\{H^{+}_{ij},H^{-}_{ij},L_{ij}\}$ is therefore a disjoint partition of $\mathbb{R}^2$.
Whenever the two sites we consider are on the boundary of $\CHS$, 
	they are denoted by $w_i,w_j\in W\subseteq\Sites$, 
	and we always choose $H^{+}_{ij}$ to be the ``empty" half space of the two 
	(such that $H^{+}_{ij} \cap \Sites =\phi$).

\begin{lemma}\label{lem:halfspace}
    Given a Voronoi edge $\Vor_{ij}$ 
    	corresponding to neighboring sites $s_i,s_j\in\Sites$, 
if $\Vor_{ij}\cap  H^{+}_{ij}$ ($\Vor_{ij}\cap  H^{-}_{ij}$) is unbounded, then it is $ H^{+}_{ij}\cap \Sites=\phi$
 ($ H^{-}_{ij}\cap \Sites=\phi$), 
where $H^{+}_{ij},H^{-}_{ij}$ are open half spaces on either side of the
supporting line of $s_i,s_j$. 
\end{lemma}
\begin{proof}
Assume $w\in H^{+}_{ij}\cap \Sites$. 
Since $\Vor_{ij}\cap H^{+}_{ij}$ is unbounded, 
we choose $p\in \Vor_{ij}\cap H^{+}_{ij}$ of sufficiently large norm, 
	so that assumption~\ref{ass:BAA} 
implies that $\D{w}{p} < \D{(w_i+w_j)/2}{p}$. 
By the convexity of $\D{\cdot}{p}$, this means that $\D{w}{p} < \D{(s_i+s_j)/2}{p} <
\D{s_i}{p}=\D{s_j}{p}$. 
Since $p$ is closer to $w$ than to $s_i,s_j$, it is $p\notin \Vor_{ij}$, a
contradiction. 
\end{proof}

\noindent{\bf Lemma~\ref{boundary_easy}}
\emph{	
 To every topological boundary edge of $G$ corresponds a segment in the boundary of $\CHS$.
}
\begin{proof}
By the definition of $G$, 
to every boundary edge $(s_i,s_j)\in B$ corresponds a primal edge $\Vor_{ij}$
that is unbounded. 

Consider the two open half-planes $H^{+}_{ij}$ and $H^{-}_{ij}$ on either
side of the supporting line $L_{ij}$ of $s_i,s_j$. 
We split $\Vor_{ij}$ in three parts: $\Vor_{ij}\cap H^{+}_{ij}$, $\Vor_{ij}\cap L_{ij}$, and $\Vor_{ij}\cap H^{-}_{ij}$, 
	at least one which must be unbounded. 
	
Since, by lemma~\ref{lem:midpoint}, it is $\Vor_{ij}\cap L_{ij} \in \overline{s_i s_j}$ (and therefore bounded), 
	then it must be that either $\Vor_{ij}\cap H^{+}_{ij}$ or $\Vor_{ij}\cap H^{-}_{ij}$ are unbounded. 
By lemma~\ref{lem:halfspace}, they cannot both be, or else 
$H^{+}_{ij}\cap \Sites=\phi$ and $H^{-}_{ij}\cap \Sites=\phi$, and therefore all sites
would be in $L_{ij}$ (all colinear). 
Assume w.l.o.g.\  that $\Vor_{ij}\cap H^{+}_{ij}$ is unbounded. 

By lemma~\ref{lem:halfspace}, $\Vor_{ij}\cap H^{+}_{ij}$ unbounded implies 
$H^{+}_{ij}\cap \Sites=\phi$, and 
so $s_i,s_j$ must lie in the boundary of $\CHS$ 
($s_i,s_j\in W$ and $\overline{s_i s_j}\subseteq\partial\CHS$). 

It only remains to show that $s_i,s_j$ are consecutive in the sequence $(w_i :
i=1,\dots,m)$. 
We prove this by contradiction. 
If they were not, then since $\overline{s_i s_j}\subseteq\partial\CHS$, 
there must be some site $w\in \overline{s_i s_j}$, $w\neq s_i,s_j$. 
However, 
	this is not possible. 
To see this, simply pick some point $p\in \Vor_{ij}$, by definition closest to $s_i,s_j$;
 by the convexity of $\D{\cdot}{p}$, it must be $\D{w}{p} < \D{s_i}{p}=\D{s_j}{p}$, a contradiction. 

Since $s_i,s_j$ are consecutive vertices in $(w_i : i=1,\dots,m)$, then 
$(s_i,s_j)\in\mathcal{B}$. 
\end{proof}

\begin{figure}
\centering
\subfloat[]{\includegraphics[height=3.0cm]{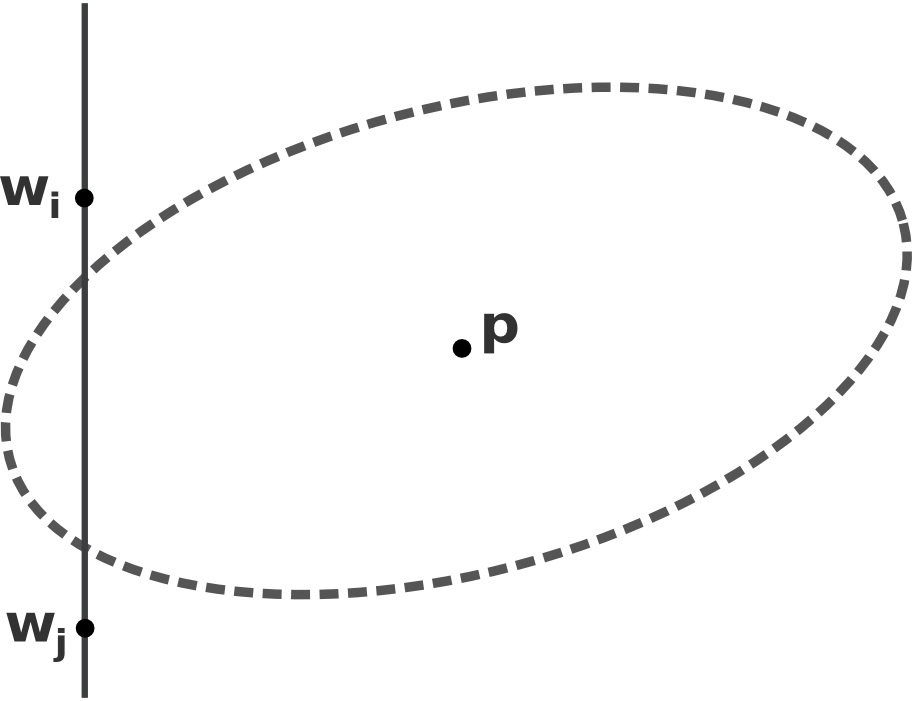}\label{fig:wijcap}}
\subfloat[]{\includegraphics[height=3.0cm]{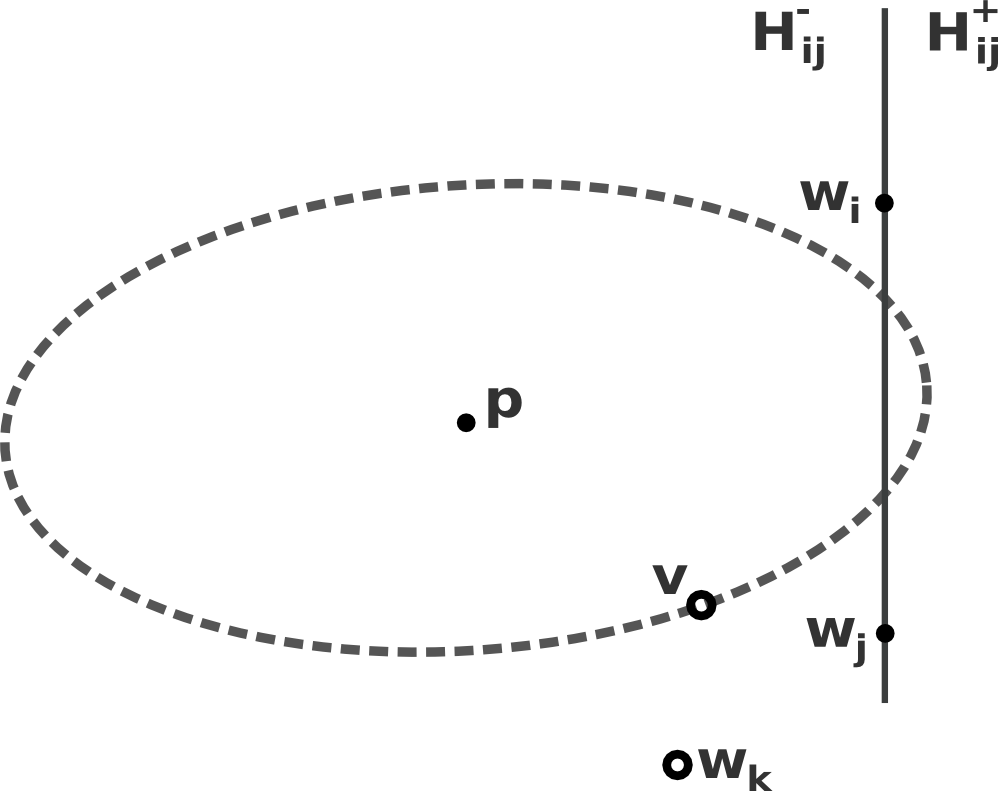}\label{fig:vH-}}
\quad\quad
\subfloat[]{\includegraphics[height=3.0cm]{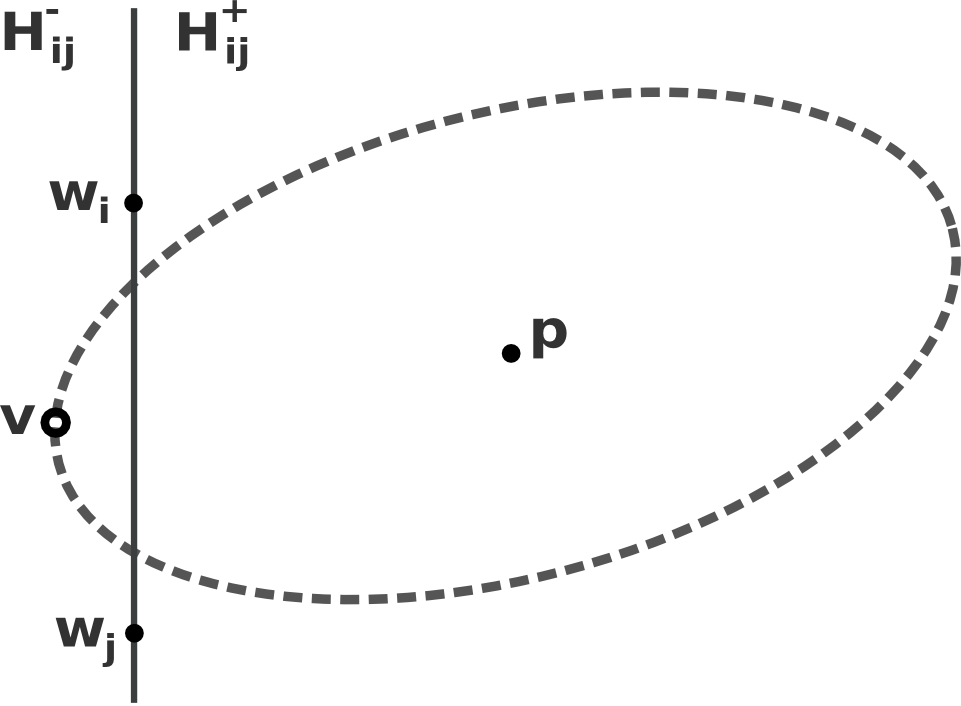}\label{fig:vH+}}
\label{fig:VW}
\caption{Diagram for the proof of lemma~\ref{lem:VW}.}
\end{figure}

\begin{lemma}\label{lem:VW}
	There is $\rho$ such that all $p\in\mathbb{R}^2$ with $\|p\| > \rho$  are closer to $W$ than to $\Sites\setminus W$. 
\end{lemma}
\begin{proof}

Let $W\subseteq S$ be the sites that lie in the boundary of the convex hull $\CHS$. 
We prove that there is a sufficiently large value $\rho$ such that all $p\in\mathbb{R}^2$ with $\|p\|>\rho$
	are strictly closer (in the sense of $D$) to $W$ than to the remaining sites $S\setminus W$. 

Pick any pair of consecutive sites $w_i,w_{i\oplus 1}$ along the boundary of $\CHS$, and any third site $v$ from $S\setminus W$. 
We first show that there are values $\rho_{w_i,w_{i\oplus 1},v}$ such that all $p\in\mathbb{R}^2$ with $\|p\|>\rho_{w_i,w_{i\oplus 1},v}$
	are strictly closer to $w_i,w_{i\oplus 1}$ than to $v$. 
By letting
	\[ \rho \equiv \displaystyle{\max_{\substack{w_i,w_{i\oplus 1} \in W \\ v \in S\setminus W}} 
		\rho_{w_i,w_{i\oplus 1},v}}, \]
	lemma 6.5 of~\ref{ass:BAA}. 

For a triple $w_i,w_{i\oplus 1},v$ with $w_i,w_{i\oplus 1}\in W$ consecutive vertices in $\partial\CHS$, and $v\in S\setminus W$, 
	we consider the supporting line $l_{i,{i\oplus 1}}$ of $w_i,w_{i\oplus 1}$, 
		which divides space into two half-spaces $H^{+}_{i,{i\oplus 1}}$ and $H^{-}_{i,{i\oplus 1}}$, 
	where we pick $H^{-}_{i,{i\oplus 1}}$ so as to contain $v$. 
Note that $v\notin l_{i,{i\oplus 1}}$ or else it would be $v\in W$. 

First, let $\rho_{w_i,w_{i\oplus 1},v}$ be large enough so that all $p\in\mathbb{R}^2$ with $\|p\| > \rho$ are outside the convex hull $\CHS$. 
The divide the proof in three cases, 
depending on whether 
	point $p$ belongs to $l_{i,{i\oplus 1}}$, $H^{-}_{i,{i\oplus 1}}$, or $H^{+}_{i,{i\oplus 1}}$, respectively. 
Each case will result in a different constraint $\|p\| > \rho^{o}_{i,i\oplus 1, v}$, $\|p\| > \rho^{-}_{i,i\oplus 1, v}$, $\|p\| > \rho^{+}_{i,i\oplus 1, v}$, 
	\[ \rho_{i,i\oplus 1, v} \equiv \max\{\rho^{o}_{i,i\oplus 1, v},  \rho^{-}_{i,i\oplus 1, v},  \rho^{+}_{i,i\oplus 1, v} \}. \]

\vspace*{0.1in}\noindent{\bf [Case $p\in l_{i,{i\oplus 1}}$]}.
Since not all sites are colinear, then it is $|W| > 2$, and therefore the segment $\overline{w_{i\oplus 1} w_{i\oplus 2}}$ 
	is different from $\overline{w_i,w_{i\oplus 1}}$. 
If $p\in l_{i,{i\oplus 1}}$, 
	we can consider the ``next" segment $\overline{w_{i\oplus 1} w_{i\oplus 2}}$, 
	for which, since  $p$ is outside $\CHS$, it must hold 
	$p\in H^{+}_{w_{i\oplus 1} w_{i\oplus 2}, v}$, 
	and therefore we can simply let $\rho^{o}_{i,i\oplus 1, v} \equiv \max_{v\in S\setminus W} \rho^{+}_{i\oplus 1, i\oplus 2, v}$. 
Note that there is no circular dependency in this definition, since we can resolve it by simply 
	letting 
	\[ \rho^{o}_{i,i\oplus 1, v} \equiv \displaystyle{\max_{w_i\in W |||}  \rho^{+}_{i, i\oplus 1, v}}. \]


\vspace*{0.1in}\noindent{\bf [Case $p\in H^{-}_{i,{i\oplus 1}}$]}.
By assumption~\ref{ass:BAA}, there is $ \rho^{-}_{w_i,w_{i\oplus 1},v} > 0$ 
	such that all $p\in\mathbb{R}^2$ with $\|p\| > \rho^{-}_{w_i,w_{i\oplus 1},v}$ 
	are closer to $v$ than to $w_i$ or $w_{i\oplus 1}$.

\vspace*{0.1in}\noindent{\bf [Case $p\in H^{+}_{i,{i\oplus 1}}$]}.
Let $p\in H^{+}_{i,{i\oplus 1}}$, and consider the ball $B$ of the first kind centered at $p$ with radius $\D{v}{p}$. 
Since $B$ is convex, we can find a line $l_v$ passing through $v$ that ``separates" $B$, 
	that is, $B$ lies in the half-space $H^{+}_{l_v}$ associated to $l_v$. 
It follows that $v$ is the closest point to $p$ in the line $l_v$. 
Note that, because $v\in H^{-}_{i,{i\oplus 1}}$, and $l_v$ passes through $v$, 
	then it must be either $w_i \in H^{+}_{l_v}$ or $w_{i\oplus 1} \in H^{+}_{l_v}$ 
	(otherwise, if $w_i,w_{i\oplus 1}\in H^{-}_{l_v}$ then it would be $v\in H^{+}_{i,{i\oplus 1}}$, a contradiction). 
Without loss of generality, let $w_i$ be in $H^{+}_{l_v}$. 

Pick two point $a,b$ along $l_v$ such that $v$ is between $a$ and $b$. 
We are now ready to apply assumption~\ref{ass:BAA}, using the substitution
	$p=a$, $q=b$, $r=w_i$, and $c=p$, 
	from which it follows that there is a sufficiently large $\rho_{w_i,w{i\oplus 1},v} > 0$ 
		such that if $p\in H^{+}_{i,{i\oplus 1}}$ and $\|p\| > \rho_{w_i,w{i\oplus 1},v}$, 
		then $p$ is closer (in the sense of $D$) to $w_i$ than to $v$. 

\end{proof}

\noindent{\bf lemma~\ref{lem:contrad}}
\emph{
There is $\rho > 0$ such that, for any segment $(w_i,w_j)\in\mathcal{B}$ with supporting line $L_{ij}$, 
	 every $p\in H^{-}_{ij}$ with $\|p\| > \rho$ whose closest point in $L_{ij}$ belongs to $\overline{w_i w_j}$ is 
	closer to a site in $\Sites\setminus\{w_i,w_j\}$ than to $L_{ij}$. }
\begin{proof}
For each edge $(w_i,w_j)\in\mathcal{B}$ with supporting line $L_{ij}$, 
	pick a site $v\in\Sites\setminus\{w_i,w_j\}$ that isn't in $L_{ij}$ (which always exists since not all sites are colinear). 
By assumption~\ref{ass:BAA}, there is a sufficiently large $\rho_{ij}$ such that every point $p\in H^{-}_{ij}$ 
	whose closest point $m_p$ to $L_{ij}$ satisfies $m_p\in\overline{w_i w_j}$ is closer to $v$ than to $m_p$. 
Since $p\in H^{-}_{ij}$, then $p$ is closer to $m_p$ than to either $w_i,w_j$, and 
thus $p$ is closer to $v$ than to either $w_i,w_j$. 

Letting $\rho$ be the maximum of $\rho_{ij}$ over all edges  $(w_i,w_j)\in\mathcal{B}$ completes the proof. 
\end{proof}

\noindent{\bf lemma~\ref{lem:Sn}}
\emph{
	Every continuous function $F:\mathbb{S}^n\rightarrow\mathbb{S}^n$ that is not onto has a fixed point. 
}
\begin{proof}
	Assume $F$ misses $p\in\mathbb{S}^n$, and let
$\gamma:\mathbb{S}^n\setminus\{p\}\rightarrow D^n$ be a diffeomorphism
between the punctured sphere and the open unit disk. 
Since $\gamma\circ F$ is continuous and $\mathbb{S}^n$ is compact,
then the set $C = (\gamma\circ F) (\mathbb{S}^n)\subset D^n$ is compact.

The function $g:C\rightarrow C$ with $g = \gamma \circ F\circ\gamma^{-1}$ 
is continuous and therefore, by Brouwer's fixed point theorem~\cite{Milnor}, has a fixed point $x\in C$. 
The fact that $(\gamma \circ F\circ \gamma^{-1}) (x) = x$ implies $F(\gamma^{-1}(x)) = \gamma^{-1}(x)$ 
and thus $\gamma^{-1}(x)\in\mathbb{S}^n$ is a fixed point of $F$. 
\end{proof}

\noindent{\bf lemma~\ref{lem:hard}}
($B\supseteq\mathcal{B}$)
\emph{
To every segment  in the boundary of $\CHS$ corresponds a boundary edge of $G$.  
}
\begin{proof}  
Let $(w_i,w_j)\in\mathcal{B}$ be a segment in the boundary of $\CHS$, 
	as shown in figures~\ref{fig:pinu} and~\ref{fig:pi}. 
Pick a sufficiently large $\rho > \max_{v\in \Sites}\|v\|$ such that every $p$ with
$\|p\| > \rho$ is outside $\CHS$ and such that 
	Lemmas~\ref{lem:VW} and~\ref{lem:contrad}
	hold. 
For any $\sigma > \rho$, if $A:C(\sigma)\rightarrow C(\sigma)$ is the antipodal map $A(p) = -p$, 
then, by continuity of $\pi$ (property~\ref{prop:pi}.i) and 
by the continuity of $\nu_\sigma$, the function $A\circ\nu_\sigma\circ\pi:C(\sigma)\rightarrow C(\sigma)$ is
continuous. 
By 
	lemma~\ref{lem:VW} and property~\ref{prop:pi}.ii, 
	if for some $p_{ij}\in C(\sigma)$ it is 
$\pi(p_{ij})= (w_i+w_j)/2$ with
$(w_i,w_j)\in\mathcal{B}$, then $p_{ij}$ is (strictly) closest to
$w_i,w_j$, and therefore belongs to the primal edge $\Vor_{ij}$, which implies
that $(w_i,w_j)\in B$. 

Showing that $\mathcal{B}\subseteq B$ now reduces to showing that for all 
$(w_i,w_j)\in\mathcal{B}$, for all $\sigma > \rho$, there is $p_{ij}\in
C(\sigma)$ such that
$\pi(p_{ij})= (w_i+w_j)/2$. 

Assume otherwise. 
The function 
$A\circ\nu_\sigma\circ\pi$ 
is not onto and therefore, 
by lemma~\ref{lem:Sn} (and using the fact that $C(\sigma)$ is isomorphic to $\mathcal{S}^1$), 
it must have a fixed point $q$. 

Since $(A\circ\nu_\sigma\circ\pi)(q) = q$ then $(\nu_\sigma\circ\pi)(q) = -q$. 
Since $\pi(q)$ is the closest point to $q$ in $\partial\CHS$, 
	there is a segment $(w_k,w_l)\in\mathcal{B}$ such that
$\pi(q)\in\overline{w_k w_l}$. Consider two open half spaces
$H^{+}_{kl}$ and $H^{-}_{kl}$ on either side of the supporting line of
$w_k,w_l$. Since not all sites are colinear, we can choose these half spaces so that 
$H^{+}_{kl}\cap \Sites=\phi$ and $H^{-}_{kl}\cap \Sites \neq \phi$. 
By the definition of $\nu_\sigma$, and recalling that the chosen origin of
$\mathbb{R}^2$ is in the interior $\mathbf{int\ }{\CHS}$ of the convex hull, 
it is $\nu_\sigma(\pi(q))\in H^{+}_{kl}$, and $q=-\nu_\sigma(\pi(q))\in H^{-}_{kl}$. 
To see this note that the outward-facing normal $n(\pi(q))$ is defined so that
$\pi(q)+n(\pi(q))\in H^{+}_{kl}$ 
and so $\nu_\sigma(\pi(q)) = \sigma\cdot n(\pi(q)) / \|n(\pi(q))\| \in  H^{+}_{kl}$.
On the other hand, since the origin is in
$\mathbf{int\ }{\CHS}$, the fact that $\nu_\sigma(\pi(q))\in H^{+}_{kl}$ 
implies $q=-\nu_\sigma(\pi(q))\in  H^{-}_{kl}$. 

Since $\rho$ was chosen sufficiently large for lemma~\ref{lem:contrad} to
hold, and $q\in  H^{-}_{kl}$, $q$ is closer to some site $v\in
\Sites\setminus\{w_k,w_l\}$ than to $\overline{w_k w_l}$. 
Since $v\in\CHS$, this contradicts the fact
that $\pi(q)\in\overline{w_k w_l}$ is the closest point to $q$ in
$\CHS$. 
\end{proof}

\section*{Appendix C: dual triangulation (interior)}\label{app:interior}

\noindent{\bf lemma~\ref{lem:non-negative}}
\emph{
	Given a non-vanishing one-form $\xi^n$, the sum of indices of interior vertices ($\Sites\setminus W$) of $G$ is non-negative. 
}
\begin{proof}
Given non-vanishing $\xi^n$, 
the index of a face $f$ is $\mathbf{ind}_{_{\xi^n}}(f) \le 0$. 
To see this, assume otherwise: a face with vertices $v_1,\dots,v_m$ around it, and index one satisfies, by the
definition of index and of $\xi^n$, 
$n^t v_1 < \dots < n^t v_m < n^t v_1$ 
(or $n^t v_1 > \dots > n^t v_m > n^t v_1$), a
contradiction. 

Because, by corollary~\ref{cor:boundary}, the boundary edges of $G$ form a convex, 
	simple polygonal chain, then, given any non-vanishing $\xi^n$, all the boundary vertices have
index zero, except for the ``topmost" ($\underset{v\in \Sites}{\mathbf{argmax\ }} \xi^n(v)$) 
and ``bottommost" ($\underset{v\in \Sites}{\mathbf{argmin\ }} \xi^n(v)$) vertices, which have
index one (note that the topmost and bottommost vertices are unique because $n$ is chosen not to be orthogonal to any edge in the triangulation). 

Since face indices are non-positive, and the sum of indices of
boundary vertices is two then, by lemma~\ref{lem:ph},
the sum of indices of {interior} vertices must be non-negative. 
\end{proof}

%

\noindent{\bf lemma~\ref{lem:index-1}}
\emph{
If $G$ has an edge fold-over, then there is $n\in\mathbb{S}^1$ and non-vanishing one-form $\xi^n$ such
that $\mathbf{ind}_{_{\xi^n}}(v) < 0$ for some interior vertex $v\in \Sites\setminus W$. 
}
\begin{proof}
If edge fold-over $e=(v,w)$ is a non-boundary edge, then at least one
of its incident vertices, say $v$ is an interior vertex $v\in
\Sites\setminus W$. 

Consider the two faces $f_1,f_2$ incident to $e$, which, by definition of
edge fold-over, are on the same side of its
supporting line, and the two edges $e_1,e_2$ in $f_1,f_2$ respectively,
incident to $v$. 
Taking the half-line $h$ from $v$ towards $w$ as reference, consider 
the (open) set $L_i\subset\mathbb{S}^1$ of directions ranging from $h$ to
$e_i$. 
The set $L=L_1\cap L_2$ is not empty since, by proposition~\ref{prop:ECB}, 
$f_1,f_2$ are not degenerate, and therefore neither $e_1,e_2$ are parallel
to $h$. $L$ is also uncountable, since it is a range of the form
\[ L = \{n\in\mathbb{S}^1 : 
			{n_{\perp}}^t {h} < 0 \wedge
			{n_{\perp}}^t {e}_1 > 0 \wedge
			{n_{\perp}}^t {e}_2 > 0\}\]
where ${h},{e}_i$ are the direction vectors of $h,e_i$, and
${n_\perp}$ is one of the two orthogonal directions to $n$, 
chosen to fit the definition. 

Because $L$ is not empty, and is uncountable, and because the set of edges $E$ is finite,
then there is always some direction $n\in L$ that is not orthogonal to any edge in
$E$. 
Pick any such $n$. 
We prove that the non-vanishing one-form $\xi^n$ is such that
$\mathbf{ind}_{_{\xi^n}}(v)<0$. 

The (cyclic) sequence of oriented half-edges {around} $v$ is, without loss of generality,
$\mathcal{S}=\left[(v,v_1);(v,w);(v,v_2);\dots\right]$, and therefore the values of the
one-form around $v$ are $[\xi^n(v_1)-\xi^n(v)$, $\xi^n(w)-\xi^n(v)$,
$\xi^n(v_2)-\xi^n(v)$, $\dots]$. 
By the definition of $n$, it is $\xi^n(v_1)<\xi^n(v)$, $\xi^n(w)>\xi^n(v)$,
and $\xi^n(v_2)<\xi^n(v)$, and therefore 
the number of sign changes in 
the subsequence
$\mathcal{S}'=[(v,v_1);(v,w);(v,v_2)]$ is four. 
Since the number of sign changes in the full sequence $\mathcal{S}$ cannot
be less
than that of its subsequence $\mathcal{S}'$, 
it is $\mathbf{sc}_{_{\xi^n}}(v)>4$ and therefore $\mathbf{ind}_{_{\xi^n}}(v)=1 - \mathbf{sc}_{_{\xi^n}}(v)/2 <0$. 
\end{proof}

\noindent{\bf lemma~\ref{lem:index1}}
\emph{
Given $n\in\mathbb{S}^1$ and non-vanishing one-form $\xi^n$, if $G$
has an interior vertex $v\in \Sites\setminus W$ with index
$\mathbf{ind}_{_{\xi^n}}(v)=1$, then there is a face $f$ of
$G$ that does not satisfy the empty circum-ball
property (proposition~\ref{prop:ECB}). 
%
}
\begin{proof}
We must prove that there is a face $f$ all of whose circumscribing balls
contain some vertex in its interior. 

Consider the vertex $v\in \Sites\setminus W$ with
$\mathbf{ind}_{_{\xi^n}}(v)=1-\mathbf{sc}_{_{\xi^n}}(v)=1$, and thus
with $\mathbf{sc}_{_{\xi^n}}(v)=0$. 
If $\left[u_1,u_2,\dots,u_m\right]$ is the cyclic sequence of vertices neighboring
$v$, then 
$\mathbf{sc}_{_{\xi^n}}(v)=0$ implies either $\xi^n(u_i)>\xi^n(v)$, $i=1,\dots,m$, or $\xi^n(u_i)<\xi^n(v)$, $i=1,\dots,m$. 
Assume the former w.l.o.g. 
The line $l=\{x\in\mathbb{R}^2 : n^t x = n^t v\}$, passing through
$v$, strictly separates $v$ from the convex hull of its neighbors. 

Consider the mesh $G'$, with the same structure as $G$ but in which all the incident faces to $v$
are eliminated. 
We show that, in $G'$, the face count of $v$ (the number of faces in which $v$ lies) is at
least one. 
Since $l$ separates $v$ from its neighbors, it also separates all the faces
incident to $v$ from $v$ (except for $v$ itself, which lies on $l$). 
Pick any direction $d\in\mathcal{S}^1$ with $n^t{d} < 0$. 
The half-line $h$ starting at $v$ with direction $d$ 
does not intersect any face in $G$ that is incident to $v$. 
Since there is only a finite number of edges and vertices, it is always
possible to choose $h$ not to contain any vertex other than $v$, 
and not to be parallel to any edge. 
Since $\CHS$ is bounded and $h$ isn't, there is
some point $x\in h$ outside $\CHS$, whose face count must be zero. 
Moving from $x$ toward $v$, $h$ crosses $\partial\CHS$ only once 
(since $\CHS$ is convex), incrementing the face count to one. 
Because every interior edge is incident to exactly two faces, 
every subsequent edge cross (which is transversal because $h$ is not
parallel to any edge) modifies the face count by either zero, two, or
minus two. Since the face count cannot be negative, and it is one at 
$h\cap \partial\CHS$, then it must be at least one at $v$. 
Since $G'$ does not contain any face incident to $v$, 
this implies that there is
some face $f$ not incident to $v$ such that $v\in f$. 

We prove that the face $f$ above cannot satisfy the ECB property. 
Since $v$ is in $f$ but is not incident to it, and $f$ is convex
then, by Carath\'eodory's theorem~\cite{matousek2002lectures}, $v$ can be written as a 
convex combination $v=\lambda_1 u_1+\lambda_2 u_2+\lambda_3 u_3$,
$\sum_{i=1}^3\lambda_i=1$, $\lambda_i\in(0,1)$ of vertices $u_1,u_2,u_3$
incident to $f$ 
(note that this is slightly more general than required since we have 
already made sure in the beginning of section~\ref{sec:dual} that $f$ is a triangle). 
Given a ball circumscribing the vertices incident to $f$, 
because it is strictly convex, and $u_1,u_2,u_3$ lie in its boundary, 
then any convex combination of them with $\lambda_i\in(0,1)$ must be in the
interior of the circumscribing ball, and therefore $f$ does not satisfy the ECB property. 
\end{proof}

\noindent{\bf lemma~\ref{lem:ef}}
\emph{
$G$ 
has no edge fold-overs. }
\begin{proof}
Assume $G$ has an edge fold-over. 
By lemma~\ref{lem:index-1}, there is a non-vanishing one-form
$\xi^n$ such that some interior vertex $v\in \Sites\setminus W$ has $\mathbf{ind}_{_{\xi^n}}(v)<0$. 
Since, by lemma~\ref{lem:non-negative}, the sum of indices of interior vertices is
non-negative,
then there must be
at least one interior vertex $u\in \Sites\setminus W$ with positive index
$\mathbf{ind}_{_{\xi^n}}(u)=1$. 
In that case, by lemma~\ref{lem:index1}, there is a face of $G$ that does not satisfy the
ECB property, raising a contradiction. 
Therefore $G$ has no edge fold-overs. 
\end{proof}

\noindent{\bf lemma~\ref{lem:interior}}
\emph{
If its (topological) boundary is simple and closed, 
	then the straight-line dual of an orphan-free Voronoi diagram, 
	with vertices at the sites, 
	is an embedded triangulation. 
}
\begin{proof}

Given a point $p\in\mathbf{int\ }{\CHS}$ in
the interior of the convex hull of $\Sites$, we show that its \emph{face count}
(the number of straight-edge faces that contain it) is one. 
Consider a line $l$ passing through $x$ that does not pass through any vertex of
$G$, and is not parallel to any (straight) edge. 
It is always possible to find such a line since the set of vertices and edges is finite. 
Because the line is unbounded and $\CHS$ is bounded, there is a
point $x\in l$ that is outside $\CHS$. At this point clearly the
face count is zero. 
Moving from $x$ toward $p$, $l$ crosses the boundary of $\CHS$
(and therefore, by corollary~\ref{cor:boundary}, the boundary of $G$) 
only once, since it is a simple convex polygonal chain, incrementing the face
count by one. At every edge crossing (which is transversal by the choice of
line), the face count remains one since, by lemma~\ref{lem:ef} there are no
edge fold-overs, and thus every non-boundary edge is incident to two faces
that lie on either side of its supporting plane. Therefore the face count at
$p$ must be one. 

Since every point inside $\CHS$ is covered once by faces in
$G$,
and the boundaries of $G$ and $\CHS$ coincide, then
$G$ is a single-cover
of $\CHS$. 
Because two straight edges that cross at a non-vertex always generate points with
face count higher than one, then the edges of $G$ can only meet at vertices, and
therefore $G$ is embedded. 
\end{proof}

\end{document}